\documentclass[a4paper, 11pt]{article}
\usepackage[left=1.15in,right=1.15in,top=1.15in,
bottom=1.25in,dvips]{geometry}
\usepackage[utf8]{inputenc}
\usepackage{lmodern}
\usepackage{indentfirst}
\usepackage[affil-it]{authblk}
\usepackage{rotating}
\usepackage{tocloft}
\usepackage{titlesec}
 \titleformat{\subparagraph}[runin]{\normalfont}{\thesubparagraph}{0pt}{\underline}[.]
 \titleformat{\paragraph}[hang]{\normalfont\bfseries}{\theparagraph}{0pt}{}
\usepackage{titletoc}

\usepackage[textsize=scriptsize]{todonotes} 
\makeatletter
\renewcommand{\todo}[2][]{\tikzexternaldisable\@todo[#1]{#2}\tikzexternalenable}
\makeatother

\usepackage{standalone}
\usepackage{subfiles}

\usepackage[hidelinks]{hyperref}
\usepackage{enumitem}
\setlist[enumerate]{itemsep=2.0pt plus 1.0 pt minus 0.5pt, topsep=4.0pt plus 2.0 pt minus 1.0pt}
\setlist[itemize]{itemsep=2.0pt plus 1.0 pt minus 0.5pt, topsep=4.0pt plus 2.0 pt minus 1.0pt}
\usepackage{tikz}
\usetikzlibrary{external}
\usepackage{calc}
\usepackage{graphicx}
\usepackage{subcaption}
\usepackage{epsfig}
\usepackage{color, soul}

\usepackage{marginnote}

\usepackage{color,soulutf8}

\usepackage{amsmath}
\usepackage{amsthm, slashed}
\usepackage[capitalize]{cleveref}
\usepackage{amssymb}
\usepackage{array}
\usepackage{bbm}
\usepackage{bm}
\usepackage{marvosym}
\usepackage{mathtools}
\usepackage{mathrsfs}
\usepackage{upgreek}
\usepackage{accents}
\usepackage{textcomp}
\usepackage{slashed} 
\usepackage{harmony} 
\usepackage{wasysym} 
\wasyfamily
\DeclareFontShape{U}{wasy}{b}{n}{ <-10> ssub * wasy/m/n
 <10> <10.95> <12> <14.4> <17.28> <20.74> <24.88>wasyb10 }{}
\DeclareMathAlphabet\mathbfcal{OMS}{cmsy}{b}{n}
\newcommand\numberthis{\addtocounter{equation}{1}\tag{\theequation}}
\renewcommand{\Re}{\operatorname{Re}}
\renewcommand{\Im}{\operatorname{Im}}

\allowdisplaybreaks

\numberwithin{equation}{section}

\newtheorem{definition}{Definition}[section]
\newtheorem{theorem}{Theorem}[section]
\newtheorem*{theorem*}{Theorem}
 
\newtheorem*{conjecture*}{Conjecture}
\newtheorem{corollary}{Corollary}[section]
\newtheorem*{corollary*}{Corollary}

\newtheorem{proposition}{Proposition}[section]
\newtheorem{lemma}[proposition]{Lemma}
\newtheorem{remark}[proposition]{Remark}

\usepackage{csquotes}
\usepackage[backend=biber,
style=alphabetic, date=year, doi=false,isbn=false,maxnames=10]{biblatex}
\renewbibmacro{in:}{}

\AtEveryBibitem{\clearfield{pages}} 
\DeclareBibliographyCategory{needsurl}
\newcommand{\entryneedsurl}[1]{\addtocategory{needsurl}{#1}}
\renewbibmacro*{url+urldate}{%
  \ifcategory{needsurl}
    {\printfield{url}%
     \iffieldundef{urlyear}
       {}
       {\setunit*{\addspace}%
        \printurldate}}
    {}}
\entryneedsurl{Olver2018}
\entryneedsurl{Tao2018}

\newcommand{\mb}[1]{\mathbb{#1}}

\newcommand{\mc}[1]{\mathcal{#1}}
\newcommand{\mr}[1]{\mathrm{#1}}

\newcommand{\p}{\partial}
\newcommand{\lp}{\left}
\newcommand{\rp}{\right}
\newcommand\supp{\mathrm{supp}}

\newcommand{\swei}[2]{#1^{[{#2}]}}

\title{\textbf{Mode stability for the Teukolsky equation \\ on extremal and subextremal Kerr spacetimes}}

\date{\today}

\author{\Large Rita Teixeira da Costa}

\affil{\small University of Cambridge, Department of Pure Mathematics and Mathematical
Statistics, Wilberforce~Road,~Cambridge~CB3~0WA,~United~Kingdom }

\begin{document}

\maketitle

\vspace{-1.2\baselineskip}

\begin{abstract} 
We prove that there are no exponentially growing modes nor modes on the real axis for the Teukolsky equation on  Kerr black hole spacetimes, both in the extremal and subextremal case.  We also give a quantitative refinement of mode stability. As an immediate application, we show that the transmission and reflection coefficients of the scattering problem are bounded, independently of the specific angular momentum of the black hole, in any compact set of real frequencies excluding zero frequency and the superradiant threshold.

While in the subextremal setting these results were known previously, the extremal case is more involved and has remained an open problem. Ours are the first results outside axisymmetry and could serve as a preliminary step towards understanding  boundedness, scattering and decay properties of general solutions to the Teukolsky equation on extremal Kerr black holes.
\end{abstract}

\vspace{\baselineskip}

\tableofcontents

\section{Introduction}

Einstein's equation in vacuum,
\begin{equation}
\mr{Ric}(g)=0\,, \label{eq:einstein}
\end{equation}
and in particular its black hole solutions, have been the object of intense study by the general relativity community. Among such solutions, the Kerr family of spacetimes \cite{Kerr1963} is especially prominent. A Kerr black hole is characterized by a mass $M>0$ and specific angular momentum $|a|\leq M$; we say it is subextremal if $|a|<M$ and extremal if $|a|=M$. 
A fundamental problem concerning these spacetimes is whether they are nonlinearly stable as solutions of the Einstein equation (\ref{eq:einstein}); we refer the reader to \cite{Dafermos2011} for an introduction to the question of nonlinear stability. Another interesting problem is to understand scattering processes on Kerr black holes spacetimes. This paper will concern the Teukolsky equation, which plays a key role in both these problems.

In order to solve the nonlinear stability problem, the first step is to show linear stability, i.e.\ boundedness and decay statements for the linearization of the Einstein equation \eqref{eq:einstein} around the Kerr solutions. In this context, the  Teukolsky equation \cite{Teukolsky1973}
\begin{equation}
\begin{gathered}
\Box_g\upalpha^{[s]} +\frac{2s}{\rho^2}(r-M)\p_r\upalpha^{[s]} +\frac{2s}{\rho^2}\lp[\frac{a(r-M)}{\Delta}+i\frac{\cos\theta}{\sin^2\theta}\rp]\p_\phi\upalpha^{[s]} 
\\
+\frac{2s}{\rho^2}\lp(\frac{M(r^2-a^2)}{\Delta}-r -ia\cos\theta\rp)\p_t\upalpha^{[s]} +\frac{1}{\rho^2}\lp(s-s^2\cot^2\theta\rp)\upalpha^{[s]}  =0
\end{gathered}\label{eq:Teukolsky-equation}
\end{equation}
for $s=\pm 2$ plays a central role, as it describes the dynamics of the extremal curvature components of the metric in the Newman–Penrose formalism \cite{Newman1962}. (Here, we have used  Boyer--Lindquist coordinates $(t,r,\theta,\phi)$ and we refer the reader to \cref{sec:kerr-spacetime} for a definition of the weights $\Delta$ and $\rho$.) This equation can also be considered for arbitrary values of $s\in\frac12\mathbb{Z}$. For $s=0$, the Teukolsky equation \eqref{eq:Teukolsky-equation} reduces to the wave equation. For $s=\pm 1$, it describes the evolution of the extreme components of the Maxwell equations in a null frame (see \cite{Chandrasekhar} for an extended discussion).

As shown by Teukolsky, equation \eqref{eq:Teukolsky-equation} is separable; thus, it is natural to begin its study by focusing on fixed frequency solutions, called \textit{mode solutions}, given by
\begin{align}
\swei{\alpha}{s}(t,r,\theta,\phi) = e^{-i\omega t}e^{im\phi}S(\theta)R(r)\,, \quad \omega\in \mathbb{C},\, m\in\mathbb{Z}\,,  \label{eq:mode-sol-intro}
\end{align}
where $S$ and $R$ satisfy certain ODEs with boundary conditions that ensure $\swei{\alpha}{s}$ has finite energy along suitable spacelike hypersurfaces (see \cref{def:mode-solution}). 

If one expects boundedness and decay for \eqref{eq:Teukolsky-equation} to hold, the most basic statement one can hope for is that there are no exponentially growing mode solutions, expressible as \eqref{eq:mode-sol-intro}, with $\Im \omega >0$, or solutions with $\omega\in\mathbb{R}\backslash\{0\}$ to \eqref{eq:Teukolsky-equation} that would threaten the validity of a boundedness property. Such a statement is trivially true for Schwarzschild black holes ($a=0$) when $s=0$, by the conservation of the energy associated with the Killing field $\p_t$. However, when $a\neq 0$, this proof no longer holds, due the existence of an ergoregion, where $\p_t$ becomes spacelike, and the phenomenon of \textit{superradiance} that can thus occur for frequencies satisfying
\begin{equation}
\omega(\omega-m\upomega_+)\leq 0\,,\quad \omega_+:=\frac{a}{2M(M+\sqrt{M^2-a^2})}\label{eq:intro-superrad}\,.
\end{equation}
While in the general $s\in\frac12\mathbb{Z}$ case, the $\p_t$ energy is not conserved, the special algebraic properties of \eqref{eq:Teukolsky-equation} yield another identity for fixed frequency solutions which easily gives mode stability statements for Schwarzschild black holes ($a=0$), but not for general rotating spacetimes. At least for $s\leq 2$, the proof once again only fails for the superradiant frequencies in \eqref{eq:intro-superrad}.

Despite these difficulties, Whiting \cite{Whiting1989} was able to show mode stability for the Teukolsky equation of any spin $s\in\frac12\mathbb{Z}$ on subextremal Kerr with $\omega$ in the upper half-plane, thanks to a clever transformation of the mode solutions into solutions of another wave equation for which there is no ergoregion. Extending this result to the real axis was an open problem for almost 25 years, and it was finally put to rest by  Shlapentokh-Rothman \cite{Shlapentokh-Rothman2015}, who showed in addition that the radial integral transformation was enough to obtain Whiting's result, and generalized the result for real $\omega\neq 0$. Shlapentokh-Rothman's result was obtained for $s=0$; more recently, \cite{Andersson2017} extended \cite{Shlapentokh-Rothman2015} for all spins. The results in this paragraph are summarized in
\begin{theorem}[Mode stability for subextremal Kerr] \label{thm:mode-stability-subextremal} There are no non-trivial mode solutions to the homogeneous Teukolsky equation \eqref{eq:Teukolsky-equation} of any spin $s\in\frac12\mathbb{Z}$ on \textbf{subextremal} Kerr backgrounds with $\omega$ such that $\Im \omega\geq 0$ and $\omega\neq 0$.
\end{theorem}

For subextremal Kerr, boundedness and decay for $s=0$ was finally settled in \cite{Dafermos2016b}, with the proof of Theorem~\ref{thm:mode-stability-subextremal} on the real axis, in \cite{Shlapentokh-Rothman2015}, playing a key role, via a quantitative statement which we will soon describe. Over the last years, some boundedness and decay statements for higher spins have been obtained. For $s=\pm 1$, Kerr backgrounds with $a=0$ and $|a|\ll M$ were considered in \cite{Pasqualotto2016} and \cite{Ma2017}, respectively, previously addressed in \cite{Andersson2015a} using the Maxwell equations instead of the Teukolsky formalism. Moreover, for $s=\pm 2$, we recall that full linear stability has been shown for $a=0$ in \cite{Dafermos2016a} and  boundedness and decay for \eqref{eq:Teukolsky-equation} in the $|a|\ll M$ case was shown in \cite{Dafermos2017,Ma2017a}. 

\medskip

To complete the above results, one would like to understand the limit $|a|\to M$, which is moreover  of great relevance for astrophysics \cite{Volonteri2005}. However, our understanding is significantly more lacking for extremal Kerr black holes. This is due to the plethora of new difficulties that the degenerate horizon poses, even at the level of $s=0$ (see also \cite{Aretakis2018a}):
\begin{enumerate}[label=(\alph*)]
\item \textit{Degeneracy of the redshift and Aretakis instability}. For subextremal Kerr, the event horizon has a positive surface gravity,  which allows one to infer strong decay properties of solutions to the wave equation along the horizon. As the surface gravity vanishes in the extremal case, this stabilizing mechanism is absent. Indeed, in \cite{Aretakis2012}, it was observed that transverse derivatives of axisymmetric waves ($m=0$) grow polynomially along the horizon, a phenomenon which is known as \textit{Aretakis instability}; the result was later generalized for higher spins in \cite{Lucietti2012}. Unlike the remaing phenomena we will discuss, the Aretakis instability is present even for nonrotating extremal spacetimes, such as extremal Reissner-Nordström black holes (spherically symmetric charged black holes); however, in that case, the remaining difficulties present on the list do not occur and it is known that general solutions of \eqref{eq:Teukolsky-equation} with $s=0$ are bounded and decay \cite{Aretakis2011,Aretakis2011b}, and their precise late-time asymptotics have been understood in \cite{Angelopoulos2018a}.

More recently, in the heuristic work \cite{Casals2016}, the authors suggest that solutions with fixed azimuthal mode for $s=0$ also display the effect, with a faster growth for transverse derivatives than the $m=0$ case\footnote{We remark that the conclusions of \cite{Casals2016} can only hold assuming a mode stability result for extremal Kerr of the type of Theorem~\ref{thm:mode-stability-full}, to follow; this  enables the authors to restrict their investigation into the regularity properties of a Green's function for \eqref{eq:Teukolsky-equation} with $s=0$ to frequencies close to the superradiant theshold, $\omega\approx m\upomega_+$, rather than the wider range of real frequencies where, in the absence of a mode stability statement, the Green's function could have branch points or poles.}. 

\item \textit{Trapping meets superradiance}. In \cite{Dafermos2016b}, a fundamental insight for establishing boundedness and decay for  the Teukolsky equation with $s=0$ on subextremal Kerr was the fact that frequencies which are superradiant cannot correspond to \textit{trapped null geodesics}, high-energy geodesics which neither intersect the event horizon nor escape to null infinity. However, when $|a|=M$, superradiant frequencies, which exist for $m\neq 0$ by \eqref{eq:intro-superrad}, can be marginally trapped (see the introduction in \cite{Aretakis2012a}).

\item \textit{No available mode stability statement}. In \cite{Dafermos2016b}, mode stability on the real axis played a very important role. While the obstacle to the proof of mode stability is the superradiance present for both rotating subextremal and extremal Kerr backgrounds (see \eqref{eq:intro-superrad}), the degenerate horizon in extremal Kerr drastically changes the character of the Teukolsky equation with respect to the subextremal case. This suggests that a new Whiting-type transformation is necessary.
\end{enumerate}

In view of difficulties we have discussed, it is not surprising that, in the extremal case, only axisymmetric ($m=0$) solutions to the Teukolsky equation with $s=0$, which do not exhibit superradiance, have been thoroughly studied in \cite{Aretakis2012}. The present paper represents a first step to address general (i.e.\ not necessarily axisymmetric) solutions of any spin $s\in\frac12\mathbb{Z}$ of \eqref{eq:Teukolsky-equation} by providing a definitive resolution to (c):

\begin{theorem}[Mode stability for extremal Kerr] There are no non-trivial mode solutions to the Teukolsky equation \eqref{eq:Teukolsky-equation} of any spin $s\in\frac12\mathbb{Z}$ on \textbf{extremal} Kerr backgrounds with $\omega$ and $m$ such that either $\Im \omega> 0$ or $\omega\in\mathbb{R}\backslash\{0,m\upomega_+\}$. \label{thm:mode-stability-full}
\end{theorem}

It turns out that for  $\Im\omega>0$, \cref{thm:mode-stability-full} can be inferred from \cref{thm:mode-stability-subextremal} (in the $\Im\omega>0$ case) by a continuity argument. However, such an argument cannot be used to establish mode stability for $\omega$ on the real axis. Indeed, the main idea of this paper is to construct a novel integral transformation adapted to the $|a|=M$ case. Our transformation is given by\footnote{The equality here should be understood in the space of square-integrable functions; see already Proposition~\ref{prop:ode-u-tilde} for a precise statement.}
\begin{align}
\begin{split}
\tilde{u}(x) &:=(x^2+2M^2)^{1/2}(x-M)^{-s}(x-2M)^{-2iM\omega}  \times\\
&\qquad\times \int_{M}^\infty e^{\frac{2i\omega}{M}(x-M)(r-M)}(r-M)^{-2iM\omega}e^{2M^2i(\omega-m\upomega_+)(r-M)^{-1}}e^{i\omega r}R(r) dr\,,
\end{split} \label{eq:transformation-intro}
\end{align}
where $R$ is as given in \eqref{eq:mode-sol-intro}. Our \eqref{eq:transformation-intro} retains the relevant properties of Whiting's transformation which enable it to be used in a proof of Theorem~\ref{thm:mode-stability-full}: the $\p_t$ energy is conserved and can be used to infer mode stability for $\tilde{u}$, moreover the map $R\mapsto \tilde{u}$ is injective. However, the transformation is adapted to an extremal Kerr spacetime, where the radial ODE has a different character comparing to the subextremal case (due to the degenerate nature of the horizon when $r=M$) and, moreover, unlike in Whiting's case, does not produce an ODE for $\tilde{u}$ of the same type as that of the radial ODE for $R$.

We note already that the fact that the superradiance threshold $\omega=m\upomega_+$ is omitted from Theorem \ref{thm:mode-stability-full} due to the fact that, for mode solutions, the boundary conditions required of $R$ (see \eqref{eq:mode-sol-intro}) at the future event horizon are not \textit{a priori} continuously defined in the limit $\omega\to m\upomega_+$, similarly to what occurs with the boundary conditions at future null infinity in the limit $\omega\to 0$. 

Our integral transformation allows us to prove Theorem~\ref{thm:mode-stability-full} for both real $\omega$ and $\Im\omega>0$.  In doing so, we combine Shlapentokh-Rothman's approach with the original insight of Whiting for dealing with nontrivial spin: we define a novel radial transformation only for $s\leq0$ and the $s>0$ case is addressed using a set of differential relations connecting $+s$ and $-s$ spin solutions to the Teukolsky equation, known as the Teukolsky--Starobinsky identities (see \cite{Teukolsky1974,Starobinsky1974} for the original papers for $|s|=1,2$ and \cite{Kalnins1989} for a generalization to all $s\in\frac12\mathbb{Z}$), which we will show are a nondegenerate map between mode solutions of different signs of spin. Moreover, we juxtapose our proof of Theorem~\ref{thm:mode-stability-full} with that of \cref{thm:mode-stability-subextremal}, presenting a unified proof of the two statements.

We would like to point out that Theorems \ref{thm:mode-stability-subextremal} and \ref{thm:mode-stability-full} are a testament to the particular nature of the Teukolsky equation on a Kerr black hole spacetime. For a general wave-type equation on a Kerr background, mode stability is not true in general: unstable modes have been constructed for the Klein--Gordon equation \cite{Shlapentokh-Rothman2014} and even for the wave equation to which a non-negative and compactly supported potential is added \cite{Moschidis2017b}. 

\medskip

Finally, we note that Theorems~\ref{thm:mode-stability-subextremal} and \ref{thm:mode-stability-full} can be made quantitative, in particular in the arguably more relevant case of $\omega$ on the real axis. Recall that imposing that $R(r)$ is the radial part of a fixed-frequency solution \eqref{eq:mode-sol-intro} of the Teukolsky equation \eqref{eq:Teukolsky-equation} yields an ODE for $R(r)$. We can define the following solutions of this ODE by their asymptotic behavior (the notation will be made clear in Section~\ref{sec:radial-ODE}),
\begin{alignat*}{3}
\Delta^{s}(r^2+a^2)^{1/2}\swei{R}{s}_{\mc{H}^+}&\sim e^{-i(\omega-m\upomega_+)r^*}&&\text{~ as~}r^*\to-\infty\,, \\
\Delta^{s}(r^2+a^2)^{1/2}\swei{R}{s}_{\mc{I}^+}&\sim e^{i\omega r^*}&&\text{~ as~}r^*\to\infty\,, \\
\Delta^{s}(r^2+a^2)^{1/2}\swei{R}{s}_{\mc{I}^-}&\sim e^{-i\omega r^*}&&\text{~ as~}r^*\to\infty\,, 
\end{alignat*}
where $r^*=\pm \infty$ correspond to $r=\infty$ and $r=M+\sqrt{M^2-a^2}$, respectively. By Theorems~\ref{thm:mode-stability-subextremal} and \ref{thm:mode-stability-full}, $\swei{R}{s}_{\mc{H}^+}$ and $\swei{R}{s}_{\mc{I}^+}$ are linearly independent, i.e.
\begin{align*}
\swei{\mathfrak{W}}{s}:=\Delta^{1+s}\lp[\swei{R}{s}_{\mc{H}^+}\frac{d}{dr}\swei{R}{s}_{\mc{I}^+}-\swei{R}{s}_{\mc{I}^+}\frac{d}{dr}\swei{R}{s}_{\mc{H}^+}\rp] \,,
\end{align*}
is nonvanishing and, thus, admits a lower bound on any compact range of frequencies where the theorems apply. Note that $\swei{\mathfrak{W}}{s}$ depends only on the frequency parameters. The lower bound which can be inferred directly from Theorems \ref{thm:mode-stability-subextremal} and \ref{thm:mode-stability-full} is not explicit; making these results quantitative means providing an \textit{explicit, computable} bound in terms of the black hole parameters, the Teukolsky spin, $s$, and the compact range of frequencies one considers. Continuity of the Wronskian in the entire range $|a|\leq M$ (a natural extension of the arguments in \cite{Hartle1974}) in fact allows us to obtain a bound uniform in $|a|\leq M$:

\begin{theorem}[Quantitative mode stability, rough statement]\label{thm:quantitative-intro} Given Kerr parameters $a$ and $M$ satisfying $|a|\leq M$ and a spin $s$, in any compact set, $\mc{A}$, of real frequency parameters where Theorems~\ref{thm:mode-stability-subextremal} and \ref{thm:mode-stability-full} hold,
\begin{align*}
\lp|\swei{\mathfrak{W}}{s}\rp|^{-1}\leq C(\mc{A},M,s)<\infty\,,
\end{align*}
for any $|a|\leq M$, where $C(\mc{A},M,s)$ can be explicitly computed.
\end{theorem}

As we mentioned, Theorem~\ref{thm:quantitative-intro} restricted to a compact range of subextremal Kerr black holes, which was obtained in \cite{Shlapentokh-Rothman2015}, is instrumental to the proof of boundedness and decay for the wave equation on a subextremal Kerr background \cite{Dafermos2016b}. We expect that our Theorem \ref{thm:quantitative-intro}, which gives a bound uniform in the specific angular momentum of the black hole as well, is also a preliminary step for future understanding of general solutions of the Teukolsky equation on an extremal Kerr background.

Theorem~\ref{thm:quantitative-intro} admits a more direct application which appears in the context of the \textit{scattering problem} for the Teukolsky equation. Recall that, for $\omega$ real,  $\lp(\swei{R}{s}_{\mc{I}^+},\swei{R}{s}_{\mc{I}^-}\rp)$ and $\lp(\swei{R}{s}_{\mc{H}^+},\swei{R}{s}_{\mc{H}^-}\rp)$ are pairs of linearly independent solutions to the radial ODE with spin $s$. Hence, we can write
\begin{align*}
\begin{split}
\frac{\swei{\mathfrak{T}}{s}}{-i(\omega-m\upomega_+)}\swei{R}{s}_{\mc{H}^+}&=\frac{\swei{\mathfrak{R}}{s}}{i\omega}\swei{R}{s}_{\mc{I}^+} +\frac{1}{i\omega}\swei{R}{s}_{\mc{I}^-}\,,\quad s\geq 0\,,\\
\frac{\swei{\mathfrak{{T}}}{s}}{i\omega}\swei{R}{s}_{\mc{I}^+}&=\frac{\swei{\mathfrak{{R}}}{s}}{-i(\omega-m\upomega_+)}\swei{R}{s}_{\mc{H}^+} +\frac{1}{-i(\omega-m\upomega_+)}\swei{R}{s}_{\mc{H}^-}\,,\quad s<0\,.
\end{split}
\end{align*}
for some complex $\swei{\mathfrak{T}}{s}$ and $\swei{\mathfrak{R}}{s}$. The complex numbers $\swei{\mathfrak{T}}{s}$ and $\swei{\mathfrak{R}}{s}$ are known as \textit{transmission} and \textit{reflection coefficients}, respectively, as they measure the fraction of energy of the initial flux on past null infinity, $\mc{I}^-$, if $s\geq 0$, or past event horizon, $\mc{H}^-$, if $s<0$, that is scattered to the two null surfaces to the future, the future event horizon, $\mc{H}^+$, and future null infinity, $\mc{I}^+$ (see Section \ref{sec:scattering} for a discussion for general $s$ and \cite{Dafermos2014} for $s=0$). An immediate application of Theorem~\ref{thm:quantitative-intro} is
\begin{corollary}[rough statement]\label{cor:scattering-intro} Given Kerr parameters $a$ and $M$ satisfying $|a|\leq M$ and a spin $s\in\lp\{0,\pm 1\frac12, \pm 1, \pm  \frac32,\pm 2\rp\}$, in any \textbf{admissible} compact set, $\mc{B}\subset \mc{A}$, of real frequency parameters with $\omega\neq 0$ and $\omega\neq m\upomega_+$,
\begin{align*}
\lp|\swei{\mathfrak{T}}{s}\rp|^{2}+\lp|\swei{\mathfrak{R}}{s}\rp|^{2}\leq C(\mc{B},M,s)<\infty\,,
\end{align*}
where $C(\mc{B},M,s)$ can be explicitly computed. (The definition of admissibility is specified in the full statement, Corollary~\ref{cor:bddness-scattering}.)
\end{corollary}

On its own, Corollary~\ref{cor:scattering-intro} cannot yet yield a  complete scattering theory (we direct the reader to the introduction of \cite{Dafermos2014} for an overview of the topic) for the Teukolsky equation on Kerr black hole spacetimes, as such a theory would require uniform boundedness of the transmission and reflection coefficients for \textit{all} admissible frequency parameters. For subextremal Kerr spacetimes, we note that the uniform boundedness of the transmission and reflection coefficients for $s=0$ follows directly from the estimates in \cite{Dafermos2016b}, which enable  the authors to construct a full scattering theory in \cite{Dafermos2014}. 

\bigskip

\textit{Outline.} This paper is  organized as follows. In \cref{sec:preliminaries}, we introduce the Kerr spacetime. We also discuss fixed-frequency solutions to the Teukolsky equation~\eqref{eq:Teukolsky-equation}, giving a precise definition of mode solution. Finally, we introduce the Teukolsky--Starobinsky identities which are then used to derive an energy identity for nontrivial spin. We finish the section by discussing how superradiance and the presence of an ergoregion present an obstruction to the proof of Theorems \ref{thm:mode-stability-subextremal} and \ref{thm:mode-stability-full}. In \cref{sec:integral-transformation}, we introduce the relevant integral transformations we will consider for the radial ODE when $|a|<M$ and $|a|=M$ and derive the transformed wave-type equations that they satisfy. In \cref{sec:proof}, we use these transformations to prove mode stability for $s\leq 0$ and, using  the Teukolsky--Starobinsky identities, extend the result to $s>0$; we also provide a different proof that Theorem \ref{thm:mode-stability-full} for $\Im\omega>0$ follows by Theorem \ref{thm:mode-stability-subextremal} for the same case. Finally, in \cref{sec:quantitative}, we give a precise statement and proof of Theorem \ref{thm:quantitative-intro}. We moreover discuss the setup of scattering problem for the Teukolsky equation and give the precise statement and proof of Corollary~\ref{cor:scattering-intro}. The paper concludes with a discussion of the frequencies $\omega=0$ and $\omega=m\upomega_+$ which are left out of Theorems~\ref{thm:mode-stability-subextremal} and \ref{thm:mode-stability-full}.

\bigskip 
\textit{Acknowledgements.} This work was supported by the EPSRC grant EP/L016516/1. The author would like express their gratitude to Dejan Gajic for suggesting this problem and pointing to useful references, to Mihalis Dafermos and Yakov Shlapentokh-Rothman for many fruitful discussions and helpful remarks, and to Hamed Masaood for sharing his insight into the scattering problem for the Teukolsky equation.

\section{Preliminaries}
\label{sec:preliminaries}

In this section, we start by introducing the Kerr spacetime for $|a|\leq M$ in \cref{sec:kerr-spacetime}. 
In Section \ref{sec:separation}, we review the separation of the Teukolsky equation for fixed frequency solutions, leading up to a rigorous definition of mode solution in Section \ref{sec:def-mode-sol}. 

In \cref{sec:teukolsky-starobinsky}, we review the Teukolsky--Starobinsky equations relating mode solutions of spin $+s$ to solutions of spin $-s$ and vice-versa. Finally, in \cref{sec:superradiance} we review the definition of energy in the context of the Teukolsky equation and recall the main obstruction to obtaining a mode stability statement for a rotating, $a\neq 0$, Kerr black hole, thus justifying the need for a Whiting-type transformation.

\subsection{The Kerr spacetime}
\label{sec:kerr-spacetime}

In this section, we recall some well-known properties of the Kerr spacetime; we refer the reader to \cite{Chandrasekhar,ONeill1995} or any standard textbook in general relativity for further details. 

Fix parameters $(a,M)$ with $M>0$ and $|a|\leq M$ and define
\begin{equation}
r_\pm:=M\pm\sqrt{M^2-a^2}\,. \label{eq:rpm}
\end{equation} 
The  Kerr exterior spacetime, $\mc{R}$, is a manifold-with-boundary which is covered by Kerr-star coordinates  $(t^*,r,\theta^*,\phi^*)\in \mathbb{R}\times [r_+,\infty)\times \mathbb{S}^2$ globally, apart from the usual degeneration of spherical coordinates. The future event horizon is defined to be $\mc{H^+}:=\p \mc{R}=\{r=r_+\}$.

More frequently, we will work in Boyer--Lindquist coordinates $(t,r,\theta,\phi)\in\mathbb{R}\times (r_+,\infty)\times \mathbb{S}^2$, which are obtained by the relations
\begin{gather}
t(t^*,r):=t^*-\overline{t}(r)\,,  \quad \quad \theta:=\theta^*\, , \quad \quad \phi(\phi^*,r):=\phi^* -\overline{\phi}(r) \mod 2\pi\,, \label{eq:kerr-star}
\end{gather}
where the functions $\overline{\phi}$ and $\overline{t}(r)$ are defined by
\begin{gather}
\frac{d\overline{\phi}}{dr}:=\frac{a}{\Delta}\,, \quad \frac{d\overline{t}}{dr}:=\frac{r^2+a^2}{\Delta}\,,
\label{eq:overline-t-phi}
\end{gather}
and some initial condition. These coordinates parametrize a manifold without boundary, $\tilde{\mc{R}}$, such that $\mc{R}=\tilde{\mc{R}}\cup \mc{H}^+$. With respect to Boyer--Lindquist coordinates, the Kerr metric becomes
\begin{align}
\begin{split}
g_{a,M} &= -\frac{\Delta -a^2\sin^2\theta}{\rho^2}dt^2- \frac{4Mar\sin^2\theta}{\rho^2} dtd\phi  \\ 
&\qquad+\left(\frac{(r^2+a^2)^2-\Delta a^2 \sin^2\theta}{\rho^2}\right)\sin^2\theta d\phi^2 +\frac{\rho^2}{\Delta}dr^2+\rho^2 d\theta^2\,,
\end{split} \label{eq:kerr-metric}
\end{align}
where 
\begin{gather*}
\rho^2 := r^2+a^2\cos^2\theta, \quad\quad \Delta:=r^2-2Mr+a^2=(r-r_+)(r-r_-) \,.
\end{gather*}

Starting from Boyer--Lindquist coordinates, we can define coordinates $(^*t,r,^*\theta,^*\phi)\in\mathbb{R}\times[r_+,\infty)\times\mathbb{S}^2$ by 
\begin{gather}
^*t(t,r):=t-\overline{t}(r)\,,  \quad \quad ^*\theta:=\theta\,, \quad \quad ^*\phi(\phi,r):=\phi -\overline{\phi}(r) \mod 2\pi\,, \label{eq:kerr-star-2}
\end{gather}
which, as the metric extends smoothly to $\mc{H}^-:=\{r=r_+\}$ in this chart, enable us to extend $\mc{R}$ to a larger manifold with boundary $\mc{R}\cup\mc{H}^-$.

It will also be convenient to define a new radial coordinate $r^*\colon(r_+,\infty)\to (-\infty,\infty)$ which is the unique function satisfying $r^*(3M)=0$ and
\begin{gather}
\frac{dr^*}{dr}=\frac{r^2+a^2}{\Delta} \text{,~~with~~} \Delta:=r^2-2Mr+a^2=(r-r_+)(r-r_-)\,. \label{eq:r-star}
\end{gather}

\subsection{Separation of the Teukolsky equation}
\label{sec:separation}

In this section, we will consider fixed frequency solutions to the Teukolsky equation~\eqref{eq:Teukolsky-equation}. We begin by identifying the range of frequencies we will be working with in Section \ref{sec:admissible-freq}. In Sections~\ref{sec:angular-ODE} and \ref{sec:radial-ODE}, we introduce the relevant ODEs arising from the separation of the Teukolsky equation~\eqref{eq:Teukolsky-equation}. The definition of mode solution, central to the proof of mode stability, is finally given in Section~\ref{sec:def-mode-sol}.

\subsubsection{Admissible frequencies}
\label{sec:admissible-freq}

For $\omega\in\mathbb{C}$ and $m\in\frac12\mathbb{Z}$, it will be convenient to define:
\begin{equation}
\xi := -i\frac{2M r_+}{r_+-r_-}(\omega-m\upomega_+) \,,\quad \beta:=2iM^2(\omega-m\upomega_+)\,, \quad \upomega_+:=\frac{a}{2Mr_+}\,. \label{eq:xi-upomega+}
\end{equation}

For the remainder of this paper, we will be interested in the following parameters:
\begin{definition}[Admissible frequencies] \label{def:admissible-freqs} Fix $s\in\frac12\mathbb{Z}$. 
\begin{enumerate}
\item We say the frequency parameter $m$ is admissible with respect to $s$ when, if $s$ is an integer, $m$ is also an integer and when, if $s$ is a half-integer, so is $m$.
\item We say the frequency pair $(m,l)$ is admissible with respect to $s$ when $m$ is admissible with respect to $s$, $l$ is an integer or half-integer if $s$ is an integer or half-integer, respectively, and $l\geq \max\{|m|,|s|\}$. 
\item We say the frequency triple $(\omega,m,l)$ is admissible with respect to $s$ when the pair $(m,l)$ is admissible with respect to $s$ and $\omega\in\mathbb{R}\backslash\{0\}$.
\item We say the frequency triple $(\omega,m,\lambda)$ is admissible with respect to $s$ when $m$ is admissible with respect to $s$ and 
$$(\omega,\lambda)\in\lp\{(\omega,\lambda)\in\mathbb{C}^2\colon\Im\omega>0 \text{~and~} \Im(\lambda\overline{\omega})<0\rp\}\cup\lp(\mathbb{R}\backslash\{0\}\rp)\times\mathbb{R}\,.$$
\item Fix $M>0$, $|a|\leq M$. We say the frequency pair $(\omega,m)$ is admissible with respect to $a$ when, if $|a|=M$, $\omega\neq m\upomega_+$. 
\end{enumerate}
\end{definition}

\subsubsection{The angular ODE}
\label{sec:angular-ODE}

We begin with a definition of smoothness for functions which we call \textit{spin-weighted}:
\begin{definition}[{\cite[Section 2.2.1]{Dafermos2017}}] \label{def:smooth-spin-weighted}
Let
\begin{gather*}
\tilde{Z}_1 = -\sin \phi \partial_\theta + \cos \phi  \left( -is \csc \theta - \cot \theta \partial_\phi\right) \,, \\
\tilde{Z}_2 = - \cos \phi \partial_\theta - \sin \phi  \left( -is \csc \theta - \cot \theta \partial_\phi\right) \,, \quad 
\tilde{Z}_3 = \partial_\phi \, .
\end{gather*}

We say a complex-valued function $f$ of the coordinates $(\theta,\phi)$ is a \textit{smooth $s$-spin-weighted function on $S^2$} if for any $k_1,k_2,k_3 \in \mathbb{N} \cup \{0\}$, $(\tilde{Z}_1)^{k_1} (\tilde{Z}_2)^{k_2} (\tilde{Z}_3)^{k_3} f$ is a smooth function of the coordinates away from the poles at $\theta=0$ and $\theta=\pi$ and such that $e^{is\phi}(\tilde{Z}_1)^{k_1} (\tilde{Z}_2)^{k_2} (\tilde{Z}_3)^{k_3} f$ and $e^{-is\phi}(\tilde{Z}_1)^{k_1} (\tilde{Z}_2)^{k_2} (\tilde{Z}_3)^{k_3} f$ extend continuously to, respectively, $\theta=0$ and $\theta=\pi$.
\end{definition}

We are now ready to introduce the angular ODE and its smooth $s$-spin-weighted solutions

\begin{proposition}[Smooth spin-weighted solutions of the angular ODE] \label{def:angular-ode}
Fix $s\in\frac12\mathbb{Z}$, let $m$ be admissible with respect to $s$, and assume $\nu\in\mathbb{C}$. Consider the angular ODE
\begin{gather}
\begin{gathered}
-\frac{d}{d\theta}\lp(\sin\theta\frac{d}{d\theta}\rp)S_{m,\bm\uplambda}^{[s],\,(\nu)}(\theta)
+ \lp(\frac{(m+s\cos\theta)^2}{\sin^2\theta}-\nu^2\cos^2\theta+2\nu s \cos\theta\rp)S_{m,\bm\uplambda}^{[s],\,(\nu)}(\theta)\\ =\bm\uplambda_{m}^{[s],\,(\nu)} S_{m,\bm\uplambda}^{[s],\,(\nu)}(\theta) \,,
\end{gathered} \label{eq:angular-ode}
\end{gather}
with the boundary condition that $e^{im\phi}S_{m,\bm\uplambda}^{[s],\,(\nu)}$ is a non-trivial smooth $s$-spin-weighted function (see Definition~\ref{def:smooth-spin-weighted}). 
\begin{enumerate}
\item {\normalfont Basic properties of the eigenvalues.}
\begin{enumerate}
\item We have $\overline{\bm\uplambda_{m}^{[s],\,(\nu)}}=\bm\uplambda_{m}^{[s],\,(\overline{\nu})}$. Hence, if $\nu\in\mathbb{R}$, then $\bm\uplambda_{m}^{[s],\,(\nu)}\in\mathbb{R}$.
\item If $\Im\nu>0$, then $\Im\lp( \overline{\nu}\,\bm\uplambda_{m}^{[s],\,(\nu)}\rp)<0$.
\end{enumerate}

\item {\normalfont Eigenvalues in $\mathbb{R}$.} For each $\nu\in\mathbb{R}$, there are countably many solutions to the problem; using $l$ as an index, we write such solutions, also called $s$-spin-weighted spheroidal harmonics with spheroidal parameter $\nu$, as $e^{im\phi}S_{ml}^{[s],\,(\nu)}$ and denote the corresponding eigenvalues by $\bm\uplambda_{ml}^{[s],\,(\nu)}$. The parameter $l$ is chosen to be admissible with respect to $s$ and such that the $\bm\uplambda_{ml}^{[s],\,(0)}=l(l+1)-s^2$ for $\nu =0$ and $\bm\uplambda_{ml}^{[s],\,(\nu)}$ varies smoothly with $\nu$. Moreover, $\lp\{e^{im\phi}S_{ml}^{[s],\,(\nu)}\rp\}_{ml}$ form a complete orthonormal basis on the space of smooth $s$-spin-weighted spheroidal functions (see Definition~\ref{def:smooth-spin-weighted}).

\item {\normalfont Eigenvalues in $\mathbb{C}$}. Given $\nu_0\in\mathbb{R}$, an eigenvalue $\bm\uplambda_{ml}^{[s],\,(\nu_0)}\in \mathbb{R}$ can be analytically continued to $\nu\in\mathbb{C}$ except for finitely many branch points (with no finite accumulation point), located away from the real axis, and branch cuts emanating from these. We define $\bm\uplambda_{ml\nu_0}^{[s],\,(\nu)}$, for $\nu_0\in\mathbb{R}$, as a global multivalued complex function of $\nu$ such that $\bm\uplambda_{ml\nu_0}^{[s],\,(\nu_0)}=\bm\uplambda_{ml}^{[s],\,(\nu_0)}$.  (We note that each branch point of $\bm\uplambda_{m\nu_0}^{[s],\,(\nu)}$ is a point where $\bm\uplambda_{ml\nu_0}^{[s],\,(\nu)}=\bm\uplambda_{m\tilde{l}\tilde{\nu}_0}^{[s],\,(\nu)}$, for $\tilde{\nu}_0\neq\nu_0$ and/or $\tilde{l}\neq l$.) For each $\nu$ and $\nu_0$, there are solutions $e^{im\phi}S_{ml\nu_0}^{[s],\,(\nu)}$ to \eqref{eq:angular-ode} associated with the eigenvalue $\bm\uplambda_{ml\nu_0}^{[s],\,(\nu)}$; however $\lp\{e^{im\phi}S_{ml\nu_0}^{[s],\,(\nu)}\rp\}_{ml}$ but they do not necessarily form a complete basis of the space of smooth $s$-spin-weighted spheroidal functions.
\end{enumerate}
\end{proposition}
\begin{remark} \label{rmk:angular-prop-admissibility}
With Proposition~\ref{def:angular-ode}, setting $\nu=a\omega$, we can now motivate the admissibility conditions in Definition~\ref{def:admissible-freqs}, in particular for frequency parameters $\lambda$ and $l$: statement 1 in the proposition translates into the item 4 in the definition.
\end{remark}

\begin{remark} \label{rmk:complex-eigenvalues-completeness}
For $\nu\in\mathbb{C}$, Proposition~\ref{def:angular-ode} does not yield a complete basis of the space of $s$-spin-weighted spheroidal functions. For this reason, when $\omega\in\mathbb{C}$, the separation of the radial and angular variables does not necessarily account for the full space of admissible solutions to the Teukolsky equation. Consequently, while Theorems~\ref{thm:mode-stability-subextremal} and \ref{thm:mode-stability-full}, in the case $\Im\omega>0$, rule out mode solutions of the form \eqref{eq:mode-sol-intro}, they do not necessarily rule out solutions of the form
\begin{align*}
\swei{\alpha}{s}(t,r,\theta,\phi)= e^{-i\omega t}e^{im\phi}A(r,\theta)\,,
\end{align*}
for $A(r,\theta)$ with suitable boundary conditions. Thankfully, by following the approach in \cite{Dafermos2016b,Dafermos2014} for $s=0$, we expect that, to investigate the boundedness, scattering and decay properties of the Teukolsky equation, it is enough to consider $\omega\in\mathbb{R}$, where indeed we have a complete basis of smooth spin-weighted functions provided by the spin-weighted spheroidal harmonics.
\end{remark}

\begin{proof}[Proof sketch of Proposition~\ref{def:angular-ode}] In what follows, we will drop the sub- and superscripts on the angular eigenvalue whenever this does not lead to ambiguity, setting $\bm\uplambda(\nu):=\bm\uplambda_{m}^{[s],\,(\nu)}$.

We begin with statement 1. Assume $S$ is normalized to have unit $L^2$ norm. Multiplying the angular ODE by $\overline{S}$ and integrating by parts, we obtain
\begin{align*}
\Im\lp(\bm\uplambda_{m}^{[s],\,(\nu)}\rp) =\int_{0}^\pi \lp|S_{m,\bm\uplambda}^{[s],\,(\nu)}\rp|^2(\theta) \lp(-\cos^2\theta\Im(\nu^2)+2s\cos\theta\Im\nu\rp)\sin\theta d\theta\,,
\end{align*}
from where (a) becomes clear (see also the case $s=0$ as an application of \cite[Proposition 7, Chapter 1.5]{Meixner1954}). For (b), we multiply the angular ODE \eqref{eq:angular-ode} by $\overline{\omega S}$, integrate by parts and take the imaginary part to obtain
\begin{align*}
\begin{split}
&\Im\lp(\overline{\nu}\bm\uplambda_{m}^{[s],\,(\nu)}\rp)\\
&\quad=-\Im\nu\int_{0}^{\pi}\lp[\lp|\frac{dS_{m,\bm\uplambda}^{[s],\,(\nu)}}{d\theta}\rp|^2+\lp(\frac{(m+s\cos\theta)^2}{\sin^2\theta}+|\nu|^2\cos^2\theta\rp)\lp|S_{m,\bm\uplambda}^{[s],\,(\nu)}\rp|^2\rp]\sin\theta d\theta  \leq 0\,.
\end{split}
\end{align*}
Here, equality is achieved if and only if 
\begin{align*}
\lp|\frac{dS_{m,\bm\uplambda}^{[s],\,(\nu)}}{d\theta}\rp|^2=0\,,\quad\lp(\frac{(m+s\cos\theta)^2}{\sin^2\theta}+|\nu|^2\sin^2\theta\rp)\lp|S_{m,\bm\uplambda}^{[s],\,(\nu)}\rp|^2=0\,,\ \qquad \forall\theta\in(0,\pi)\,,
\end{align*}
in which case one obtains $S\equiv 0$.

Let us focus now on statement 2. For $\nu\in\mathbb{R}$, the operator 
$$\mathring{\slashed{\triangle}}^{[s]}_m=-\frac{d}{d\theta}\lp(\sin\theta\frac{d}{d\theta}\rp)
+ \lp(\frac{(m+s\cos\theta)^2}{\sin^2\theta}-\nu^2\cos^2\theta+2\nu s \cos\theta\rp)$$
is self-adjoint and it follows from Sturm-Liouville theory that it has a countable set of eigenfunctions, which form a complete basis os the space of $s$-spin-weighted spheroidal functions, and countable set of corresponding eigenvalues (see, for instance, \cite{Dafermos2017}). For each $s$, $m$ and $\nu$, these can be indexed by $l\in\frac12\mathbb{Z}$ satisfying the constraints in the statement, so that, in particular, $\lambda_{ml}^{[s],\,(\nu)}$ is smooth in $\nu$ (see \cite[Section 3.22, Proposition 1]{Meixner1954}). This concludes the proof of statement 2.

Now, for statement 3, consider $\nu\in\mathbb{C}$  and fix $m$ and $s$. The angular ODE \eqref{eq:angular-ode} has regular singular points at $\theta=0,\pi$, where one can apply an asymptotic analysis to determine the possible behaviors of a solution. By the boundary conditions in the statement, we are looking for solutions which are finite at $\theta=0,\pi$.

Fix $\nu\in\mathbb{C}$. For $\theta\in(0,\pi]$, we can define a solution satisfying the boundary condition at $\theta=\pi$ by a power series  \cite[Chapter 5]{Olver1973}:
$$S_{m,\bm\uplambda}^{[s],\,(\nu_0)}:=\sum_{k=0}^\infty \swei{c}{s}_k (\cos\theta+1)^{|m-s|/2+k}\,,\quad\theta\in(0,\pi]\,,$$
where $\swei{c}{s}_k$ are uniquely determined by some $\swei{c}{s}_0$, $\nu$, $m$ and $\bm\uplambda$, and we require $\swei{c}{s}_0\neq 0$ so that the solution we have constructed is not trivial. On the other hand, an asymptotic analysis near $\theta=0$ shows that 
\begin{align*}
S_{m,\bm\uplambda}^{[s],\,(\nu_0)}&=F_m^{[s]}(\nu,\bm\uplambda)\sum_{k=0}^\infty \swei{\tilde{b}}{s}_k (\cos\theta-1)^{-|m+s|/2+k}\\
&\qquad+G_m^{[s]}(\nu,m,\bm\uplambda)\sum_{k=0}^\infty \swei{b}{s}_k (\cos\theta-1)^{|m+s|/2+k}\,,\quad\theta\in[0,\pi)\,,
\end{align*}
for $\tilde{b}_0=b_0=1$ and some complex-valued function $F_m^{[s]}$ and $G_m^{[s]}$ which are analytic in their arguments (see \cite[Chapter 5]{Olver1973} or \cite[Appendix A]{Shlapentokh-Rothman2014}). By construction, at least one of $F_m^{[s]}$ and $G_m^{[s]}$ do not vanish. To satisfy the boundary conditions at $\theta=0$, we must have $F_m^{[s]}=0$. Indeed, 
\begin{align}
\text{the pair $(\nu,\bm\uplambda)$ corresponds to an eigenvalue} \Leftrightarrow F_m^{[s]}(\nu,\bm\uplambda)=0\,. \label{eq:iff-eigenvalue} 
\end{align}

Starting from an eigenvalue pair $(\nu_0,\bm\uplambda_0=\bm\uplambda(\nu_0))$ at some $\nu_0\in\mathbb{C}$, we can define the curve $\bm\uplambda=\bm\uplambda(\nu)$ by the implicit function theorem as long as
\begin{align*}
\frac{\p  F_m^{[s]}}{\p\bm\uplambda}\Big|_{(\nu_0,\bm\uplambda_0)} \neq 0\,.
\end{align*} 
It follows from \cite[Proposition B.1]{Shlapentokh-Rothman2014} that this condition holds when $(\nu_0,\bm\uplambda_0)\in\mathbb{R}^2$, so an eigenvalue $\bm\uplambda_{ml}^{[s],\,(\nu_0)}$ for $\nu_0\in\mathbb{R}$ admits a unique analytic extension into a region of the plane $\nu\in\mathbb{C}$ sufficiently close to the real axis; we call the extension $\bm\uplambda_{ml\nu_0}^{[s],\,(\nu)}$.

However, further away from the real axis, there are, in general, branch points and branch cuts; we refer the reader to \cite[Proposition 5 of Chapter 1.5, Proposition 2 of Chapter 3.22]{Meixner1954} for a proof in the case $s=0$ and to \cite[Chapter 3.2]{Meixner1980} for some numerical information on the branch points. For $\nu\in\mathbb{C}$, a branch point at $(\nu,\bm\uplambda_{ml\nu_0}^{[s],\,(\nu)})$ is due to
\begin{align*}
\frac{\p  F_m^{[s]}}{\p\bm\uplambda}\Big|_{(\nu,\bm\uplambda_{ml\nu_0}^{[s],\,(\nu)})} = 0\,,
\end{align*} 
i.e.\ having some finite number of complex pairs $(\nu,\bm\uplambda_{(i)})$, $i=1,...,N$ where $N\geq 2$, at which $F_m^{[s]}$ vanishes. In light of the equivalence \eqref{eq:iff-eigenvalue}, it is clear that $(\nu,\bm\uplambda_{(i)})=(\nu,\bm\uplambda_{ml\nu_0}^{[s],\,(\nu)})$ for some admissible $l$ and some $\nu_0\in\mathbb{R}$. We also refer the reader to \cite{Hartle1974,Stewart1975} for intuition regarding this point. This concludes our proof of statement 3. 
\end{proof}

\subsubsection{The radial ODE}
\label{sec:radial-ODE}

In this section, we discuss the radial ODE
\begin{equation}
\begin{gathered}
\Delta^{-s}\frac{d}{dr}\lp(\Delta^{s+1}\frac{d}{dr}\rp)R_{m\lambda}^{[s],\,(a\omega)}(r)\\
+\lp(\frac{[\omega(r^2+a^2)-a m]^2-2is(r-M)[\omega(r^2+a^2)-a m]}{\Delta}\rp)R_{m\lambda}^{[s],\,(a\omega)}(r) \\
+\lp(4is\omega r-\lambda-a^2\omega^2+2a m \omega\rp)R_{m\lambda}^{[s],\,(a\omega)}(r) = \hat{F}_{m\lambda}^{[s],\,(a\omega)}(r)\,, \end{gathered}\label{eq:radial-ODE}
\end{equation}
with inhomogeneity $\hat{F}_{m\lambda}^{[s],\,(a\omega)}(r)$. In this section and throughout this paper, whenever we refer to \eqref{eq:radial-ODE} as a homogeneous radial ODE, we are implicitly assuming $\hat{F}_{m\lambda}^{[s]}\equiv 0$.

\begin{remark}\label{rmk:admissible-triples} Often, we make statements regarding the radial ODE \eqref{eq:radial-ODE} by itself, so we will consider admissible frequency triples to have the form $(\omega,m,\lambda)$; in this case, solutions are written as $R_{m\lambda}^{[s],\,(a\omega)}$. 

When discussing fixed-frequency solutions to the Teukolsky equation~\eqref{eq:Teukolsky-equation}, the radial ODE we are considering arises from the separation of variables $r$ and $\theta$, so $\lambda$ is \textbf{fixed} to be the separation constant, i.e.\ a choice of $\bm\uplambda_{m}^{[s],\,(a\omega)}$ from the set identified in Proposition \ref{def:angular-ode}. We will not make this choice in our proofs, however it can help gain intuition in the case of $\omega$ real where, by contrast with $\Im\omega>0$ (see Remark~\ref{rmk:complex-eigenvalues-completeness}), Proposition \ref{def:angular-ode} gives a complete basis of the space of smooth spin-weighted functions. For $\omega\in\mathbb{R}$, we will consider admissible frequency triples arising from fixing $\lambda=\bm\uplambda_{ml}^{[s],\,(a\omega)}$ to be of the form $(\omega,m,l)$ and denote solutions of that radial ODE by $R_{ml}^{[s],\,(a\omega)}$.
\end{remark}

For the homogeneous case of \eqref{eq:radial-ODE}, the ODE has a singularity at $r=r_\pm$, which is regular if $|a|<M$ but irregular of rank 1 if $|a|=M$, and an irregular singularity of rank 1 at $r=\infty$. By standard asymptotic analysis, we can span the solution space by two linearly independent asymptotic solutions at each of the singularities (we refer the reader to \cite[Chapters 5 and 7]{Olver1973} and \cite[Appendix A]{Shlapentokh-Rothman2014} for more detail; see also \cite{Erdelyi1956,Ince1956}). Our basis is given by the following:
\begin{definition} Fix $M>0$, $|a|\leq M$, $s\in\frac12\mathbb{Z}$ and an admissible frequency triple $(\omega,m,\lambda)$ with respect to $a$ and $s$. 
\begin{enumerate}
\item Define $\swei{R}{s}_{\mc{H}^+}$ and $\swei{R}{s}_{\mc{H}^-}$ to be the unique classical solutions to the {\normalfont homogeneous} radial ODE~(\ref{eq:radial-ODE}) with boundary conditions
\begin{enumerate}
\item if $|a|<M$,
  \begin{enumerate}
\item $\swei{R}{s}_{\mc{H}^+}(r)(r-r_+)^{-\xi+s}$ is smooth at $r=r_+\,,$
\item $\lp|\lp((r^2+a^2)^{1/2}\Delta^{s}(r-r_+)^{-\xi}\swei{R}{s}_{\mc{H}^+}\rp)\big|_{r=r_+}\rp|^2=1\,;$
  \end{enumerate}
\item if $|a|<M$ and additionally $\omega-m\upomega_+\neq 0$,
  \begin{enumerate}
\item $\swei{R}{s}_{\mc{H}^-}(r)(r-r_+)^{\xi}$ is smooth at $r=r_+\,,$
\item $\lp|\lp((r^2+a^2)^{1/2}\Delta^{s/2}(r-r_+)^{\xi}\swei{R}{s}_{\mc{H}^-}\rp)\big|_{r=r_+}\rp|^2=1\,;$
  \end{enumerate}
\item if $|a|=M$ and additionally $\omega-m\upomega_+\neq 0$,
  \begin{enumerate}
\item $(r-M)^{2iM\omega+2s}\swei{R}{s}_{\mc{H}^+}(r)e^{-\beta(r-M)^{-1}}$ and $(r-M)^{-2iM\omega}\swei{R}{s}_{\mc{H}^-}(r)e^{\beta(r-M)^{-1}}$ are both smooth at $r=M\,,$
\item $\lp|\lp((r^2+M^2)^{1/2}(r-M)^{2iM\omega+2s}e^{-\beta(r-M)^{-1}}\swei{R}{s}_{\mc{H}^+}\rp)\big|_{r=M}\rp|^2=1\,,$ and 
\item $\lp|\lp((r^2+M^2)^{1/2}(r-M)^{-2iM\omega}e^{\beta(r-M)^{-1}}\swei{R}{s}_{\mc{H}^-}\rp)\big|_{r=M}\rp|^2=1\,.$
  \end{enumerate}
\end{enumerate}
\item Define $\swei{R}{s}_{\mc{I}^+}$ and $\swei{R}{s}_{\mc{I}^-}$ to be the unique classical solution to the {\normalfont homogeneous} radial ODE~(\ref{eq:radial-ODE}) and  boundary conditions
\begin{enumerate}
\item $\swei{R}{s}_{\mc{I}^+}\sim e^{i\omega r}r^{2Mi\omega-2s-1}$ and $\swei{R}{s}_{\mc{I}^-}\sim e^{-i\omega r}r^{-2Mi\omega-1}$ asymptotically\footnote{This notation means that there are constants $\{c_k\}_{k=0}^\infty$ such that for every $N\geq 1$, $\swei{R}{s}_{\mc{I}^+}(r) = e^{i\omega r+ 2iM\omega \log r}\sum_{k=0}^{N}c_kr^{-2s-k-1}+O(r^{-2s-N-2})$ for large $r$.} as $r\to \infty\,.$
\item $\lp|\lp(e^{-i\omega r}r^{-2iM\omega}\Delta^{s}(r^2+a^2)^{1/2}\swei{R}{s}_{\mc{I}^+}\rp)\big|_{r=\infty}\rp|^2=1\,,$ and 
\item $\lp|\lp(e^{i\omega r}r^{2iM\omega-2s}\Delta^{s}(r^2+a^2)^{1/2}\swei{R}{s}_{\mc{I}^-}\rp)\big|_{r=\infty}\rp|^2=1\,.$
\end{enumerate}
\end{enumerate}
\label{def:uhor-uout}
\end{definition}

In light of the previous comments, we find that we have the following representation for solutions:
\begin{lemma} \label{lemma:R-general-asymptotics} Fix $M>0$, $|a|\leq M$, $s\in\frac12\mathbb{Z}$ and an admissible frequency triple $(\omega,m,\lambda)$ with respect to $a$ and $s$ such that $\omega\neq m\upomega_+$. A solution $R_{m\lambda}^{[s],\,(a\omega)}$ to the homogeneous radial ODE~\eqref{eq:radial-ODE} can be written as, dropping most sub and superscripts,
\begin{equation}
\begin{split}
\swei{R}{s}&= \frac{1}{2Mr_+}\lp(\swei{a}{s}_{\mc{H}^+}\swei{R}{s}_{\mc{H}^+}+\swei{a}{s}_{\mc{H}^-}\swei{R}{s}_{\mc{H}^-}\rp)\,,\\
\swei{R}{s}&= \swei{a}{s}_{\mc{I}^+}\swei{R}{s}_{\mc{I}^+}+\swei{a}{s}_{\mc{I}^-}\swei{R}{-s}_{\mc{I}^-}\,,
\end{split}\label{eq:R-general-asymptotics}
\end{equation}
for some $\swei{a}{s}_{\mc{H}^\pm},\swei{a}{s}_{\mc{I}^\pm}\in\mathbb{C}$. 
\end{lemma}

In the more general inhomogeneous setting of \eqref{eq:radial-ODE}, we define
\begin{definition}[Outgoing solution to the radial ODE]\label{def:outgoing-radial-solution}  Fix $M>0$, $|a|\leq M$, $s\in\frac12\mathbb{Z}$ and an admissible frequency triple $(\omega,m,\lambda)$ with respect to $a$ and $s$. Suppose $\hat{F}_{m\lambda}^{[s],\,(a\omega)}$ is compactly supported with support away from $r=r_+$. We say $R$ is an {\normalfont outgoing} solution to the radial ODE \eqref{eq:radial-ODE} if it is nontrivial and satisfies the boundary conditions
\begin{enumerate}
\item if $|a|=M$, $R_{m\lambda}^{[s],\,(a\omega)}(r)(r-M)^{2iM\omega+2s}e^{-\beta(r-M)^{-1}}$ is smooth at $r=M$;
\item if $|a|<M$, $R_{m\lambda}^{[s],\,(a\omega)}(r)(r-r_+)^{s-\xi}$ is smooth at $r=r_+$;
\item for $|a|\leq M$, $R_{m\lambda}^{[s],\,(a\omega)}(r) \sim e^{i\omega r}r^{2iM\omega-1-2s}$ asymptotically as $r\to \infty$;
\end{enumerate}
\end{definition}

\begin{remark} \label{rmk:outgoing-radial-solution}
If $R_{m\lambda}^{[s],\,(a\omega)}$ is an outgoing solution to the \textit{homogeneous} radial ODE~\eqref{eq:radial-ODE}, in the sense of Definition~\ref{def:outgoing-radial-solution}, then, for some for some $\swei{a}{s}_{\mc{H}^+},\swei{a}{s}_{\mc{I}^+}\in\mathbb{C}$, it can be written as \begin{equation}
\begin{split}
\swei{R}{s}&= \frac{1}{2Mr_+}\swei{a}{s}_{\mc{H}^+}\swei{R}{s}_{\mc{H}^+}\,,\\
\swei{R}{s}&= \swei{a}{s}_{\mc{I}^+}\swei{R}{s}_{\mc{I}^+}\,.
\end{split}\label{eq:R-general-asymptotics-outgoing}
\end{equation}
\end{remark}

Often, it will be useful to work with the following rescaling of a solution $R$ of the inhomogeneous radial ODE \eqref{eq:radial-ODE}:
$$u_{m\lambda}^{[s],\,(a\omega)}(r)=(r^2+a^2)^{1/2}\Delta^{s/2}R_{m\lambda}^{[s],\,(a\omega)}(r)\,.$$ 
In terms of $u$, \eqref{eq:radial-ODE} becomes, dropping sub and superscripts,
\begin{gather*}
u''+(\omega^2-V)u=H\,,\numberthis\label{eq:u-Schrodinger}
\end{gather*}
where $H=\Delta^{1+s/2}(r^2+a^2)^{-3/2}\hat{F}$, the derivatives are taken with respect to the $r^*$ coordinate \eqref{eq:r-star}, and
\begin{align}
V&=\frac{\Delta(\lambda+s+s^2+a^2\omega^2)+4Mam\omega r-a^2m^2}{(r^2+a^2)^2}+\frac{\Delta}{(r^2+a^2)^4}\lp(a^2\Delta +2Mr (r^2-a^2)\rp)\nonumber \\
&\qquad+\frac{M^2-a^2}{(r^2+a^2)^2}s^2  -2is \frac{\omega(r^3-3Mr^2+a^2r+Ma^2)+am(r-M)}{(r^2+a^2)^2}\,. \label{eq:Teukolsky-potential}
\end{align}

We note that Definition~\ref{def:uhor-uout} can be rephrased in terms of this rescaling: we can define 
\begin{gather}\label{eq:uout-uhor}
\swei{u}{s}_{\mc{I}^+}:=(r^2+a^2)^{1/2}\Delta^{s/2}\swei{R}{s}_{\mc{I}^+}\,,\qquad\swei{u}{s}_{\mc{H}^+}:=(r^2+a^2)^{1/2}\Delta^{s/2}\swei{R}{s}_{\mc{H}^+}\,.
\end{gather} 

\subsubsection{Mode solutions}
\label{sec:def-mode-sol}

We are now ready to define precisely what is meant by \textit{mode solution}:
\begin{definition}[Mode solution] \label{def:mode-solution}  Fix $M>0$, $|a|\leq M$ and $s\in\frac12\mathbb{Z}$. Let $\upalpha^{[s]}$ be a solution to the Teukolsky equation \eqref{eq:Teukolsky-equation} on $\tilde{\mc{R}}$ (see Section~\ref{sec:preliminaries}) which depends on the variables $\theta$ and $\phi$ as a smooth $s$-spin weighted function. We say $\upalpha^{[s]}$ is a mode solution if there exists an admissible frequency triple $(\omega,m,\bm\uplambda_{m}^{[s],\,(a\omega)})$ with respect to $a$ and $s$ (see Definition~\ref{def:admissible-freqs}) and with $\bm\uplambda_{m}^{[s], (a\omega)}$, as identified in Proposition~\ref{def:angular-ode}, such that, in Boyer--Lindquist coordinates,
\begin{equation} 
\swei{\alpha}{s}(t,r,\theta,\phi) = e^{-i\omega t}e^{im\phi}S_{m\bm\uplambda}^{[s],\,(a\omega)}(\theta)R_{m\bm\uplambda}^{[s],\,(a\omega)}(r)\,, \label{eq:mode-solution}
\end{equation}
where 
\begin{enumerate}
\item $e^{im\phi}S_{m\bm\uplambda}^{[s],\,(a\omega)}(\theta)$ is a smooth $s$-spin-weighted solution of the angular ODE \eqref{eq:angular-ode} (see Proposition~\ref{def:angular-ode}), with respect to the spheroidal parameter $\nu=a\omega$;
\item $R_{m\bm\uplambda}^{[s],\,(a\omega)}(r)$ is an outgoing solution, in the sense of Definition \ref{def:outgoing-radial-solution} and Remark \ref{rmk:outgoing-radial-solution}, to the {\normalfont homogeneous} radial ODE \eqref{eq:radial-ODE} with parameter $\lambda$ replaced by $\bm\uplambda_{m}^{[s], (a\omega)}$, identified in Proposition~\ref{def:angular-ode}.
\end{enumerate}
\end{definition}

For $s=0$, we can motivate the boundary conditions we have imposed for the radial ODE \eqref{eq:radial-ODE} in \cref{def:mode-solution}. Recall that, at the singularities of the radial ODE \eqref{eq:radial-ODE}, $r=r_+$ and $r=\infty$, there are \textit{a priori} two possible linearly independent asymptotic behaviors (see Lemma~\ref{lemma:R-general-asymptotics}). Our choice of boundary conditions at $r=r_+$ ensures that mode solutions of the wave equation ($s=0$) extend smoothly to the future event horizon $\mc{H}^+$:
\begin{lemma} \label{lemma:smooth-extension-horizon} Mode solutions, in the sense of Definition~\ref{def:mode-solution},  to the Teukolsky equation~\eqref{eq:Teukolsky-equation} with $s=0$  extend smoothly to $\mc{H}^+$.
\end{lemma}
\begin{proof}
To check that our choice of boundary conditions lead to solutions which are smooth at the horizon, we must change to coordinates which are well defined at the horizon. In Kerr-star coordinates \eqref{eq:kerr-star},
\begin{equation}
\swei{\alpha}{s}(t^*,r,\theta,\phi^*) = e^{i\lp(\omega \overline{t}-m\overline{\phi}\rp)}e^{-i\omega t^*}e^{im\phi^*}S_{m\ell}^{[s],\,(a\omega)}(\theta)R_{m\ell}^{[s],\,(a\omega)}(r)\,. \label{eq:mode-sol-kerr-star}
\end{equation}

For subextremal Kerr, we have 
\begin{align*}
&\frac{d\overline{t}}{dr}=\frac{r^2+a^2}{\Delta} = 1+ \frac{2Mr_+}{r_+-r_-}\frac{1}{r-r_+}-\frac{2Mr_-}{r_+-r_-}\frac{1}{r-r_-}\\
&\quad \Rightarrow \overline{t}(r)=r+\frac{2Mr_+}{r_+-r_-}\log\lp(\frac{r-r_+}{r}\rp)-\frac{2Mr_-}{r_+-r_-}\log\lp(\frac{r-r_-}{r}\rp) + C_1\,,\\
&\frac{d\overline{\phi}}{dr}=\frac{a}{\Delta} = \frac{a}{r_+-r_-}\lp(\frac{1}{r-r_+}-\frac{1}{r-r_-}\rp)\Rightarrow \overline{\phi}(r)=\frac{a}{r_+-r_-}\log\lp(\frac{r-r_+}{r-r_-}\rp)+C_2\,,
\end{align*}
so, as $r\to r_+$, we have
\begin{align*}
i\lp(\omega\overline{t}-m\overline{\phi}\rp)= -\xi\log(r-r_+)+ o\lp(\log(r-r_+)\rp)\,.
\end{align*}

For extremal Kerr, since $|a|=M$, we have
\begin{align*}
&\frac{d\overline{t}}{dr}=\frac{r^2+M^2}{(r-M)^2} = 1+ \frac{2Mr}{(r-M)^2} \Rightarrow \overline{t}(r)=r -\frac{2M^2}{r-M}+ 2M\log(r-M)+C_3\,,\\
&\frac{d\overline{\phi}}{dr}=\frac{a}{(r-M)^2} \Rightarrow \overline{\phi}(r)=-\frac{a}{r-M}+C_4\,,
\end{align*}
so, as $r\to M$, we have
\begin{align*}
i\lp(\omega\overline{t}-m\overline{\phi}\rp)=-\frac{\beta}{r-M}+ 2Mi\omega\log(r-M)+ o\lp(\log(r-M)\rp)\,.
\end{align*}

Thus, the boundary conditions at $r=r_+$ in our definition precisely ensure that (\ref{eq:mode-sol-kerr-star}) is smooth at the horizon in both the extremal and subextremal case.
\end{proof}

At $r=\infty$, the boundary condition in \cref{def:mode-solution} ensures that, when $s=0$, mode solutions have finite energy on suitable spacelike hypersurfaces on a Kerr spacetime: hyperboloidal and asymptotically flat hypersurfaces when $\Im\omega=0$ and asymptotically flat hypersurfaces when $\Im\omega>0$ (see \cref{fig:hypersurfaces} for a sketch of these hypersurfaces and \cite[Appendix D]{Shlapentokh-Rothman2015} for precise definitions and a proof that the behavior as $r\to\infty$ is compatible with these statements).

\begin{figure}[htbp]
    \centering
    \begin{subfigure}[t]{0.5\textwidth}
        \centering
 \includegraphics[scale=1]{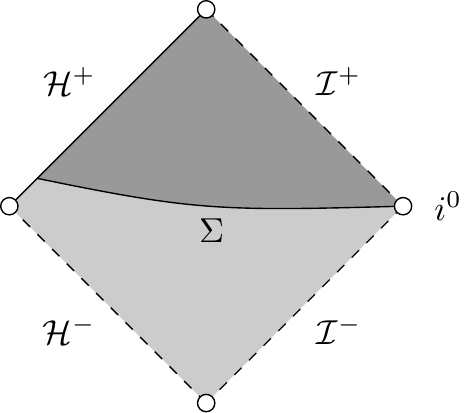}
        \caption{}
    \end{subfigure}%
    ~ 
    \begin{subfigure}[t]{0.45\textwidth}
        \centering
\includegraphics[scale=1]{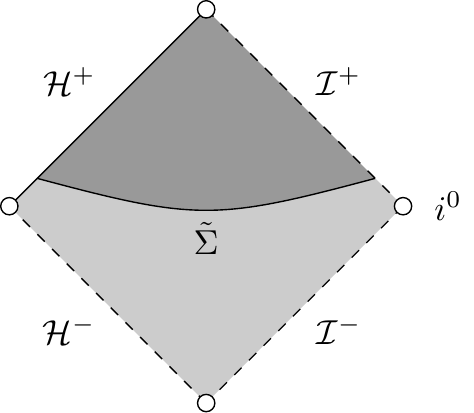}
        \caption{}
    \end{subfigure}
    \caption{These Penrose diagrams for the exterior of a Kerr black hole, $\mc{R}$, with $|a|\leq M$ show asymptotically flat (a) and hyperboloidal hypersurfaces (b), which are spacelike hypersurfaces that intersect the future event horizon $\mc{H}^+$ and, respectively, spacelike infinity, $i^0$, or null infinity, $\mc{I}^+$.   Mode solutions have finite energy on surfaces (a)  and (b) if $\Im\omega>0$ and on surfaces (b) if $\Im\omega=0$.}
    \label{fig:hypersurfaces}
\end{figure}

For general $s$, since the Teukolsky equation \eqref{eq:Teukolsky-equation} is derived in the Newman--Penrose formalism with the algebraically special frame for Kerr, which is itself not regular at $\mc{H}^+$, \cref{lemma:smooth-extension-horizon} does not apply. However, using an appropriately weighted definition of energy for solutions of \eqref{eq:Teukolsky-equation}, we can also show that mode solutions for $s\neq 0$ have finite energy on the same spacelike hypersurfaces as for $s=0$ (see \cite{Dafermos2017} for definitions of energy for $s=\pm 2$).

\subsection{The Teukolsky--Starobinsky identities}
\label{sec:teukolsky-starobinsky}

In this section, we will introduce the Teukolsky--Starobinsky identities. We begin by defining the differential operators
\begin{align}
\mc{D}^{\pm}_n &= \frac{d}{dr}\pm i\lp(\frac{\omega(r^2+a^2)}{\Delta} -\frac{am}{\Delta}\rp)+\frac{2n(r-M)}{\Delta}\,, \label{eq:def-D-pm}\\
\mc{L}^{\pm}_n &= \frac{d}{d\theta}\pm\lp(\frac{m}{\sin\theta}-a\omega\cos\theta\rp)+n\cot\theta\,. \label{eq:def-L-pm}
\end{align}

\subsubsection{The Teukolsky--Starobinsky constant}

We state the following two propositions, which are implicit in much of the literature:
\begin{proposition}\label{prop:eigenfunctions-angular} Fix $s\in\{0,\frac12,1,\frac32,2\}$, $\omega\in\mathbb{C}$ and a frequency parameter $m$ admissible with respect to $s$. Then, a solution of the angular ODE~\eqref{eq:angular-ode} with spin $\pm s$ corresponding to an angular eigenvalue $\bm\uplambda_{m}^{[s],\,(a\omega)}$ is an eigenfunction of the operator 
\begin{align}
\prod_{j=0}^{2s-1}\mc{L}^\mp_{s-j}\prod_{k=0}^{2s-1}\mc{L}^\pm_{s-k}\,, \label{eq:twice-TS-ang}
\end{align}
with indices $j,k$ increasing from right to left on the product. The eigenvalue, which is called the {\normalfont angular Teukolsky--Starobinsky constant}, $\mathfrak{B}_s=\mathfrak{B}_s(|s|,\omega,m,\bm\uplambda_{m}^{[s],\,(a\omega)})$, can be computed explicitly. 

If $\omega\in\mathbb{R}$, we will alternatively denote the angular Teukolsky--Starobinsky constant corresponding to an $s$-spin-weighted spheroidal harmonic $S_{ml}^{[s],\,(a\omega)}$ by $\mathfrak{B}_s(|s|,\omega,m,l)$, which makes the choice $\bm\uplambda_{m}^{[s],\,(a\omega)}=\bm\uplambda_{ml}^{[s],\,(a\omega)}$ explicit. Moreover, in this case, $\mathfrak B_s\geq 0$ for  integer $s$ and $\mathfrak{B}_s\leq 0$ for half-integer $s$.
\end{proposition}
\begin{proof}
Existence of the angular Teukolsky--Starobinsky constant can be shown by incrementally factoring out the differential operators and using the angular ODE \eqref{eq:angular-ode} to replace second derivatives  by first derivatives and zeroth order terms. This yields $\mathfrak{B}_s(\omega,m,l)$ explicitly and one can check that it depends on the spin only through $|s|$. The reader may refer to \cite[Sections 70 and 81]{Chandrasekhar} for a detailed proof for $|s|=1,2$, respectively. 

Recall the integration by parts lemma \cite[Section 68, Lemma 1]{Chandrasekhar} 
\begin{lemma} For $f$ and $h$ sufficiently regular functions of $\theta$,
\begin{align*}
\int_0^\pi h \mc{L}^{\pm}_n f \sin\theta d\theta = -\int_0^\pi f \mc{L}^{\mp}_{-n+1} h \sin\theta d\theta\,.
\end{align*}
\end{lemma}

Assume that $\omega$ is real. Then \eqref{eq:angular-ode} and \eqref{eq:twice-TS-ang} are real. Without loss of generality, let $S^{[\pm s]}$ be  real solutions to \eqref{eq:angular-ode} normalized to have unit $L^2$ norm. Then, by the lemma, assuming existence of the constant,
\begin{align*}
\mathfrak{B}_s &= \int_0^\pi S^{[\pm s]}\prod_{j=0}^{2s-1}\mc{L}^\mp_{s-j}\prod_{k=0}^{2s-1}\mc{L}^\pm_{s-k}S^{[\pm s]} \sin\theta d\theta \\
& = (-1)^{2s} \int_0^\pi \lp(\prod_{k=0}^{2s-1}\mc{L}^\pm_{s-k}S^{[\pm s]}\rp)^2 \sin\theta d\theta\,,
\end{align*}
where the integral on the right hand side is non-negative.
\end{proof}

\begin{proposition}\label{prop:eigenfunctions} Fix $s\in\{0,\frac12,1,\frac32,2\}$, $(\omega,\lambda)\in\mathbb{C}^2$ and an admissible frequency parameter $m$  with respect to $s$. Then, for $R^{[\pm s]}$ solutions of the radial ODE~\eqref{eq:radial-ODE} with spin $\pm s$, we set
\begin{align}
P^{[+s]}:=\Delta^s R^{[+s]}\,,\quad P^{[-s]}:=R^{[-s]}\,. \label{eq:def-P-TS}
\end{align}
$P^{[\pm s]}$ is an eigenfunction of the operator
\begin{align}
\Delta^s\lp(\mc{D}^{\mp}_0\rp)^{2s}\Delta^s\lp(\mc{D}^{\pm}_0\rp)^{2s}\,, \label{eq:twice-TS}
\end{align}
where indices $j,k$ increase from right to left in the product. The eigenvalue, which is called the {\normalfont radial Teukolsky--Starobinsky constant} $\mathfrak{C}_s$, depends explicitly \cite{Kalnins1989} on $|s|$ and the frequency triple $(\omega,m,\lambda)$ only. Moreover, if $\omega$ and $\lambda$ are real, so is $\mathfrak{C}_s$. 
\end{proposition}
\begin{proof}
Existence of the radial Teukolsky--Starobinsky constant can be shown by incrementally factoring out the differential operators and using the radial ODE \eqref{eq:TS-radial} to replace second derivatives  by first derivatives and zeroth order terms. This yields $\mathfrak{C}_s(\omega,m,\lambda)$ explicitly and one can check that it depends on the spin only through $|s|$ and that it is real when $\omega$ is real. The reader may refer to \cite[Sections 70 and 81]{Chandrasekhar} for a detailed proof for $|s|=1,2$. 
\end{proof}

\begin{remark} Though we have stated Propositions~\ref{prop:eigenfunctions-angular} and \ref{prop:eigenfunctions} only for $|s|\leq 2$, we have verified these statements for $|s|\leq 9/2$ in the manner described in the proofs and \cite{Kalnins1989}, for instance, establish Proposition \ref{prop:eigenfunctions} for $|s|\leq 7/2$. Indeed, we expect that that there is an inductive structure would that allow one to obtain these statements for general $s\in\frac12\mathbb{Z}$, though it is not yet present in the literature.
\end{remark}

\begin{remark} \label{rmk:TS-constant-sign} In Sections \ref{sec:superradiance-s} and \ref{sec:scattering}, it will be of interest to understand the sign of $\mathfrak{C}_s$ for $\omega\in\mathbb{R}$. In contrast with the angular case, it is not clear that an integration by parts lemma can resolve this question. One hope would be to instead relate the angular and radial Teukolsky--Starobinsky constants, when we assume that the radial ODE arises specifically from a separation of variables, i.e.\ when we let $\lambda=\bm\uplambda_{ml}^{[s],\,(a\omega)}$ (see Proposition \ref{def:angular-ode}). In this case, in general, one can show that $\mathfrak{B}_s=(-1)^{2s}\mathfrak{C}_s|_{M=0}$ \cite{Kalnins1992}, but this is not enough to determine the sign of $\mathfrak{C}_s$. By explicit computation, however, we can also show that 
\begin{gather}
\begin{gathered}
\mathfrak{C}_{1/2}=-\mathfrak{B}_{1/2}\geq 0\,,\quad \mathfrak{C}_{1}=\mathfrak{B}_{1}\geq 0\,, \quad \mathfrak{C}_{3/2}=-\mathfrak{B}_{3/2}\geq 0\,,\\ \mathfrak{C}_2=\mathfrak{B}_2+144M^2\omega^2\geq 0\,.
\end{gathered}\label{eq:TS-constant-sign}
\end{gather}
Were there frequencies for which there exist nontrivial solutions of the radial ODE~\eqref{eq:radial-ODE} with boundary conditions as in \cref{def:mode-solution}, the inequalities in \eqref{eq:TS-constant-sign} would be strict, as we will show in \cref{lemma:algebraically-special}. However, it is worth pointing out that this is already the case for $\mathfrak{C}_2$: $\mathfrak{C}_2\geq 144M^2\omega^2$ could only vanish for $\omega=0$, but in this case, using the explicit form of the angular eigenvalues, we find that $\mathfrak{C}_2(\omega=0)>0$.

For higher spins, $\mathfrak{C}_s-(-1)^{2s}\mathfrak{B}_s$ will depend on $\bm\uplambda_{ml}^{[s],\,(a\omega)}$ and we cannot expect to prove $\mathfrak{C}_s-(-1)^{2s}\mathfrak{B}_s\geq 0$ without appealing to specific properties of the angular eigenvalues.
\end{remark}

\begin{remark} So as to apply for any $s\in\frac12\mathbb{Z}$, our proofs of Theorems~\ref{thm:mode-stability-subextremal}, \ref{thm:mode-stability-full} and \ref{thm:quantitative-intro} will not rely on Propositions~\ref{prop:eigenfunctions-angular} and \ref{prop:eigenfunctions}. 
These results require the existence only of the radial Teukolsky--Starobinsky constant; however, as we will always be dealing with radial functions with prescribed outgoing boundary conditions for the homogeneous radial ODE~\eqref{eq:radial-ODE} in these theorems, we will define $\mathfrak{C}_s$ using the Teukolsky--Starobinsky identities (see \cref{prop:TS-radial}) from the next section  to avoid appealing to \cref{prop:eigenfunctions} (see Remark~\ref{rmk:def-radial-TS-const}).

In \cref{sec:superradiance-s} and in our application to scattering theory (Corollary~\ref{cor:scattering-intro}), we will be working with solutions to the homogeneous radial ODE~\eqref{eq:radial-ODE} with more general asymptotics. In these two instances, we will have to define the Teukolsky--Starobinsky constants by Propositions \ref{prop:eigenfunctions-angular} and \ref{prop:eigenfunctions}. In \cref{sec:superradiance-s}, we moreover appeal to Remark~\ref{rmk:TS-constant-sign} for information on the the sign of $\mathfrak{C}_s$ when $\lambda$ arises from a separation of variables and $\omega$ is real. 
\end{remark}

\subsubsection{The radial Teukolsky--Starobinsky identities}

The radial Teukolsky--Starobinsky identities are differential identities, which relate solutions of the homogeneous radial ODE~\eqref{eq:radial-ODE} with spin $+s$ and spin $-s$, were introduced in \cite{Teukolsky1974, Starobinsky1974} and extended for general $s$ in \cite{Kalnins1989}. We will present a rigorous statement:
\begin{proposition}[Radial Teukolsky--Starobinsky identities] \label{prop:TS-radial} Fix $M>0$, $|a|\leq M$ and $s\in\frac12\mathbb{Z}_{\geq 0}$. Let $(\omega,m,\lambda)$ be an admissible frequency triple with respect to $s$ and $a$. Dropping most subscripts, let $R^{[\pm s]}$ be solutions to the homogeneous radial ODE~\eqref{eq:radial-ODE} of spin $\pm s$. Consider the representation \eqref{eq:R-general-asymptotics} 
\begin{equation*}
\begin{split}
\swei{R}{\pm s}&= \frac{1}{2Mr_+}\swei{a}{\pm s}_{\mc{H}^+}\swei{R}{\pm s}_{\mc{H}^+}+\frac{1}{2Mr_+}\swei{a}{\pm s}_{\mc{H}^-}\swei{R}{\pm s}_{\mc{H}^-}= \swei{a}{\pm s}_{\mc{I}^+}\swei{R}{\pm s}_{\mc{I}^+}+\swei{a}{\pm s}_{\mc{I}^-}\swei{R}{\pm s}_{\mc{I}^-}\,,
\end{split}
\end{equation*}
if additionally $\omega\neq m\upomega_+$ or, if the solution is outgoing,
\begin{equation*}
\begin{split}
\swei{R}{\pm s}&= \frac{1}{2Mr_+}\swei{a}{\pm s}_{\mc{H}^+}\swei{R}{\pm s}_{\mc{H}^+}= \swei{a}{\pm s}_{\mc{I}^+}\swei{R}{\pm s}_{\mc{I}^+}\,,
\end{split}
\end{equation*}
for some complex $\swei{a}{\pm s}_{\mc{H}^{+}}$, $\swei{a}{\pm s}_{\mc{H}^{-}}$, $\swei{a}{\pm s}_{\mc{I}^{+}}$ and $\swei{a}{\pm s}_{\mc{H}^{-}}$. Then we have
\begin{align}
\begin{split}
\Delta^s \lp(\mc{D}_0^{+}\rp)^{2s}\lp(\Delta^s \swei{R}{+s}\rp)&=\mathfrak{C}_s^{(1)}\swei{a}{+s}_{\mc{I}^{+}}\swei{R}{-s}_{\mc{I}^+}+C_s^{(7)}\swei{a}{+s}_{\mc{I}^{-}}\swei{R}{-s}_{\mc{I}^-}\\
&=\mathfrak{C}_s^{(4)}\frac{1}{2Mr_+}\swei{a}{+s}_{\mc{H}^{+}}\swei{R}{-s}_{\mc{H}^+}+{\mathfrak{C}_s^{(6)}}\frac{1}{2Mr_+}\swei{a}{+s}_{\mc{H}^{-}}\swei{R}{-s}_{\mc{H}^-}\,, \\
\Delta^s \lp(\mc{D}_0^{-}\rp)^{2s}\swei{R}{-s}&=\mathfrak{C}_s^{(3)}\swei{a}{-s}_{\mc{I}^{+}}\Delta^s\swei{R}{+s}_{\mc{I}^+}+{\mathfrak{C}_s^{(5)}}\swei{a}{+s}_{\mc{I}^{-}}\Delta^s\swei{R}{+s}_{\mc{I}^-}\\
&=\mathfrak{C}_s^{(2)}\frac{1}{2Mr_+}\swei{a}{-s}_{\mc{H}^{+}}\Delta^s\swei{R}{+s}_{\mc{H}^+}+\mathfrak{C}_s^{(8)}\frac{1}{2Mr_+}\swei{a}{-s}_{\mc{H}^{-}}\Delta^s\swei{R}{+s}_{\mc{H}^-}\,, 
\end{split}\label{eq:TS-radial}
\end{align}
for some $\mathfrak{C}_s^{(i)}\in\mathbb{C}$, $i=1,...,8$, which depend only on $s$ and $(\omega,m,\lambda)$ and such that $\mathfrak C_0^{(i)}=1$,  $\mathfrak{C}_s^{(1)}=(2i\omega)^{2s}$, ${\mathfrak{C}_s^{(5)}}=(-2i\omega)^{2s}$ and 
\begin{gather*}
\mathfrak{C}_s^{(2)}=
\begin{dcases}
\prod_{j=0}^{2s-1}\lp(2\xi+s-j\rp) &\text{if~}|a|<M \\
(-2\beta)^{2s} &\text{if~}|a|=M
\end{dcases}\,,\quad  \mathfrak{C}_s^{(6)}=
\begin{dcases}
\prod_{j=0}^{2s-1}\lp(-2\xi+s-j\rp) &\text{if~}|a|<M \\
(2\beta)^{2s} &\text{if~}|a|=M
\end{dcases} \,.
\end{gather*}
\end{proposition}

\begin{remark}[Definition of the radial Teukolsky--Starobinsky constant] \label{rmk:def-radial-TS-const}
Assume that $R^{[\pm s]}$ are outgoing solutions of the homogeneous radial ODE, i.e. $\swei{a}{\pm}_{\mc{H}^-}=\swei{a}{\pm}_{\mc{I}^-}=0$ in \eqref{eq:R-general-asymptotics}. Then, \eqref{eq:TS-radial} provide an {\normalfont alternative definition} of the radial Teukolsky--Starobinsky constant from \cref{prop:eigenfunctions} as 
\begin{equation}
\mathfrak{C}_s:=\mathfrak{C}_s^{(1)}\mathfrak{C}_s^{(3)}=\mathfrak{C}_s^{(2)}\mathfrak{C}_s^{(4)}\,. \label{eq:redef-Cs}
\end{equation}

If $R^{[\pm s]}$ are not outgoing, we consider the radial Teukolsky--Starobinsky constant as defined in \cref{prop:eigenfunctions} and conclude 
\begin{align*}
\mathfrak{C}_s^{(1)}\mathfrak{C}_s^{(3)}=\mathfrak{C}_s^{(2)}\mathfrak{C}_s^{(4)}={\mathfrak{C}_s^{(5)}}{\mathfrak{C}_s^{(7)}}={\mathfrak{C}_s^{(6)}}{\mathfrak{C}_s^{(8)}}=\mathfrak{C}_s\,.
\end{align*}

Having defined the radial Teukolsky--Starobinsky constant, we can write \eqref{eq:TS-radial} in terms of this constant and the $\mathfrak{C}_s^{(i)}$ for which we have simple formulas in terms of the frequency parameters:
\begin{align}
\begin{split}
\Delta^s \lp(\mc{D}_0^{+}\rp)^{2s}\lp(\Delta^s \swei{R}{+s}\rp)&=\mathfrak{C}_s^{(1)}\swei{a}{+s}_{\mc{I}^{+}}\swei{R}{-s}_{\mc{I}^+}+\mathfrak{C}_s\lp({\mathfrak{C}_s^{(5)}}\rp)^{-1}\swei{a}{+s}_{\mc{I}^{-}}\swei{R}{-s}_{\mc{I}^-}\\
&=\mathfrak{C}_s\lp(\mathfrak{C}_s^{(2)}\rp)^{-1}\frac{1}{2Mr_+}\swei{a}{+s}_{\mc{H}^{+}}\swei{R}{-s}_{\mc{H}^+}+{\mathfrak{C}_s^{(6)}}\frac{1}{2Mr_+}\swei{a}{+s}_{\mc{H}^{-}}\swei{R}{-s}_{\mc{H}^-}\,, \\
\Delta^s \lp(\mc{D}_0^{-}\rp)^{2s}\swei{R}{-s}&=\mathfrak{C}_s\lp(\mathfrak{C}_s^{(1)}\rp)^{-1}\swei{a}{-s}_{\mc{I}^{+}}\Delta^s\swei{R}{+s}_{\mc{I}^+}+{\mathfrak{C}_s^{(5)}}\swei{a}{+s}_{\mc{I}^{-}}\Delta^s\swei{R}{+s}_{\mc{I}^-}\\
&=\mathfrak{C}_s^{(2)}\frac{1}{2Mr_+}\swei{a}{-s}_{\mc{H}^{+}}\Delta^s\swei{R}{+s}_{\mc{H}^+}+\mathfrak{C}_s\lp({\mathfrak{C}_s^{(6)}}\rp)^{-1}\frac{1}{2Mr_+}\swei{a}{-s}_{\mc{H}^{-}}\Delta^s\swei{R}{+s}_{\mc{H}^-}\,. 
\end{split}\label{eq:TS-radial-2}
\end{align}
\end{remark}

In order to prove Proposition~\ref{prop:TS-radial}, we recall the following two results from \cite{Kalnins1989}:
\begin{lemma} If $P^{[\pm s]}$, $s>0$, are as defined in \cref{prop:TS-radial},  they satisfy the homogeneous radial ODE
\begin{align}
\label{eq:radial-ODE-TS} \lp[\Delta \mc{D}_{1-s}^{\mp} \mc{D}_0^{\pm} \pm 2(2s-1)i\omega r\rp]P^{[\pm s]}(r)=\lp(\lambda^{[s]}+s+a^2\omega^2-2am\omega+|s|\rp)P^{[\pm s]}(r)\,.
\end{align}
\end{lemma}

By induction on $s$, one can also show that:

\begin{lemma}\label{lemma:TS-commutation} The operators $\mc{D}^{\pm}_0$, defined in \eqref{eq:def-D-pm}, satisfy the commutation relation
\begin{align*}
\Delta^s \lp(\mc{D}^{\pm}_0\rp)^{2s} \lp[\Delta \mc{D}^{\pm}_{1-s} \mc{D}^{\mp}_0 \mp 2(2s-1)i\omega r\rp]
= \lp[\Delta \mc{D}^{\mp}_{1-s} \mc{D}^{\pm}_0 \pm 2(2s-1)i\omega r\rp]\Delta^s \lp(\mc{D}^{\mp}_0\rp)^{2s}\,.
\end{align*}
\end{lemma}

Combining the two lemmas, we see that $\Delta^s\lp(\mc{D}_0^{\pm}\rp)^{2s}P^{[\pm s]}$ is a solution of the homogeneous radial ODE (\ref{eq:radial-ODE-TS}) for spin  $\mp s$. We are now ready to prove Proposition~\ref{prop:TS-radial}:

\begin{proof}[Proof of Proposition~\ref{prop:TS-radial}]
The statement is trivial for $s=0$, since the map referred is just the identity. For $s\neq 0$, we obtain the result by using the asymptotic representations (see \cite[Chapters 5 and 7]{Olver1973} and \cite[Appendix A]{Shlapentokh-Rothman2014}) of $\swei{R}{\pm s }_{\mc{H}^+}$, $\swei{R}{\pm s}_{\mc{I}^+}$, $\swei{R}{\pm s }_{\mc{H}^-}$ and  $\swei{R}{\pm s}_{\mc{I}^-}$, 
\begin{align}
\begin{split} \label{eq:model-solutions-series}
\swei{R}{\pm s }_{\mc{I}^+}&=e^{i\omega r} r^{2iM\omega\mp 2s-1}\lp[\sum_{k=0}^{2s} \swei{c}{\pm s}_{k} r^{-k}+O\lp(r^{-2s-1}\rp)\rp]\,,\\
\swei{R}{\pm s }_{\mc{I}^-}&=e^{-i\omega r} r^{-2iM\omega-1}\lp[\sum_{k=0}^{2s} \swei{c}{\pm s}_{k,2} r^{-k}+O\lp(r^{-2s-1}\rp)\rp]\,,\\
\swei{R}{\pm s }_{\mc{H}^+}&=\begin{cases}
\sum_{k=0}^\infty \swei{b}{\pm s}_{k} (r-r_+)^{\xi\mp s} &\text{if~} |a|<M\\
\sum_{k=0}^\infty \swei{b}{\pm s}_{k} e^{\beta(r-M)^{-1}}(r-M)^{2iM\omega-2s}&\text{if~} |a|=M
\end{cases}\,,\\
\swei{R}{\pm s }_{\mc{H}^-}&=\begin{cases}
\sum_{k=0}^\infty \swei{b}{\pm s}_{k,2} (r-r_+)^{-\xi} &\text{if~} |a|<M\\
\sum_{k=0}^\infty \swei{b}{\pm s}_{k,2} e^{-\beta(r-M)^{-1}}(r-M)^{-2iM\omega}&\text{if~} |a|=M
\end{cases}\,,
\end{split}
\end{align}
 on which we can act directly with the Teukolsky--Starobinsky operators. Without loss of generality, we assume $\swei{c}{+s}_0=\swei{c}{-s}_0$, $\swei{c}{+s}_{0,2}=\swei{c}{-s}_{0,2}$, $\swei{b}{+s}_0=\swei{b}{-s}_0$ and $\swei{b}{+s}_{0,2}=\swei{b}{-s}_{0,2}$

For computations at $r=r_+$, it will be use to recall from the definition of $\xi$ and $\beta$ in (\ref{eq:xi-upomega+}),
\begin{align}
\begin{split}
&\xi(r-r_-)+i\omega(r^2+a^2)-iam \\
&\qquad= \xi(r_+-r_-)+i\omega(r_+^2+a^2)-iam +\xi(r-r_+)+i\omega(r-r_+)(r+r_+)  \\ &\qquad= \xi(r-r_+)+i\omega(r-r_+)(r+r_+)\,,
\end{split}\label{eq:algebraic-cancellation-xi}\\
&\beta+i\omega(r-M)(r+M)=i\omega(r^2+M^2)-iam \,. \label{eq:algebraic-cancellation-beta}
\end{align}

\medskip
\noindent \textit{The cases $\swei{R}{+s}_{\mc{I}^+}$ and $\swei{R}{-s}_{\mc{H}^+}$}. In these cases, the constant arising from application of the Teukolsky--Starobinsky operators to these functions involves only the first term of the series.

We begin at $r\to\infty$, with positive spin. For any $s\geq 1/2$, we will have
\begin{align*}
&\mc{D}_0^{+}\lp(\Delta^s R^{[+s]}_{\mc{I}^+}\rp)=\lp(\frac{d}{dr}+i\omega\frac{r^2+a^2}{\Delta}-\frac{iam}{\Delta}\rp)\lp(\Delta^s R^{[+s]}\rp) \\
&\quad=e^{i\omega r}r^{2iM\omega-1}\lp[\sum_{k=0}^{2s} \swei{c}{+s}_k \lp(2i\omega-\frac{k+1-2iM\omega}{r}+\frac{2Mi\omega r}{\Delta}-\frac{iam}{\Delta}\rp)r^{-k}+O\lp(r^{-2s-1}\rp)\rp]\,.
\end{align*}
Repeating $2s-1$ times, we obtain
\begin{align}
&\Delta^{s}\lp(\mc{D}_0^{+}\rp)^{2s}\lp(\Delta^s R^{[+s]}_{\mc{I}^+}\rp) \nonumber\\
&\quad = e^{i\omega r}r^{2iM\omega+2s-1}\lp[\sum_{k=0}^{2s} \swei{c}{+s}_k \lp[(2i\omega)^{2s}+O\lp(r^{-1}\rp)\rp] r^{-k}+O\lp(r^{-2s-1}\rp)\rp]\,. \label{eq:intermediate-1}
\end{align}
By Lemma \ref{lemma:TS-commutation}, \eqref{eq:intermediate-1} satisfies the spin $-s$ radial ODE, hence, by Definition~\ref{def:uhor-uout}, we conclude
\begin{align*}
&\Delta^{s}\lp(\mc{D}_0^{+}\rp)^{2s}\lp(\Delta^s R^{[+s]}_{\mc{I}^+}\rp)=\mathfrak{C}_{s}^{(1)} R^{[-s]}_{\mc{I}^+}\,, \qquad \mathfrak{C}_s^{(1)}:=(2i\omega)^{2s}\,.
\end{align*}

We now turn to $r=r_+$, and consider solutions with spin $-s$. We start with $|a|<M$; for any $s\geq 1/2$, using \eqref{eq:algebraic-cancellation-xi}, we obtain
\begin{align*}
&\mc{D}_0^{-}\swei{R}{-s}_{\mc{H}^+} = \lp(\frac{d}{dr}-i\omega\frac{r^2+a^2}{\Delta}+\frac{iam}{\Delta}\rp)\swei{R}{-s}
\\
&\quad = \sum_{k=0}^\infty \swei{b}{-s}_k\lp[\lp(\xi+k+s\rp)-\frac{i\omega(r^2+a^2)}{r-r_-}+\frac{iam}{r-r_-}\rp](r-r_+)^{\xi+k+s-1} \\
&\quad =\sum_{k=0}^\infty \swei{b}{-s}_k\lp[\lp(2\xi+k+s\rp)-\frac{(\xi+i\omega(r+r_+))(r-r_+)}{r-r_-}\rp](r-r_+)^{\xi+k+s-1}\,,
\end{align*}
which, repeating $2s-1$ times, yields
\begin{align}
&\Delta^s\lp(\mc{D}_0^{-}\rp)^{2s}\swei{R}{-s}_{\mc{H}^+} =\sum_{k=0}^\infty \swei{b}{-s}_k\lp[\prod_{j=0}^{2s-1}\lp(2\xi+k+s-j\rp)+O(r-r_+)\rp]\Delta^{s}(r-r_+)^{\xi+k-s} \,.\label{eq:intermediate-2}
\end{align}
On the other hand, if $|a|=M$, then for any $s\geq 1/2$, using \eqref{eq:algebraic-cancellation-beta}, we obtain
\begin{align*}
\mc{D}_0^{-}R^{[-s]}_{\mc{H}^+} &= \lp(\frac{d}{dr}-i\omega\frac{r^2+M^2}{(r-M)^2}+\frac{iam}{(r-M)^2}\rp)R^{[-s]}
\\
& = \lp(\frac{d}{dr}-\frac{\beta}{(r-M)^2}-\frac{(r+M)i\omega}{(r-M)}\rp)e^{\beta(r-M)^{-1}}\sum_{k=0}^\infty \swei{b}{-s}_k(r-M)^{-2iM\omega+2s+k} \\
& = \lp[-2\beta+O(r-M)\rp]  e^{\beta(r-M)^{-1}}\sum_{k=0}^\infty \swei{b}{-s}_k(r-M)^{-2iM\omega+2s-2+k}\,,
\end{align*}
which, repeating $2s-1$ times, yields
\begin{align}
&\Delta^s\lp(\mc{D}_0^{-}\rp)^{2s}\swei{R}{-s}_{\mc{H}^+}=\sum_{k=0}^\infty \swei{b}{-s}_k\lp[\lp(-2\beta\rp)^{2s}+O(r-M)\rp]e^{\beta(r-M)^{-1}}(r-M)^{-2iM\omega+k}\,. \label{eq:intermediate-3}
\end{align}

By Lemma \ref{lemma:TS-commutation}, \eqref{eq:intermediate-2} and \eqref{eq:intermediate-3} satisfy the homogeneous spin $+s$ radial ODE \eqref{eq:radial-ODE-TS} for $|a|<M$ and $|a|=M$, respectively, hence, by Definition~\ref{def:uhor-uout}, we conclude
\begin{align*}
&\Delta^s\lp(\mc{D}_0^{-}\rp)^{2s}\swei{R}{-s}_{\mc{H}^+} =\mathfrak{C}_s^{(2)}\Delta^s\swei{R}{+s}_{\mc{H}^+}\,,
\end{align*}
where we have defined
\begin{align*}
\mathfrak{C}_s^{(2)}:=\prod_{j=0}^{2s-1}\lp(2\xi+k+s-j\rp) \text{~if~}|a|<M\,,\qquad \mathfrak{C}_s^{(2)}:=(-2\beta)^{2s} \text{~if~}|a|=M\,.
\end{align*}

\medskip
\noindent \textit{The cases $\swei{R}{+s}_{\mc{H}^+}$ and $\swei{R}{-s}_{\mc{I}^+}$}.  We note that, using \cref{prop:eigenfunctions}, the result for outgoing solutions would now follow. If we do not wish to use \cref{prop:eigenfunctions}, however, we need to take the same series-based approach as in the previous case. As we will see, for $\swei{R}{+s}_{\mc{H}^+}$ and $\swei{R}{-s}_{\mc{I}^+}$, the constant arising from application of the Teukolsky--Starobinsky operators to the series for these functions is more complicated to compute, as it involves up to $2s$ terms of their series expansion.

We begin with spin $-s$ at $r=\infty$. For any $s\geq 1/2$, we have 
\begin{align*}
\mc{D}_0^{-}R^{[-s]}_{\mc{I}^+} &= \lp(\frac{d}{dr}-i\omega\frac{r^2+a^2}{\Delta}+\frac{iam}{\Delta}\rp)\lp[e^{i\omega r} r^{2iM\omega+2s-1}\sum_{k=0}^{2s} \swei{c}{-s}_k r^{-k} +O\lp(r^{-2s-1}\rp)\rp]\\
&=e^{i\omega r} r^{2iM\omega+2s-1}\lp(2iM\omega\frac{a^2-2Mr}{r\Delta}+\frac{iam}{\Delta}+\frac{2s-1-k}{r}\rp)\sum_{k=0}^{2s} \swei{c}{-s}_k r^{-k}\,.\numberthis \label{eq:s-negative-hor}
\end{align*}
The key difficulty, which is not present in the case of spin $+s$, is the cancellation that occurs between the $\omega$ dependence of $\mc{D}_0^-$ and that of the leading order behavior of $R^{[-s]}_{\mc{I}^+}$. For instance, if $s=1/2$, then 
\begin{align*}
\mc{D}_0^{-}R^{[-1/2]}_{\mc{I}^+} &= e^{i\omega r} r^{2iM\omega}\lp[\sum_{k=0}^1\lp(2iM\omega\frac{a^2-2Mr}{r\Delta}+\frac{iam}{\Delta}-\frac{k}{r}\rp)\swei{c}{-1/2}_k r^{-k}+O(r^{-4})\rp]\\
&=\lp[(-4iM^2\omega+iam)\swei{c}{-s}_0-\swei{c}{-1/2}_1+O(r^{-1})\rp] e^{i\omega r} r^{2iM\omega-2}\,,
\end{align*}
so taking $\mc{D}_0^{-}$ lowers the leading order behavior by a factor of $r^{-2}$. Indeed, an application of $(\mc{D}_0^{-})^{2s}$ lowers the leading order polynomial decay of $R^{[-s]}_{\mc{I}^+}$ by a factor of $r^{-4s}$ and the coefficient of the leading order term on the right hand side depends on coefficients $\swei{c}{-s}_{0},...,\swei{c}{-s}_{2s}$ of the asymptotic expansion for $R^{[-s]}_{\mc{I}^+}$:
\begin{align}
\Delta^s\lp(\mc{D}_0^{-}\rp)^{2s}R^{[-s]}_{\mc{I}^+} &=\Delta^s e^{i\omega r} r^{2iM\omega-2s-1}\lp[f_+\lp(\swei{c}{-s}_{0},...,\swei{c}{-s}_{2s},\omega,m\rp)+O(r^{-1})\rp]   \label{eq:intermediate-4}
\end{align}
where $f_+$ is independent of $r$. By standard theory of asymptotic analysis \cite{Olver1973}, $\swei{c}{-s}_{1},...,\swei{c}{-s}_{2s}$ are uniquely determined by $\swei{c}{-s}_{0}$ and the ODE parameters $(\omega,m,\lambda,s)$, so we can define
\begin{align*}
\mathfrak C_{s}^{(3)}(\omega,m,\lambda,s) := f_+\lp(\swei{c}{-s}_{0},...,\swei{c}{-s}_{2s},\omega,m\rp)/\swei{c}{-s}_{0}\,.
\end{align*}
By Lemma \ref{lemma:TS-commutation}, \eqref{eq:intermediate-4} satisfies the spin $-s$ radial ODE. As before, by Definition~\ref{def:uhor-uout}, we conclude
\begin{align*}
\Delta^s\lp(\mc{D}_0^{-}\rp)^{2s}R^{[-s]}_{\mc{I}^+} &=\mathfrak{C}_s^{(3)}\Delta^s R^{[+s]}_{\mc{I}^+} \,.
\end{align*}

Now consider spin $+s$ near $r=r_+$, and recall \eqref{eq:algebraic-cancellation-xi} and \eqref{eq:algebraic-cancellation-beta}. If $|a|<M$, then 
\begin{align*}
\mc{D}_0^{+}\lp(\Delta^s\swei{R}{+s}_{\mc{H}^+}\rp) &= \lp(\frac{d}{dr}+i\omega\frac{r+r_+}{r-r_+}-\frac{\xi}{r-r_+}+\frac{\xi}{r-r_-}\rp)\sum_{k=0}^\infty\swei{b}{+s}_k (r-r_+)^{\xi+k}(r-r_-)^s
\\
&= \sum_{k=0}^\infty\swei{b}{+s}_k\lp[\frac{k}{r-r_+}+\frac{\xi+s +i\omega (r+r_+)}{r-r_-}\rp] (r-r_+)^{\xi+k}(r-r_-)^s\,,
\end{align*}
and if $|a|=M$,
\begin{align*}
\mc{D}_0^{+}\lp(\Delta^s\swei{R}{+s}_{\mc{H}^+}\rp) &= \lp(\frac{d}{dr}+\frac{\beta}{(r-M)^2}+\frac{2Mi\omega}{r-M}+i\omega\rp)\sum_{k=0}^\infty\swei{b}{+s}_k e^{\beta(r-M)^{-1}}(r-M)^{-2iM\omega+k}
\\
&= \sum_{k=0}^\infty\swei{b}{+s}_k\lp[\frac{k}{r-M}+i\omega\rp]  e^{\beta(r-M)^{-1}}(r-M)^{-2iM\omega+k}\,.
\end{align*}
As in the previous case, there is cancellation between the $\xi$ or $\beta$ dependence of the Teukolsky--Starobinsky operators and that of the leading order behavior of $\Delta^s\swei{R}{+s}_{\mc{H}^+}$. Proceeding as above, we note that application of $(\mc{D}_0^{-})^{2s}$ does not change the leading order polynomial decay of $R^{[+s]}_{\mc{H}^+}$ as $r\to r_+$, but makes the coefficient of the resulting leading order term depend on coefficients $\swei{c}{-s}_{0},...,\swei{c}{-s}_{2s}$ of the asymptotic series for $R^{[+s]}_{\mc{H}^+}$. We again appeal to theory of asymptotic analysis \cite{Olver1973} to define $\mathfrak C_{s}^{(4)}(\omega,m,\lambda,s)$ such that
\begin{align*}
\Delta^s\lp(\mc{D}_0^{+}\rp)^{2s}\lp(\Delta^s R^{[+s]}_{\mc{H}^+}\rp) &=\mathfrak{C}_s^{(4)} R^{[-s]}_{\mc{H}^+} \,.
\end{align*}

\medskip
\noindent \textit{The cases $\swei{R}{\pm s}_{\mc{H}^-}$ and $\swei{R}{\pm s}_{\mc{I}^-}$}. Finally by direct inspection of \eqref{eq:model-solutions-series}, for $\swei{R}{\pm s}_{\mc{H}^-}$ and $\swei{R}{\pm s}_{\mc{I}^-}$, we can easily obtain the Teukolsky--Starobinsky identities by analogy with $\swei{R}{\mp s}_{\mc{H}^+}$ and $\swei{R}{\mp s}_{\mc{I}^+}$, respectively.
\end{proof}

\subsubsection{Algebraically special frequencies}

In this section, we consider frequencies for which $\mathfrak{C}_s=0$.

\begin{definition} \label{def:algebraically-special} Fix $M>0$, $|a|\leq M$ and $s\in\mathbb{Z}$. We say an admissible frequency triple $(\omega,m,\lambda)$ is algebraically special if the radial Teukolsky--Starobinsky constant, $\mathfrak{C}_s(\omega,m,\lambda)$, vanishes.
\end{definition}

\begin{lemma} \label{lemma:algebraically-special}
Fix $M>0$, $|a|\leq M$ and $s\in\mathbb{Z}$. Let $(\omega,m,\lambda)$ be an admissible frequency triple which is algebraically special. Let $\swei{R}{\pm s}$ be solutions to  the homogeneous radial ODE \eqref{eq:radial-ODE} of spin $\pm s$.
\begin{enumerate}
\item If $\swei{R}{+s}=\frac{1}{2Mr_+}\swei{a}{+s}_{\mc{H}^+}\swei{R}{+s}_{\mc{H}^+}$, then we have
\begin{align*}
\swei{R}{+s}=\swei{a}{+s}_{\mc{I}^-}\swei{R}{+s}_{\mc{I}^-}=\frac{1}{2Mr_+}\swei{a}{+s}_{\mc{H}^+}\swei{R}{+s}_{\mc{H}^+}\,.
\end{align*}
\item If $\swei{R}{-s}=\swei{a}{-s}_{\mc{I}^+}\swei{R}{-s}_{\mc{I}^+}$, then $\swei{R}{-s}$ is trivial or, if $\omega\neq m\upomega_+$, we have
\begin{align*}
\swei{R}{-s}=\frac{1}{2Mr_+}\swei{a}{-s}_{\mc{H}^-}\swei{R}{-s}_{\mc{H}^-}=\swei{a}{-s}_{\mc{I}^+}\swei{R}{-s}_{\mc{I}^+}\,.
\end{align*}
\end{enumerate}
\end{lemma}
\begin{proof}
The result follows by direct inspection of \eqref{eq:TS-radial-2}.
\end{proof}

A simple corollary of Lemma~\ref{lemma:algebraically-special} is that, for algebraically special frequencies, there are no nontrivial outgoing solutions to the radial ODE~\eqref{eq:radial-ODE}, hence

\begin{lemma}\label{lemma:TS-radial-boundary-conditions}
The Teukolsky--Starobinsky identities \eqref{eq:TS-radial} map nontrivial outgoing solutions of the radial ODE \eqref{eq:radial-ODE} of spin $+s$ to nontrivial outgoing solutions of the radial ODE \eqref{eq:radial-ODE} of spin $-s$ and vice-versa (see Definition~\ref{def:outgoing-radial-solution}).
\end{lemma}

\subsection{Energy and superradiance for the radial ODE}
\label{sec:superradiance}

In this section, we will present currents at the level of the radial ODE~\eqref{eq:radial-ODE} which give a notion of energy for real $\omega$ (see Propositions~\ref{prop:T-identity} and \ref{prop:wronskian-identity}). For generality, we will work with solutions of the homogeneous radial ODE~\eqref{eq:u-Schrodinger} with the general asymptotics \eqref{eq:R-general-asymptotics}
\begin{equation}
\begin{split}
\swei{u}{s}&= \swei{a}{s}_{\mc{H}^+}\swei{u}{s}_{\mc{H}^+}+\swei{a}{s}_{\mc{H}^-}\swei{u}{s}_{\mc{H}^-}= \swei{a}{s}_{\mc{I}^+}\swei{u}{s}_{\mc{I}^+}+\swei{a}{s}_{\mc{I}^-}\swei{u}{-s}_{\mc{I}^-}\,,
\end{split}\label{eq:u-general-asymptotics}
\end{equation}
for some complex $\swei{a}{s}_{\mc{H}^\pm}$ and $\swei{a}{s}_{\mc{I}^\pm}$.

For outgoing solutions (compatible with Definition~\ref{def:outgoing-radial-solution}) of the homogeneous radial ODE~\eqref{eq:u-Schrodinger}, we will moreover show that the main obstacle in proving mode stability via the argument for $a=0$ is the presence of a region, including $\mc{H}^\pm$, where $\p_t$ becomes spacelike, leading to superradiance. The integral transformations we will eventually present in \cref{sec:integral-transformation} are designed to specifically circumvent this issue, making \textit{all} frequencies non-superradiant.

\subsubsection{The case \texorpdfstring{$s=0$}{of the wave equation}}
\label{sec:superradiance-wave}

Define the frequency-localized current generated by the Killing vector field $T$ as
\begin{equation}
Q^T[u]:=\Im(u'\overline{\omega u})\,, \quad -\lp(Q^T[u]\rp)' = \Im(\omega)\lp[|u'|^2+|\omega|^2 |u|^2\rp] +\Im (\overline{V}\omega) |u|^2\,, \label{eq:T-current}
\end{equation}
where the last equality is obtained using \eqref{eq:u-Schrodinger}. When $\omega\in\mathbb{R}$, \eqref{eq:T-current} yields the following conservation law
\begin{proposition}[Energy identity for the homogeneous radial ODE, $s=0$] \label{prop:T-identity}
Fix $M>0$ and $|a|\leq M$. Let $(\omega,m,\lambda)$ be an admissible frequency triple with respect to $a$ such that $\omega$ is real. Let $u$ be a solution to \eqref{eq:u-Schrodinger} with $s=0$ given by \eqref{eq:u-general-asymptotics}. Then, dropping the superscripts,
\begin{equation}
\omega^2\lp|a_{\mc{I}^+}\rp|^2+\omega(\omega-m\upomega_+)\lp|a_{\mc{H}^+}\rp|^2=\omega^2\lp|a_{\mc{I}^-}\rp|^2+\omega(\omega-m\upomega_+)\lp|a_{\mc{H}^-}\rp|^2 \label{eq:energy-identity-wave}
\end{equation}
holds for a general $u$ if $\omega\neq m\upomega_+$. The same identity with vanishing right hand side holds for an outgoing, in the sense of Definition~\ref{def:outgoing-radial-solution}, solution $u$.  
\end{proposition}
\begin{proof}
The proof follows by applying the fundamental theorem of calculus to \eqref{eq:T-current} and evaluating $Q^T[u](\pm \infty)$ using the asymptotic behavior of $u$ (recall Definition~\ref{def:uhor-uout}).
\end{proof}

We will now recall why such a current can be used to show Theorems~\ref{thm:mode-stability-subextremal} and \ref{thm:mode-stability-full} in the case $a=0$ but fails to do so whenever $|a|>0$.

Consider an outgoing solution of the homogeneous radial ODE~\eqref{eq:u-Schrodinger} with $\omega$ real. If $a=0$, \eqref{eq:energy-identity-wave} shows that $u(\pm \infty)=0$ and we can apply a unique continuation result (see, for instance, Lemma~\ref{lemma:unique-continuation}) or Lemma~\ref{lemma:R-general-asymptotics} to conclude that $u\equiv 0$. Thus, Theorem~\eqref{thm:mode-stability-real-axis} certainly holds for $a=0$. However, when $|a|>0$, if $\omega\in\mathbb{R}$ and $m$ are \textit{superradiant}, that is
\begin{equation}
\omega(\omega-m\upomega_+)\leq 0\,,\label{eq:superradiant-freqs}
\end{equation}
then \eqref{eq:energy-identity-wave} fails to give an estimate for $|u(\pm \infty)|^2$, so we can no longer infer mode stability.

Now consider $\Im\omega>0$ and $\lambda=\bm\uplambda$ from Proposition~\ref{def:angular-ode}; then $(Q^T)'\neq 0$. On the other hand, the boundary conditions for the radial ODE in Definition \ref{def:outgoing-radial-solution} give us strong decay as $r^*\to \pm \infty$, so an application of the fundamental theorem of calculus yields, recalling \eqref{eq:Teukolsky-potential},
\begin{align}\label{eq:T-estimate-wave-upper-half}
\begin{split}
0&=\int_{-\infty}^\infty \lp(Q^T[u]\rp)' dr^*\\
& \geq \int_{-\infty}^\infty \Im(\omega)\lp[|u'|^2+\lp(|\omega|^2\lp(1-\frac{a^2\Delta}{(r^2+a^2)^2}\rp)-\frac{a^2m^2}{(r^2+a^2)^2}\rp) |u|^2\rp]dr^*
\end{split}
\end{align} 
where we have used the fact that, by statement 1(b) of Proposition~\ref{def:angular-ode}, $\Im (\overline{\bm\uplambda}\omega)\geq 0$. It is clear that Theorem~\ref{thm:mode-stability-upper-half} certainly holds if $a=0$: the coefficient on $|u|^2$ is positive for all $r^*$ and we can thus conclude that $u\equiv 0$. On the other hand, if $|a|>0$, there exist $\omega$ with $\Im\omega>0$ and $m\in\mathbb{Z}$ such that the coefficient on the $|u|^2$ term is negative near $r=r_+$. This negativity for certain frequencies is a manifestation of the fact that, near $r=r_+$, the Killing field $\p_t$ becomes spacelike and prevents us from inferring mode stability from \eqref{eq:T-estimate-wave-upper-half} when $|a|>0$.

\subsubsection{The case \texorpdfstring{$s\neq 0$}{of nontrivial spin}}
\label{sec:superradiance-s}

In the case of nontrivial spin, we cannot hope to be able to use the procedure we just outlined even if $a=0$: as the Teukolsky potential \eqref{eq:Teukolsky-potential} is complex even when $\omega$ is real, the $T$-current \eqref{eq:T-current} is not a conserved quantity for the system. However, recall that, for solutions of the radial ODE \eqref{eq:u-Schrodinger} $\swei{u}{+s}$ and $\swei{u}{-s}$, the Wronskian defined by 
\begin{equation}\label{eq:wronskian-current}
W\lp(\swei{u}{+s},\overline{\swei{u}{-s}}\rp):=\lp(\swei{u}{+s}\rp)'\cdot\overline{\swei{u}{-s}}-\swei{u}{+s}\cdot\lp(\overline{\swei{u}{-s}}\rp)'
\end{equation}
satisfies
\begin{align*}
&\frac{d}{dr^*}\lp[W\lp(\swei{u}{+s},\overline{\swei{u}{-s}}\rp)\rp]\\
&\qquad=\lp[4i\Im \omega\Re\omega+\frac{\Delta(\lambda^{[+s]}+s-\overline{\lambda^{[-s]}-s})}{(r^2+a^2)^2}\rp]\swei{u}{+s}\cdot\overline{\swei{u}{-s}}\\
&\qquad\qquad+2i\Im\omega\lp[-\frac{2a^2\Delta\Re\omega+4Mamr)}{(r^2+a^2)^2}+\frac{2s(r^3-3Mr^2+a^2r+a^2M}{(r^2+a^2)^2}\rp] \swei{u}{+s}\cdot\overline{\swei{u}{-s}}\,.
\end{align*}
Clearly, $W'=0$ when $\omega$ is real, yielding a conservation law which substitutes the $T$-current \eqref{eq:T-current} we had for $s=0$.

\begin{proposition}[Energy identity for the homogeneous radial ODE, $s\neq 0$]\label{prop:wronskian-identity} Fix $M>0$, $|a|\leq M$ and $s\in\frac12\mathbb{Z}_{> 0}$. Let $(\omega,m,\lambda)$ be an admissible frequency triple with respect to $s$ and $a$ such that $\omega$ is real. Let $\swei{u}{\pm s}$ be solutions to \eqref{eq:u-Schrodinger} given by \eqref{eq:u-general-asymptotics}.

Define $\mathfrak{C}_s$ by \eqref{eq:redef-Cs} if $\swei{u}{\pm s}$ are outgoing in the sense of Definition~\ref{def:outgoing-radial-solution} and by Proposition~\ref{prop:eigenfunctions} otherwise. Then, the following identities hold for a general $\swei{u}{\pm s}$ if $\omega\neq m\upomega_+$ and, with vanishing right hand side, for an outgoing $\swei{u}{\pm s}$.

For spin $+s$, if $|a|=M$,
\begin{gather}
\begin{gathered}
2^{4s}\omega^2[2M^2(\omega-m\upomega_+)]^{2s}\omega^{4s}\lp|\swei{a}{+s}_{\mc{I}^+}\rp|^2+\omega(\omega-m\upomega_+)\omega^{2s}\mathfrak{C}_s\lp|\swei{a}{+s}_{\mc{H}^+}\rp|^2 \\
=\omega^2[2M^2(\omega-m\upomega_+)]^{2s}\mathfrak{C}_s\lp|\swei{a}{+s}_{\mc{I}^-}\rp|^2\\
+2^{4s}\omega(\omega-m\upomega_+)\omega^{2s}[2M^2(\omega-m\upomega_+)]^{4s}\lp|\swei{a}{+s}_{\mc{H}^-}\rp|^2\,, 
\end{gathered}\label{eq:energy-identity-extremal}
\end{gather}
if $|a|<M$ and $s$ is an integer,
\begin{gather*}
\begin{gathered}
(2\omega)^{2s}\mathfrak{C}_s\lp|\swei{a}{+s}_{\mc{H}^+}\rp|^2+4\omega(\omega-m\upomega_+)\prod_{j=1}^{s-1}\lp[4|\xi|^2+(s-j)^2\rp](2\omega)^{4s}\lp|\swei{a}{+s}_{\mc{I}^+}\rp|^2\\
=4(2\omega)^{2s}|\xi|^2(4|\xi|^2+s^2)\prod_{j=1}^{s-1}\lp[4|\xi|^2+(s-j)^2\rp]^2\lp|\swei{a}{+s}_{\mc{H}^-}\rp|^2\\
+4\omega(\omega-m\upomega_+)\prod_{j=1}^{s-1}\lp[4|\xi|^2+(s-j)^2\rp]\mathfrak{C}_s\lp|\swei{a}{+s}_{\mc{I}^-}\rp|^2\,, 
\end{gathered} \label{eq:energy-identity-sub-integer}
\end{gather*}
taking the products denoted by $\Pi$ to be the identity if $s=1$; and, finally, if $|a|<M$ and $s$ is a half-integer,
\begin{equation}
\begin{gathered}
(2\omega)^{4s}\frac{2Mr_+}{r_+-r_-}\prod_{j=1}^{\lfloor s\rfloor}\lp[4|\xi|^2+(s-j)^2\rp]\lp|\swei{a}{+s}_{\mc{I}^+}\rp|^2+(2\omega)^{2s-1}\mathfrak{C}_s\lp|\swei{a}{+s}_{\mc{H}^+}\rp|^2\\
=\frac{2Mr_+}{r_+-r_-}\prod_{j=1}^{\lfloor s\rfloor}\lp[4|\xi|^2+(s-j)^2\rp]\mathfrak{C}_s\lp|\swei{a}{+s}_{\mc{I}^-}\rp|^2\\
+(2\omega)^{2s-1}(4|\xi|^2+s^2)\prod_{j=1}^{\lfloor s \rfloor}\lp[4|\xi|^2+(s-j)^2\rp]^2\lp|\swei{a}{+s}_{\mc{H}^-}\rp|^2\,, \end{gathered}\label{eq:energy-identity-sub-half}
\end{equation}
taking the products denoted by $\Pi$ to be the identity if $s=1/2$. For spin $-s$, the above identities hold, replacing $\swei{a}{+s}_{\mc{H}^\pm}$ with $\swei{a}{-s}_{\mc{H}^\mp}$ and $\swei{a}{+s}_{\mc{I}^\pm}$ with $\swei{a}{-s}_{\mc{I}^\mp}$.
\end{proposition}

\begin{proof}
We first notice that we can rewrite the Wronskian as (dropping most subscripts)
\begin{align*}
W\lp(\swei{u}{+s},\overline{\swei{u}{-s}}\rp)&=W\lp(\Delta^{-s/2}(r^2+a^2)^{1/2}\Delta^s\swei{R}{+s},\Delta^{-s/2}(r^2+a^2)^{1/2}\overline{\swei{R}{-s}}\rp)\\
&=\Delta^{-s}\lp(\Delta\frac{d}{dr}\lp(\Delta^s\swei{R}{+s}\rp)\overline{\swei{R}{-s}}-\lp(\Delta^s\swei{R}{+s}\rp)\Delta\frac{d}{dr}\overline{\swei{R}{-s}}\rp)\,.
\end{align*}
Recalling the asymptotic expansions for  $\swei{R}{\pm s}_{\mc{I}^+}$, $\swei{R}{\pm s}_{\mc{I}^-}$, $\swei{R}{\pm s}_{\mc{H}^+}$ and $\swei{R}{\pm s}_{\mc{H}^-}$ in \eqref{eq:model-solutions-series}, we have
\begin{align*}
W\lp[\swei{u}{+s},\overline{\swei{u}{-s}}\rp]&=2i\omega\Delta\lp(\swei{a}{+s}_{\mc{I}^+}\overline{\swei{a}{-s}_{\mc{I}^+}}\swei{R}{+s}_{\mc{I}^+}\overline{\swei{R}{-s}_{\mc{I}^+}}-\swei{a}{+s}_{\mc{I}^-}\overline{\swei{a}{-s}_{\mc{I}^-}}\swei{R}{+s}_{\mc{I}^-}\overline{\swei{R}{-s}_{\mc{I}^-}}\rp)\\
&=2i\omega\lp(\swei{a}{+s}_{\mc{I}^+}\overline{\swei{a}{-s}_{\mc{I}^+}}-\swei{a}{+s}_{\mc{I}^-}\overline{\swei{a}{-s}_{\mc{I}^-}}\rp)+O(r^{-1})\,, \\
W\lp[\swei{u}{+s},\overline{\swei{u}{-s}}\rp]
&=
\begin{dcases}
(2\xi-s)\frac{r_+-r_-}{2Mr_+}\lp(\swei{a}{+s}_{\mc{H}^+}\overline{\swei{a}{-s}_{\mc{H}^+}}-\swei{a}{+s}_{\mc{H}^-}\overline{\swei{a}{-s}_{\mc{H}^-}}\rp)+O(r-r_+) &\text{if~}|a|<M\\
-2i(\omega-m\upomega_+)\lp(\swei{a}{+s}_{\mc{H}^+}\overline{\swei{a}{-s}_{\mc{H}^+}}-\swei{a}{+s}_{\mc{H}^-}\overline{\swei{a}{-s}_{\mc{H}^-}}\rp)+O(r-r_+) &\text{if~}|a|=M
\end{dcases}\,.
\end{align*}

In order to reduce these expressions to identities for just one of the signs of spin, we use the Teukolsky--Starobinksy identities \eqref{eq:TS-radial-2}. For instance, to obtain the identity for spin $+ s$, we impose that
\begin{align*}
\swei{R}{-s}:=\Delta^s\lp(\mc{D}_0^+\rp)^{2s}\lp(\Delta^s\swei{R}{+s}\rp)\,,
\end{align*}
which, noting  $\omega\in\mathbb{R}$ implies $\mathfrak{C}_s^{(5)}=\overline{\mathfrak{C}_s^{(1)}}$ and $\mathfrak{C}_s^{(6)}=\overline{\mathfrak{C}_s^{(2)}}$, yields
\begin{align*}
\swei{a}{-s}_{\mc{I}^+}=\mathfrak{C}_s^{(1)}\,,\quad \swei{a}{-s}_{\mc{I}^-}=\mathfrak{C}_s\lp(\overline{\mathfrak{C}_s^{(1)}}\rp)^{-1}\,,\quad
\swei{a}{-s}_{\mc{H}^+}=\mathfrak{C}_s\lp(\mathfrak{C}_s^{(2)}\rp)^{-1}\,,\quad \swei{a}{-s}_{\mc{H}^-}=\overline{\mathfrak{C}_s^{(2)}}\,.
\end{align*}
Hence, we find that conservation of the Wronskian gives
\begin{gather}\label{eq:general-conservation-plus}
\begin{gathered}
\frac{2i\omega}{\mathfrak{C}_s^{(1)}}\lp(\lp|\mathfrak{C}_s^{(1)}\rp|^2\lp|\swei{a}{+s}_{\mc{I}^+}\rp|^2 -\mathfrak{C}_s\lp|\swei{a}{+s}_{\mc{I}^-}\rp|^2 \rp)\\
=\lp\{\begin{array}{lr}
\lp(2\xi-s\rp)\frac{r_+-r_-}{2Mr_+} & \text{if } |a|<M\\
-2i(\omega-m\upomega_+) & \text{if } |a|=M
\end{array}\rp\}\lp(\overline{\mathfrak{C}_s^{(2)}}\rp)^{-1}\lp(\mathfrak{C}_s\lp|\swei{a}{+s}_{\mc{H}^+}\rp|^2 -|\mathfrak{C}_s^{(2)}|^2\lp|\swei{a}{+s}_{\mc{H}^-}\rp|^2 \rp)\,.
\end{gathered}
\end{gather} 

Similarly, defining
\begin{align*}
\swei{R}{+s}:=\lp(\mc{D}_0^+\rp)^{2s}\swei{R}{-s}\,,
\end{align*}
is equivalent, by \eqref{eq:TS-radial-2}, to setting
\begin{align*}
\swei{a}{+s}_{\mc{I}^+}:=\mathfrak{C}_s\lp(\mathfrak{C}_s^{(1)}\rp)^{-1}\,,\quad \swei{a}{+s}_{\mc{I}^-}:=\overline{\mathfrak{C}_s^{(1)}}\,,\quad
\swei{a}{+s}_{\mc{H}^+}:=\mathfrak{C}_s^{(2)}\,,\quad \swei{a}{+s}_{\mc{H}^-}:=\mathfrak{C}_s\lp(\overline{\mathfrak{C}_s^{(2)}}\rp)^{-1}\,,
\end{align*}
which yields the conservation law
\begin{gather}\label{eq:general-conservation-minus}
\begin{gathered}
\frac{2i\omega}{\mathfrak{C}_s^{(1)}}\lp(|\mathfrak{C}_s^{(1)}|^2\lp|\swei{a}{-s}_{\mc{I}^-}\rp|^2 -\mathfrak{C}_s\lp|\swei{a}{-s}_{\mc{I}^+}\rp|^2 \rp)\\
=\lp\{\begin{array}{lr}
\lp(2\xi-s\rp)\frac{r_+-r_-}{2Mr_+} & \text{if } |a|<M\\
-2i(\omega-m\upomega_+) & \text{if } |a|=M
\end{array}\rp\}\lp(\overline{\mathfrak{C}_s^{(2)}}\rp)^{-1}\lp(\mathfrak{C}_s\lp|\swei{a}{-s}_{\mc{H}^-}\rp|^2 -|\mathfrak{C}_s^{(2)}|^2\lp|\swei{a}{-s}_{\mc{H}^+}\rp|^2 \rp)\,.
\end{gathered}
\end{gather}

We note that, if $|a|<M$ and $s$ is an integer,
\begin{align*}
\mathfrak{C}_s^{(2)} &= \prod_{j=0}^{2s-1}(2\xi+s-j)=(2\xi+s)(2\xi+s-1)\cdots 2\xi\cdots(2\xi+1-s)\\
&=2\xi(2\xi+s)\prod_{j=1}^{s-1}(2\xi+s-j)(2\xi+j-s)=-2\xi(2\xi+s)(-1)^s\prod_{j=1}^{s-1}\lp[4|\xi|^2+(s-j)^2\rp]\,,
\end{align*}
but if $s$ is a half integer,
\begin{align*}
\mathfrak{C}_s^{(2)} &= \prod_{j=0}^{2s-1}(2\xi+s-j)=(2\xi+s)(2\xi+s-1)(2\xi+s-2)\cdots(2\xi+2-s)(2\xi+1-s)\\
&=(2\xi+s)\prod_{j=1}^{\lfloor s \rfloor}(2\xi+s-j)(2\xi+j-s)=i^{2s+1}(2\xi+s)\prod_{j=1}^{\lfloor s \rfloor}\lp[4|\xi|^2+(s-j)^2\rp]\,,
\end{align*}

The result now follows by substituting in \eqref{eq:general-conservation-plus} and \eqref{eq:general-conservation-minus} the explicit constants $\mathfrak{C}_s^{(2)}$ and $\mathfrak{C}_s^{(1)}$, as given above and in Proposition~\ref{prop:TS-radial}.
\end{proof}

As in the previous section, we will now recall how such a current can be used to show Theorems~\ref{thm:mode-stability-subextremal} and \ref{thm:mode-stability-full} for real $\omega$ in the case $a=0$ but fails to do so whenever $|a|>0$.  

Consider solutions to the homogeneous radial ODE~\eqref{eq:u-Schrodinger} compatible with the outgoing condition from Definition~\ref{def:outgoing-radial-solution}. If $a=0$, then $\upomega_+=0$ and, by explicit computation of the radial Teukolsky--Starobinsky constant, we find that, if $\lambda=\bm\uplambda$ from Proposition~\ref{def:angular-ode}, then $\mathfrak{C}_s(\omega,m,\bm\uplambda)>0$. Hence, equations \eqref{eq:energy-identity-extremal} and \eqref{eq:energy-identity-sub-integer} yield $\swei{a}{\pm s}_{\mc I^+}=\swei{a}{\pm s}_{\mc H^+}=0$ we can again appeal to Lemma~\ref{lemma:R-general-asymptotics} to obtain mode stability. 

An interesting remark is that such a proof of Theorem~\ref{thm:mode-stability-real-axis} in the whole range $|a|\leq M$ would also hold for half-integer spin, assuming that one could show $\mathfrak{C}_s(\omega,m,\bm\uplambda)>0$ for any admissible frequency triple $(\omega,m,l)$ for which there could be nontrivial mode solutions. Indeed, consistently with \cref{rmk:TS-constant-sign}, the absence of superradiance for $|s|=1/2$ and $|s|=3/2$ was first shown in \cite{Unruh1973} and \cite{Mason1998}, respectively, who used a spinor formalism (instead of the Teukolsky equation) and interpreted the Wronskian conservation law we obtain in Proposition~\ref{prop:wronskian-identity} as a conservation law for the number of particles.

For integer spin and $|a|>0$, on the other hand, assuming $\mathfrak{C}_s(\omega,m,\lambda)>0$ for any admissible frequency triple $(\omega,m,\lambda)$ for which there could be nontrivial solutions to the homogeneous radial ODE~\eqref{eq:u-Schrodinger} compatible with the outgoing condition (this must hold if $\lambda=\bm\uplambda$ at least for $0<|s|\leq 2$, by \cref{rmk:TS-constant-sign}), \eqref{eq:energy-identity-extremal} and \eqref{eq:energy-identity-sub-integer} fail to give an estimate for $u(\pm \infty)$ if the frequency parameters are superradiant \eqref{eq:superradiant-freqs}. Hence, we cannot deduce mode stability on the real axis from Proposition~\ref{proposition:estimates-u-tilde} for $|a|>0$ and integer $s$.

\begin{remark}
For the integral transformations we will present in Section \ref{sec:integral-transformation}, we can obtain estimates which show, via a $T$-current \eqref{eq:T-current} (rather than the Wronskian current presented in this section), that the transformed quantities vanish identically, even when $s\neq 0$ (see Section~\ref{sec:proof}).
\end{remark}

\section{Integral transformations for \texorpdfstring{$s\leq 0$}{non-positive spin}}

In this section, we will define integral transformations for the radial ODE \eqref{eq:radial-ODE} for extremal and subextremal Kerr backgrounds. In \cref{sec:integral-transformation}, we consider an extremal Kerr solution and introduce a novel integral transformation in \cref{prop:ode-u-tilde}. On subextremal Kerr spacetimes, we consider Whiting's transformation and extend the results in \cite{Whiting1989,Shlapentokh-Rothman2015} to  \cref{prop:ode-u-tilde-sub} in \cref{sec:integral-transformation-subextremal}. 

We remark that, in the interest of presenting a unified picture for the full range $|a|\leq M$, the structure of \cref{sec:integral-transformation-subextremal} is very similar to that of \cref{sec:integral-transformation}.

\subsection{An integral radial transformation for extremal Kerr \texorpdfstring{($|a|=M$)}{}}
\label{sec:integral-transformation}

 We begin by rewriting the radial ODE \eqref{eq:radial-ODE} as
\begin{equation}
\begin{gathered}
\frac{d}{dr}\lp[(r-M)^{2}\frac{d}{dr}\rp]R(r)+2s(r-M)\frac{d}{dr} R(r)\\
-\lp[\frac{[\beta -\alpha(r-M)-\gamma  (r-M)^2]^2}{(r-M)^2}+\frac{2s[\beta -\alpha(r-M)-\gamma  (r-M)^2]}{(r-M)}\rp] R(r)\\
-\lp[4s\gamma (r-M)+2s\alpha+L\rp]R(r)=(r-M)^2 \hat{F}(r)\,,
\end{gathered}\label{eq:radial-ODE-alpha-beta-gamma}
\end{equation}
where $\hat{F}(r)$ is a smooth inhomogeneity compactly supported away from $r=r_+$ and $r=\infty$ and we have defined
\begin{align}
\alpha := -2iM\omega \,,\quad  \beta := 2iM^2(\omega-m\omega_+)\,,\quad \gamma= -i\omega\,, \quad L:= \lambda+a^2\omega^2-2am\omega \,. \label{eq:def-alpha-beta-gamma}
\end{align}

While in the subextremal case the radial ODE has a regular singularity at $r_+$ (see Section~\ref{sec:integral-transformation-subextremal}), in (\ref{eq:radial-ODE-alpha-beta-gamma}) we have an irregular singularity of rank 1 (see for instance \cite{Olver1973}), so Whiting's transformation cannot be defined in the extremal case. Indeed, to apply the strategy of \cite{Whiting1989,Shlapentokh-Rothman2015} one requires a fundamentally novel transformation. We will consider the following:

\begin{proposition} \label{prop:ode-u-tilde} Fix $M>0$, $|a|=M$ and $s\in \frac12 \mathbb{Z}_{\leq 0}$. Let $(\omega,m,\lambda)$ be admissible frequency parameters with respect to $a$ and $s$ (see Definition~\ref{def:admissible-freqs}). Let $R_{m\lambda}^{[s],(a\omega)}$ be an outgoing solution to the radial ODE~\eqref{eq:radial-ODE-alpha-beta-gamma}  as in Definition~\ref{def:outgoing-radial-solution}. Dropping subscripts, define $\tilde{u}$ as the integral transformation
\begin{align}
\begin{split}
\tilde{u}(x) &:=\lim_{y\to 0}(x^2+2M^2)^{1/2}(x-M)^{-s}(x-2M)^{\alpha}  \times\\
&\qquad\times \int_{M}^\infty e^{-\frac{2\gamma}{M}(x+iy-M)(r-M)}(r-M)^{\alpha}e^{\beta(r-M)^{-1}}e^{-\gamma r}R(r) dr\,,
\end{split} \label{eq:extremal-transformation}
\end{align}
with $\alpha$, $\beta$ and $\gamma$ as in \eqref{eq:def-alpha-beta-gamma}, where the limit is taken in the function space $L^2_x\lp([2M,\infty)\rp)$. Introduce a new coordinate $x^*(x):(2M,+\infty)\to (-\infty,\infty)$ by
\begin{align*}
\frac{dx^*}{dx}= \frac{x^2+2M^2}{(x-M)(x-2M)}\,,\quad x^*(3M)=0\,.
\end{align*}
Then the following hold:
\begin{enumerate}
\item $\tilde{u}(x)$ is in fact smooth for $x\in(2M,+\infty)$;
\item $\tilde{u}$ satisfies the ODE
\begin{align}
\tilde{u}''+\tilde{V}\tilde{u}=\frac{(x-2M)(x-M)}{x^2+2M^2}\tilde{H}\,, \label{eq:ode-u-tilde}
\end{align}
where the inhomogeneity $\tilde{H}$ is given by
\begin{align}
\begin{split}
\tilde{H}&:=(x^2+2M^2)^{1/2}(x-M)^{-s}(x-2M)^{\alpha}  \times\\
&\qquad\times \int_{M}^\infty e^{-\frac{2\gamma}{M}(x-M)(r-M)}(r-M)^{\alpha}e^{\beta(r-M)^{-1}}e^{-\gamma r}\hat{F}(r) dr\,,
\end{split}\label{eq:H-tilde}
\end{align}
and the potential is
\begin{align*}
\tilde{V}(x)&:= \frac{M (x-M)^2 [3(2x-3M)+(x-M)^2+(x-2M)^2]}{\left(x^2+2M^2\right)^2}\omega^2 -\frac{(x-2M)^2}{(x^2+2M^2)^2}s^2\\
&\qquad-\frac{(x-M) (x-2M)}{\left(x^2+2M^2\right)^2}(\lambda+s)  -\frac{2 m \omega  (x-M) (x-2M) (2 x-3M)}{\left(x^2+2M^2\right)^2}\numberthis \label{eq:V-tilde} \\ 
&\qquad-\frac{(x-M)(x-2M)}{(x^2+2M^2)^4}\lp[2M^2 (x-M)(x-2M)+3Mx(x^2-2M^2)\rp]\,.
\end{align*}
\item $\tilde{u}$ and $\tilde{u}'$ are bounded for $x^*\in\mathbb{R}$;
\item $\tilde{u}$ and $\tilde{u}'$ satisfy the boundary conditions
\begin{enumerate}[label=(\alph*)]
\item  if $\Im\omega>0$ and $\hat{F}\equiv 0$, then 
\begin{enumerate}
\item $\tilde{u}',\tilde{u}=O\lp((r-r_+)^{2M\Im\omega}\rp)$ as $x\to r_+$,
\item $|\tilde{u}'\tilde{u}|= O(|\tilde{u}|^2 x^{-1/2})=o(1)$ as $x\to\infty$,
\end{enumerate}
\item if $\omega\in\mathbb{R}\backslash\{0\}$ and $\hat{F}$ is compactly supported in $(M,\infty)$, then
\begin{enumerate}
\item $\tilde{u}'+\frac13 i\omega \tilde{u}=O(r-M)$ as $x\to M$,
\item if additionally $\omega(\omega-m\upomega_+)>0$,  
\begin{gather*}
x^{1/4}\lp(\tilde{u}'-4i\,\mr{sign}\,\omega\sqrt{2M\omega(\omega-m\upomega_+)}x^{-1/2}\tilde{u}\rp)=O(x^{-1/2})\,,
\end{gather*}
as $x\to\infty$, and
\begin{gather*}
\lp|(x^{-1/4}\tilde{u})(+\infty)\rp|^2 = \frac{\pi}{4M|\omega|}\lp|\frac{2M^3(\omega-m\upomega_+)}{\omega}\rp|^{1/2-2s}\lp|(\Delta^{s/2}u)(-\infty)\rp|^2\,,
\end{gather*}

\item if additionally $\omega(\omega-m\upomega_+)<0$,  
\begin{align*}
&\exp\lp(4\sqrt{-2M\omega(\omega-m\upomega_+)}x^{1/2}\rp)x^{1/4}\times \\
&\qquad\times\lp(\tilde{u}'-2\sqrt{2M\omega(\omega-m\upomega_+)}x^{-1/2}\tilde{u}\rp)
=O(x^{-1/2})\,,
\end{align*}
as $x\to\infty$, and
\begin{align*}
&\lp|\lp[\exp\lp(-4\sqrt{-2M\omega(\omega-m\upomega_+)}x^{1/2}\rp)x^{-1/4}\tilde{u}\rp](+\infty)\rp|^2 \\
&\qquad= \frac{\pi}{4M|\omega|}\lp|\frac{2M^3(\omega-m\upomega_+)}{\omega}\rp|^{1/2-s}\lp|(\Delta^{s/2}u)(-\infty)\rp|^2\,;
\end{align*}

\end{enumerate} 
\end{enumerate}
\item the integral transformation \eqref{eq:extremal-transformation} defines an injective map $R\mapsto \tilde{u}$: if $\tilde{u}$ vanishes identically, then $R$ must also vanish identically.
\end{enumerate}
\end{proposition}

\begin{remark} \label{rmk:prop-potential}
For our proof in \cref{sec:proof-s-negative}, it will be useful to highlight the following properties of $\tilde{V}(x)$ for $x\in(2M,\infty)$:
\begin{enumerate}[label=(\roman*)]
\item $\omega^2$ has a positive coefficient;
\item $\lambda$ has a non-positive coefficient;
\item the $(\omega,m,l)$-independent part of $\tilde{V}$ is non-negative;
\item $\tilde{V}$ and is real whenever $\omega$ is real;
\end{enumerate}
These properties follow easily from \eqref{eq:V-tilde}.
\end{remark}

To prove Proposition~\ref{prop:ode-u-tilde}, it will be useful to break up \eqref{eq:extremal-transformation} into smaller pieces which we will analyze separately. With this in mind, we define the auxiliary function
\begin{align}
g(r):=(r-M)^{-\alpha+2s} e^{-\beta (r-M)^{-1}}e^{-\gamma r} R(r)\,, \label{eq:def-g}
\end{align}
so that we have factored out the oscillatory behavior at the horizon but reinforced it as $r\to\infty$. Indeed, if $R(r)$ satisfies the boundary conditions in \cref{def:mode-solution}, we have
\begin{align}
\begin{alignedat}{3} \label{eq:g-bdry}
g(r) &= \sum_{k=0}^\infty b_k (r-M)^k &\text{~ as } r\to M\,, \\
g(r) &= e^{-2\gamma r}r^{-\alpha-1}\lp[\sum_{k=0}^N c_k r^{-k}+O\lp(r^{-N-1}\rp)\rp] &\text{~ as } r\to \infty\,. 
\end{alignedat}
\end{align}

We also define the following weighted integral of $g$: for $z=x+iy$ with $(x,y)\in[2M,\infty)\times\{y\in[-1,1]\colon y\Re\omega\geq 0 \text{~or~} y\omega>0\}$,
\begin{align}
\begin{split}
\tilde{g}(z) &:= \int_{M}^\infty e^{A(z-M)(r-M)}(r-M)^{2\alpha-2s} e^{2\beta(r-M)^{-1}}e^{2\gamma r}g(r) dr \\
&=\int_{M}^\infty e^{A(z-M)(r-M)}(r-M)^{\alpha}e^{\beta(r-M)^{-1}}e^{\gamma r}R(r) dr \,,
\end{split}\label{eq:def-g-tilde}
\end{align}
where $A=-2\gamma/M$. With these definitions, the integral transformation \eqref{eq:extremal-transformation} becomes simply
\begin{align}
\tilde{u}(x)=(x^2+2M^2)^{1/2}(x-M)^{-s}(x-2M)^{\alpha}\lim_{y\to 0}\tilde{g}(x+iy)\,, \label{eq:extremal-transformation-with-g-tilde}
\end{align}
where the limit is, \textit{a priori}, taken in the function space $L^2_x([2M,\infty))$.

We will prove \cref{prop:ode-u-tilde} over the next sections: first, in \cref{sec:integral-transformation-well-defined}, we show that $\tilde{g}(x+iy)$ admits a pointwise $C_x^{1,1/2}([2M,+\infty))$ extension as $y\to 0$; then, in \cref{sec:ode-auxiliary}, we show that it is in fact a smooth solution to a second order ODE. We also obtain precise asymptotics for $\tilde{g}(x)$ in Section \ref{sec:asymptotics-g-tilde}. Finally, in \cref{sec:integral-transformation-smooth} we put these together to prove \cref{prop:ode-u-tilde}. The structure of the section is based on the approach to Whiting's subextremal radial transformation in \cite{Shlapentokh-Rothman2015}.

\subsubsection{Defining the integral transformation for real \texorpdfstring{$\omega$}{time frequency}}
\label{sec:integral-transformation-well-defined}

In this section we will show that \eqref{eq:extremal-transformation-with-g-tilde} is well-defined, by understanding the limit of $\tilde{g}(x+iy)$ as $y\to 0$ with $x\in[2M,\infty)$. 

Clearly, when $\Im \omega >0$, the terms $e^{\beta(r-M)^{-1}}$ and $e^{A (x-M)(r-M)}$ in the integrand contribute with exponential decay as $r$ approaches $M$ and $\infty$, respectively, so we can easily take $y=0$ in \eqref{eq:def-g-tilde} and obtain a convergent integral, regardless of the sign of $s$. On the other hand, if $\omega$ is real, we no longer have this exponential decay at both ends when $y=0$. The integrand is at worst $O((r-M)^{-2s})$ as $r\to M$, which is integrable even when $y=0$, and and $O(e^{-\Im \omega x-\Re\omega yr}r^{-1-2s})$ as $r\to \infty$, which is not integrable if $y=\Im\omega=0$. Hence, to define $\tilde{g}$ properly in the limit $y\to 0$, we will integrate by parts to produce more decay near $r=\infty$:

\begin{lemma} \label{lemma:g-tilde-IBP}  Fix $s\leq 0$ and $\omega$ such that $\Im\omega>0$ or $\omega\in\mathbb{R}\backslash\{0\}$. Let $z:=x+iy$ where $(x,y)\in[2M,\infty)\times\{y\in[-1,1]\colon y\Re\omega\geq 0 \text{~or~} y\omega>0\}$. Let $\varepsilon>0$ be arbitrary; we have
\begin{align*} 
\tilde{g}(z)&=\int_{M}^{M+\varepsilon} e^{A(z-M)(r-r_-)}(r-M)^{2(\alpha-s)}e^{2\beta(r-M)^{-1}}e^{2\gamma r} g(r) dr \numberthis\label{eq:g-tilde-IBP}\\
&\qquad +\frac{1}{\lp[A(z-M)\rp]^{1-2s}}\int_{M+\varepsilon}^{\infty} \lp\{e^{A(z-M)(r-M)}\times\rp.\\
&\qquad\qquad\qquad\qquad\qquad\qquad \lp. \times\lp(\frac{d}{dr}\rp)^{1-2s}\lp((r-M)^{2(\alpha-s)}e^{2\beta(r-M)^{-1}}e^{2\gamma r} g(r)\rp)\rp\} dr \\
&\qquad+\sum_{k=1}^{2-2s} \frac{(-1)^k}{[A(z-M)]^k} \lp[e^{A(z-M)(r-M)}\rp.\times\\
&\qquad\qquad\qquad\qquad\qquad\qquad\lp.\times\lp(\frac{d}{dr}\rp)^{k}\lp((r-M)^{2(\alpha-s)}e^{2\beta(r-M)^{-1}}e^{2\gamma r} g(r)\rp)\rp]_{r=M+\varepsilon}\,,
\end{align*}
and similarly, for $j=1,2,3$,
\begin{align*} 
&\frac{1}{A^j}\frac{\p^j}{\p x^j}\tilde{g}(z)\\
&\quad=\int_{M}^{M+\varepsilon} e^{A(z-M)(r-r_-)}(r-M)^{2(\alpha-s)+j}e^{2\beta(r-M)^{-1}}e^{2\gamma r} g(r) dr \numberthis\label{eq:derivative-g-tilde-IBP}\\
&\quad\qquad +\frac{1}{\lp[A(z-M)\rp]^{1+j-2s}}\int_{M+\varepsilon}^{\infty} \lp\{e^{A(z-M)(r-M)}\times\rp.\\
&\quad\qquad\qquad\qquad\qquad\qquad\qquad \lp. \times\lp(\frac{d}{dr}\rp)^{1+j-2s}\lp((r-M)^{2(\alpha-s)+j}e^{2\beta(r-M)^{-1}}e^{2\gamma r} g(r)\rp)\rp\} dr \\
&\quad\qquad+\sum_{k=1}^{j-2s} \frac{(-1)^k}{[A(z-M)]^k} \lp[e^{A(z-M)(r-M)}\rp.\times\\
&\quad\qquad\qquad\qquad\qquad\qquad\qquad\lp.\times\lp(\frac{d}{dr}\rp)^{k}\lp((r-M)^{2(\alpha-s)+j}e^{2\beta(r-M)^{-1}}e^{2\gamma r} g(r)\rp)\rp]_{r=M+\varepsilon}\,,
\end{align*}

\noindent Moreover, $\tilde{g}(x+iy)$ admits a unique extension to $(x,y)\in[2M,\infty)\times\{y\in[-1,1]\colon y\Re\omega\geq 0 \text{~or~} y\omega\geq 0\}$ such that
\begin{enumerate}[label=(\roman*)]
\item $\tilde{g}(x+iy)\to \tilde{g}(x)$ in $H^2_x([2M,+\infty))$ and pointwise in $C^{1,1/2}_x([2M,+\infty)$ as $y\to 0$;
\item $\tilde{g}(x)$ and its weak derivative, $\p_x\tilde{g}(x)$, are $L^\infty_{x}([2M,+\infty))$.
\end{enumerate}
\end{lemma}

\begin{proof}

In \eqref{eq:def-g-tilde}, we split the integration range at $M+\varepsilon$ and integrate by parts under the integral ranging over over $[M+\varepsilon,+\infty)$, noting that
\begin{align*}
&e^{-A(z-M)(r-M)} \frac{d}{dr}\lp(e^{A(z-M)(r-M)} \rp) =A(z-M)\,.
\end{align*}
We obtain
\begin{align*}
& \int_{M+\varepsilon}^{\infty} \frac{d}{dr}\lp(e^{A(z-M)(r-M)}\rp) (r-M)^{2(\alpha-s)}e^{2\beta (r-M)^{-1}}e^{2\gamma r} g(r) dr\\
&=\lp[e^{A(z-M)(r-M)}(r-M)^{2(\alpha-s)} e^{2\beta (r-M)^{-1}}e^{2\gamma r} g(r)\rp]_{M+\varepsilon}^{\infty}  \\
&\quad -\int_{M+\varepsilon}^\infty \lp\{e^{A(z-M)(r-M)} (r-M)^{2\alpha-2s-1}e^{2\beta(r-M)^{-1}}e^{2\gamma r} g(r)\rp. \\
&\quad\quad\quad\quad \lp.\times \lp[2(\alpha-s)-\frac{2\beta}{r-M}+2\gamma(r-M) +\frac{(r-M)}{g}\frac{dg}{dr}\rp]\rp\} dr\,.
\end{align*}
If either $\Im\omega>0$ or $y\omega>0$, the upper boundary term vanishes, because it is exponentially decaying as $r\to \infty$. We are left with an integral on the right hand side whose integrand is at worst $O(r^{-2-2s})$; if $s=0$, we are done, but otherwise, we need to repeat this integration by parts procedure $-2s$ times more. In the end, we obtain an integrand which is $O(r^{-2})$, hence integrable.

For the $j$th derivative of $\tilde{g}$ we note that, if either $\Im\omega>0$ or $y\omega>0$, the integral is absolutely convergent, so we can differentiate inside the integral; this produces an $j$ extra powers of $(r-M)$ in the integrand. This is helpful at the horizon, but not so as $r\to \infty$. Thus, to obtain the same decay at infinity for the integrand in the $j$th derivative of $\tilde{g}$, we need to apply the integration by parts procedure $j$ more times, i.e. $-2s+j$ times in total. The boundary terms at infinity will vanish when either $\Im\omega>0$ or $y\omega>0$ due to the exponential decay that was similarly present in the analogous procedure for $\tilde{g}$. Iterating this procedure, we find that if either $\Im\omega>0$ or $y\omega>0$, $\tilde{g}(z)$ is actually holomorphic.

We have now obtained formulas \eqref{eq:g-tilde-IBP} and \eqref{eq:derivative-g-tilde-IBP}. Since $|z|\geq x \geq 2M>M$, the right hand side of these formulas is bounded and $O(x^{-1})$ as $x\to \infty$ for any $y$; hence,
\begin{align}
\begin{split}\label{eq:bound-tilde-g}
&|\tilde{g}(z)|^2+|\p_y\tilde{g}(z)|^2+|\p_x\tilde{g}(z)|^2+|\p_y\p_x\tilde{g}(z)|^2+|\p_x^2\tilde{g}(z)|^2+|\p_y\p_x^2\tilde{g}(z)|^2\\
&\qquad\lesssim |\tilde{g}(z)|^2+|\p_x\tilde{g}(z)|^2+|\p_x^2\tilde{g}(z)|^2+|\p_x^3\tilde{g}(z)|^2\lesssim \frac{1}{x^2}\,,
\end{split}
\end{align}
for any $(x,y)\in[2M,\infty)\times\{y\in[-1,1]\colon y\Re\omega\geq 0 \text{~or~} y\omega>0\}$. We now have two cases
\begin{enumerate}[noitemsep]
\item If $\Im\omega>0$, $\tilde{g}(z)$ is holomorphic for $(x,y)\in[2M,\infty)\times\{y\in[-1,1]\colon y\Re\omega\geq 0 \}$, so (i) must hold. Moreover, (ii) follows from \eqref{eq:bound-tilde-g}. 
\item If $\omega\in\mathbb{R}\backslash\{0\}$, then without loss of generality, we can take $\omega>0$ and $y\in(0,1]$. We define the extension of $\tilde{g}(x+iy)$ to $H^2_x([2M,+\infty))$ as follows. Let $\mathfrak{g}_x\colon (0,1]\to H^2_x([2M,+\infty))$ be given by $\mathfrak{g}_x(y)=\tilde{g}(x+iy)$. By \eqref{eq:bound-tilde-g}, $\mathfrak{g}_x$ is uniformly continuous for $y\in(0,1]$ and hence admits a unique limit as $y\to 0$; we call this limit $\tilde{g}(x) \in H_x^2([2M,\infty))$, thus proving (i). By Morrey's inequality, in fact $\tilde{g}(x+iy)\to \tilde{g}(x)$ in $C^{1,1/2}_x([2M,+\infty))$. Thus,  $\tilde{g}(x)$ and $\p_x \tilde{g}(x)$ are continuous $L^2_x$ functions and, hence, bounded, as stated in (ii).
\end{enumerate}
This concludes the proof.
\end{proof}

\subsubsection{Differential equations for the auxiliary functions}
\label{sec:ode-auxiliary}

Let $\mc{T}_r$ be the double confluent Heun operator \cite{Ronveaux1995}\footnote{Comparing to the analogous operator (\ref{eq:confluent-operator}) which is considered for the subextremal case, we find that the two distinct regular singularities $r_\pm$ in that confluent Heun operator have merged to form an irregular one in our \textit{double} confluent Heun operator, \eqref{eq:double-confluent-operator}.} given by
\begin{align}
\begin{split}
\mc{T}_r &:= (r-M)^2\frac{d^2}{dr^2}+2\lp[(\alpha-s+1)(r-M)-\beta+\gamma(r-M)^2\rp]\frac{d}{dr}\\
&\qquad+\alpha -2s-2s\alpha-L+2\gamma(1-2s)(r-M)\,.
\end{split}\label{eq:double-confluent-operator}
\end{align}

Given the definition of $g$ in \eqref{eq:def-g}, since $R(r)$ is a solution to the radial ODE~\eqref{eq:radial-ODE-alpha-beta-gamma} with inhomogeneity $\hat{F}$, we find that $g$ satisfies $\mc{T}_r g =G$ where
\begin{align*}
G(r):=(r-M)^{-\alpha+2s}e^{-\beta(r-M)^{-1}}e^{\gamma r} \hat{F}(r)
\end{align*}
and $\alpha$, $\beta$ and $\gamma$ are defined in \eqref{eq:def-alpha-beta-gamma}.  

As we will see in \cref{sec:integral-transformation-smooth-sub}, for the subextremal transformation,  one can show that the analogous $g$ (\ref{eq:def-g-sub}) and $\tilde{g}$ (\ref{eq:def-g-tilde-sub}) satisfy differential equations of the same type but with different parameters. The same cannot be true of $g$ (\ref{eq:def-g}) and $\tilde{g}$ (\ref{eq:def-g-tilde}) for the extremal transformation, introduced in \cref{prop:ode-u-tilde}, which we are now considering; however the kernel $e^{A(z-M)(r-M)}$, $z=x+iy$ with $x,y$ as before, can be used to produce a solution to a \textit{reduced confluent equation} \cite{Kazakov1996,Schmidt1995}:
\begin{align*}
&e^{-A(z-M)(r-M)}\mc{T}_r e^{A(z-M)(r-M)}\\
&\quad = (x-B)(x-M)A^2(r-M)^2+[2\gamma+(B-M)A](x-M)(r-M)^2 \\
&\quad\qquad+\lp[2(\alpha-s+1)(x-M)+\frac{2\gamma}{A}(1-2s)\rp]A(r-M) \\
&\quad\qquad +\alpha(1-2s) -2s-L -2\beta A (x-M)\,, 
\end{align*}
so if we choose $B=2M$ and $A=-2\gamma/M$, we have $\tilde{\mc{T}}_x e^{A(z-M)(r-M)} = \mc{T}_r e^{A(z-M)(r-M)}$, where
\begin{align}
\begin{split}
\tilde{\mc{T}}_x &:= (x-2M)(x-M)\frac{d^2}{dx^2}+ [(2\alpha+1)(x-M)+(1-2s)(x-2M)]\frac{d}{dx}\\
&\quad+\alpha(1-2s)-2s-L+\frac{4\beta\gamma}{M}(x-M)
\end{split} \label{eq:reduced-confluent-operator}
\end{align}
is a type of confluent Heun operator with a half-integer singularity at infinity ($\epsilon=0$ in the notation of \cite{Olver2018}). In essence, using the kernel $e^{A(z-M)(r-M)}$ allows us to split the extremal horizon, of multiplicity two, into two ``horizons'' at the expense of losing the usual asymptotics as $r\to \infty$.

\begin{lemma}\label{lemma:reduced-Heun-z} Suppose  $\Im\omega>0$ or $\omega\in\mathbb{R}\backslash\{0\}$. Let $z=x+iy$ with $(x,y)\in[2M,\infty)\times\{y\in[-1,1]\colon y\Re\omega\geq 0 \text{~or~} y\omega>0\}$.  For $y\neq 0$, we have $\tilde{\mc{T}}_x\tilde{g}(x+iy)=\tilde{G}(x+iy)$, where $\tilde{\mc{T}}_x$ is defined by \eqref{eq:reduced-confluent-operator} and 
\begin{align*}
\tilde{G}(x+iy) := \int_{M}^\infty e^{A(z-M)(r-M)}(r-M)^{2\alpha-2s} e^{2\beta(r-M)^{-1}}e^{2\gamma r}G(r) dr\,.
\end{align*}
Moreover, if $y=0$, $\tilde{g}(x)$ is smooth for $x\in(2M,+\infty)$ and satisfies
\begin{equation}
\tilde{\mc{T}}_x\tilde{g}(x)= \tilde{G}(x) \label{eq:ode-g-tilde}
\end{equation}
classically.
\end{lemma}
\begin{proof}
If either $\Im\omega>0$ and $y\Re\omega\geq 0$ or $y\omega>0$, the integral is absolutely convergent and we can thus differentiate under the integral to obtain
\begin{align*}
\tilde{\mc{T}}_x\tilde{g}(z) &= \int_M^\infty \tilde{\mc{T}}_x \lp(e^{A(z-M)(r-M)}\rp) (r-M)^{2\alpha-2s} e^{2\beta(r-M)^{-1}}e^{2\gamma r}g(r) dr \\
&=\int_M^\infty \mc{T}_r \lp(e^{A(z-M)(r-M)}\rp) (r-M)^{2\alpha-2s} e^{2\beta(r-M)^{-1}}e^{2\gamma r}g(r) dr \\
&= \tilde{G}(z)+ \lp[\lp(A(x-M)g-\frac{dg}{dr}\rp)(r-M)^{2(\alpha-s+1)}e^{2\beta(r-M)^{-1}}e^{2\gamma r}e^{A(z-M)(r-M)}\rp]_0^\infty\,,
\end{align*}
where we have made use of the properties of the kernel, the equation for $g$ and the lemma
\begin{lemma} Let $\mc{T}_r$ be a differential operator as defined in \eqref{eq:double-confluent-operator}. Then, for sufficiently regular functions $f,h$, we have
\begin{align*}
&\int_{A_1}^{A_2} \lp(h\mc{T}_r f-f\mc{T}_r h \rp)(r-M)^{2\alpha-2s}e^{2\beta(r-M)^{-1}}e^{2\gamma r} dr \\
&\quad=\lp[\lp(h\frac{df}{dr}-f\frac{dh}{dr}\rp)(r-M)^{2(\alpha-s+1)}e^{2\beta(r-M)^{-1}}e^{2\gamma r}\rp]_{r=A_1}^{r=A_2}\,,
\end{align*}
\end{lemma}
\noindent obtained by using \eqref{eq:double-confluent-operator} and integrating by parts. This lemma justifies the introduction of weights $(r-M)^{2\alpha-2s}e^{2\beta(r-M)^{-1}}e^{2\gamma r}$: when evaluated against such a measure, $\mc{T}_r$ is self-adjoint. 

We still need to show that the boundary terms vanish. This is clear at $r=\infty$, due to the decay of $g$ and $dg/dr$, together with the exponential decay brought of either $\Im\omega>0$ or $y\omega>0$. At the horizon, the extra factor of $(r-M)^2$ also make the boundary term term vanish, due to the boundary conditions (\ref{eq:g-bdry}) that $g$ satisfies.

To consider the case $y=0$, first note that $\tilde{G}(x+iy)$ is smooth even for $y=0$, by compact support\footnote{Note that $\tilde{G}(x+iy)$ can be rewritten as the Fourier transform of a compactly supported function; see the proof of statement 5 of Proposition~\ref{prop:ode-u-tilde} in Section~\ref{sec:integral-transformation-smooth}.} of $\hat{F}$. 
We have shown that $\tilde{g}(x+iy)$, for $y\neq 0$, satisfies the ODE classically, hence, it also satisfies the ODE  weakly as an $H^1_x([2M,\infty))$ function. Thus, for any smooth $\varphi$ compactly supported in $(2M,\infty)$,
\begin{gather}
\begin{gathered}\label{eq:weak-formulation}
\int_{2M}^\infty \frac{d}{dx}\tilde{g}(x)\lp[\lp(2\alpha(x-M)-2s(x-2M)\rp)\varphi(x)
+ (x-M)(x-2M)\frac{d\varphi}{dx}(x)\rp]dx\\
+\int_{2M}^\infty g(x)\lp[\alpha(1-2s)-2s-L+\frac{4\beta\gamma}{M}(x-M)\rp]\varphi(x)dx
-\int_{2M}^\infty \tilde{G}(x)\varphi(x)dx \\
=\int_{2M}^\infty \lp[\frac{d}{dx}\tilde{g}(x)-\frac{d}{dx}\tilde{g}(x+iy)\rp]\lp[2\alpha(x-M)-2s(x-2M)\rp]\varphi(x)dx\\
+\int_{2M}^\infty \lp[\frac{d}{dx}\tilde{g}(x)-\frac{d}{dx}\tilde{g}(x+iy)\rp](x-M)(x-2M)\frac{d\varphi}{dx}(x)dx\\
+\int_{2M}^{\infty} \lp[g(x)-g(x+iy)\rp]\lp[\alpha(1-2s)-2s-L+\frac{4\beta\gamma}{M}(x-M)\rp]\varphi(x)dx\\
-\int_{2M}^{\infty}\lp[\tilde{G}(x)-\tilde{G}(x+iy)\rp]\varphi(x)dx\,.
\end{gathered}
\end{gather}
As the left hand side is independent of $y$, we can take a limit as $y\to 0$ of the right hand side. Recall that $\tilde{g}(x+iy)\to\tilde{g}(x)$ in $C^{1,1/2}_x([2M,+\infty))$, by Lemma~\ref{lemma:g-tilde-IBP}, with uniform convergence in the compact support of the test function $\varphi$. Hence, we can exchange the limit with the integral, which shows that $\tilde{g}(x)$ is a weak $H^1_x([2M,+\infty))$ solution to the ODE as well. As $\tilde{G}$ is smooth, by elliptic regularity of $\tilde{\mc{T}}_x$ we conclude that $\tilde{g}(x)$ is in fact smooth for $x\in(2M,+\infty)$ and satisfies the ODE classically.
\end{proof}

\subsubsection{Asymptotics of \texorpdfstring{$\tilde{g}$}{tilded g} for large \texorpdfstring{$x$}{x}}
\label{sec:asymptotics-g-tilde}

By Lemma~\ref{lemma:g-tilde-IBP}, $\tilde{g}(x)$ is bounded and $C^1_x([2M,\infty))$, which uniquely determines the behavior at $r=2M$: 
\begin{align*}
\tilde{g}(x)&= \sum_{k=0}^\infty b_k (x-r_+)^{k} \text{~ as } x\to 2M\,
\end{align*}
On the other hand, $\tilde{g}$ satisfies \eqref{eq:ode-g-tilde} by Lemma~\ref{lemma:reduced-Heun-z}. If the inhomogeneity is compactly supported away from $x=\infty$ or vanishes, by asymptotic analysis, we find that $\tilde{g}$ must be a superposition of, for $N\in\mathbb{N}$,
\begin{gather*}
\exp\lp[-4\sqrt{-\beta\gamma} \lp(\frac{x}{M}\rp)^{1/2}\rp] x^{-3/4+s-\alpha}\lp[\sum_{k=0}^N c_k x^{-k/2}+O\lp(x^{-N/2-1/2}\rp)\rp]\,,\\
  \exp\lp[4\sqrt{-\beta\gamma} \lp(\frac{x}{M}\rp)^{1/2}\rp] x^{-3/4+s-\alpha}\lp[\sum_{k=0}^N c_k x^{-k/2}+O\lp(x^{-N/2-1/2}\rp)\rp] \,.
\end{gather*}
If $\Im\omega>0$, then unless $-\beta\gamma$ is real and negative, one of the two solutions provides exponential growth an the other provides exponential decay. Since $-\beta\gamma$ is real if and only if $-\beta\gamma = 2M^2|\omega|^2>0$, boundedness of $\tilde{g}$ restricts us to the case of exponential decay, and hence to a unique choice between the two previous solutions. The same is true for $\omega$ on the real axis if $\beta\gamma<0$, but, if $\beta\gamma>0$, we cannot \textit{a priori} rule out one of the behaviors.  Moreover, we cannot relate the amplitude, i.e.\ the first coefficients in the expansion, to the solution to the radial ODE, $R$, we began with.

With more effort, we can obtain precise asymptotics for $\tilde{g}$ when $\omega$ is real, even if the inhomogeneity is not trivial. We need to use slightly different strategies depending on the sign of  $\omega(\omega-m\upomega_+)$. The next lemma summarizes the technical core of our asymptotic analysis:
\begin{lemma}\label{lemma:technical-core-mode-stability-extremal} Let $h$ be a smooth function on $[M,+\infty)$ such that: if $\mu\nu>0$, $h$ vanishes on $[M+2,\infty)$; if $\mu\nu<0$, $|h(r)|=O(r^{-2-2s})$ as $r\to\infty$. Then, we have
\begin{align*}
Z(\nu,\mu)&:=\int_{M}^\infty e^{i\nu(r-M)}e^{i\mu(r-M)^{-1}}(r-M)^{2\alpha-2s}h(r)dr\\
& = i^{1/2}\sqrt{\frac{\pi}{|\nu|}}\lp(\frac{\mu}{\nu}\rp)^{1/4-s+\alpha}e^{2\sigma\sqrt{|\mu \nu|}}\lp[h(M)+O\lp(|\nu|^{-1/2}\rp)\rp]\,,
\end{align*}
where $\sigma = i\,\mr{sign}\,\nu$ if $\mu\nu>0$ and $\sigma = -1$ if $\mu\nu<0$.
\end{lemma}
\begin{proof}
Defining $w=|\nu|^{1/2}(r-M)$, we rewrite $Z(\nu,\mu)$ as
\begin{align*}
Z(\nu,\mu)&=|\nu|^{-1/2+s-\alpha}\int_{0}^\infty e^{ i\mr{sign}\,\nu|\nu|^{1/2}\lp(w+\mr{sign}(\mu\nu)|\mu| w^{-1}\rp)}w^{2\alpha-2s}h(|\nu|^{-1/2}w+M)dw
\end{align*}
The phase of $Z(\nu,\mu)$ is now the function $f$ satisfying
\begin{gather*}
f(w)=\mr{sign}\,\nu\lp(w+\mr{sign}(\mu\nu)\frac{|\mu|}{w}\rp)\,,\quad f'(w)=\mr{sign}\,\nu\lp(1-\mr{sign}(\mu\nu)\frac{|\mu|}{w^2}\rp)\,,\\ f''(w)=\mr{sign}\,\nu\lp(\frac{2\mr{sign}(\mu\nu)|\mu|}{w^3}\rp)\,.
\end{gather*}
The only critical points of $f$ are at $w_0=\pm\sqrt{\mr{sign}(\mu\nu)|\mu|}$, where
\begin{align*}
f(w_0)&=\begin{dcases} \pm 2i\, \mr{sign}\,\nu |\mu|^{1/2} &\text{~if~}\nu\mu<0 \\
\pm 2\,\mr{sign}\,\nu |\mu|^{1/2} &\text{~if~}\nu\mu>0
\end{dcases}\,,\quad
f''(w_0)=\begin{dcases} \mp 2i\, \mr{sign}\,\nu |\mu|^{-1/2} &\text{~if~}\nu\mu<0 \\
\pm 2\,\mr{sign}\,\nu |\mu|^{-1/2} &\text{~if~}\nu\mu>0
\end{dcases}
\end{align*}

Consider the case $\mu\nu>0$ first. In this case, the only stationary point of $f$ which is in the integration range is $w_0=+|\mu|^{1/2}$. Assuming compact support of $h$, we can apply the stationary phase approximation with large parameter $|\nu|^{1/2}$ (see e.g.\ \cite[Lemma 2.8]{Tao2018}) to obtain 
\begin{align*}
Z(\nu,\mu)& = i^{1/2}\sqrt{\frac{\pi}{|\nu|}}\lp(\frac{\mu}{\nu}\rp)^{1/4-s+\alpha}e^{2i\,\mr{sign}\,\nu\sqrt{|\mu \nu|}}\lp[h\lp(M+\sqrt{|\mu/\nu|}\rp)+O\lp(|\nu|^{-1/2}\rp)\rp]\,.
\end{align*}
Finally, by expanding $h$ with a Taylor series, we get the statement in the case $\mu\nu>0$.

If $\mu\nu<0$, then the critical points of $f$ are in the complex plane, so application of the stationary phase lemma would could only give us at most arbitrarily fast polynomial decay in $|\nu|$. Instead, we want to deform the integral in $w\in[0,\infty)$ into one on the complex plane passing through the critical points and apply the method of steepest descent. Considering Figure~\ref{fig:contour-integration}, we see that, by Cauchy's theorem, since $e^{ i|\nu|^{1/2}f(z)}z^{2\alpha-2s}h(|\nu|^{-1/2}z+M)$ is holomorphic for $z$ in the region contained by $\gamma_1\cup\gamma_2\cup\gamma_3\cup \gamma_4$, we have
\begin{align*}
&\int_{0}^\infty = \lim_{\varepsilon\to 0} \int_{\gamma_2} +\lim_{\varepsilon\to 0,\,R\to \infty}\int_{\gamma_3} +\lim_{R\to \infty}\int_{\gamma_3}\,,
\end{align*}
where we have suppressed the integrand.  To determine which of the two points is approached by $\gamma_3$, we look at the integral over $\gamma_2=\{z=\varepsilon e^{i\theta},\,\theta\in(0,\pi/2)\text{~or~}\theta\in(0,-\pi/2)\}$:
\begin{align*}
&\lp|\int_{\gamma_2}e^{ i|\nu|^{1/2}f(z)}z^{2\alpha-2s}h(|\nu|^{-1/2}z+M)dz \rp| \\
&\qquad\leq \varepsilon \int_{0}^{\pm \pi/2} e^{-\mr{sign}\nu\,|\nu|^{1/2}\lp(\varepsilon+|\mu|\varepsilon^{-1}\rp)\sin\theta}\varepsilon^{-2s} e^{2i\alpha\theta} |h(|\nu|^{-1/2}\varepsilon e^{i\theta}+M)|d\theta \to 0\,,
\end{align*}
as $\varepsilon\to 0$ if and only if $\mr{sign}\,\nu \sin\theta >0$ in the integration range. We thus pick contour (a) in Figure~\ref{fig:contour-integration} for $\nu>0$ and (b) for $\nu<0$. For the integral over $\gamma_4=\{z=Re^{i\theta},\,\theta\,\mr{sign}\,\nu\in(\pi/2,0)\}$, we can show that
\begin{align*}
&\lp|\int_{\gamma_4}e^{ i|\nu|^{1/2}f(z)}z^{2\alpha-2s}h(|\nu|^{-1/2}z+M)dz \rp| \\
&\qquad\leq  \int_{\pi/2}^{0} e^{-|\nu|^{1/2}\lp(R+|\mu|R^{-1}\rp)\sin\theta} e^{2i\alpha\theta} R^{1-2s}|h(|\nu|^{-1/2}R e^{i\theta}+M)|d\theta \,,
\end{align*}
which goes to zero as $R\to \infty$ as long as $h(r)=O(r^{2-2s})$. 

\begin{figure}[htbp]
    \centering
    \begin{subfigure}[t]{0.5\textwidth}
        \centering
 \includegraphics[scale=1]{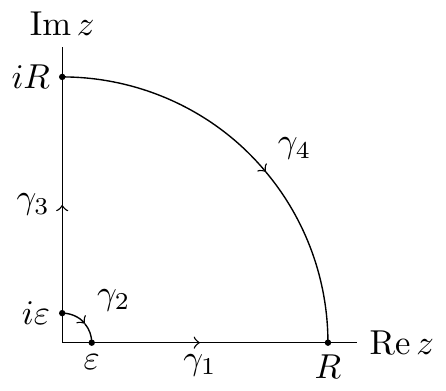}
        \caption{}
    \end{subfigure}%
    ~ 
    \begin{subfigure}[t]{0.45\textwidth}
        \centering
\includegraphics[scale=1]{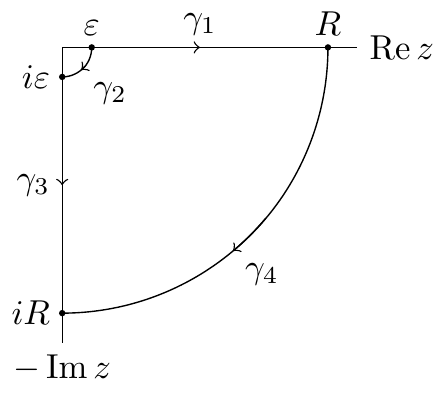}
        \caption{}
    \end{subfigure}
    \caption{Contour integrals for \cref{lemma:technical-core-mode-stability-extremal} in the case $\mu\nu<0$: (a) is applied if $\nu>0$ and (b) when $\nu<0$.}
    \label{fig:contour-integration}
\end{figure}

The only remaining contribution is over $\gamma_3=\{z=i\,\mr{sign}\,\nu y,\, y\in(\varepsilon,R)\}$, along which we have
\begin{align*}
i\,\mr{sign}\,\nu|\nu|^{1/2}f(z) = -|\nu|^{1/2}\lp(y+\frac{|\mu|}{y}\rp) = |\nu|^{1/2}S(y)\,,
\end{align*}
where 
\begin{align*}
S(y)= -\lp(y+\frac{|\mu|}{y}\rp)\,,\quad S'(y)=-\lp(1-\frac{|\mu|}{y^2}\rp)\,,\quad S''(y)=-\frac{2|\mu|}{y^3}
\end{align*}
has a unique maximum at $y_0=|\mu|^{1/2}$ which is in the integration range if $\varepsilon$ is sufficiently small and $R$ is sufficiently large. Moreover, as $S(y)\leq 0$, the integrand is exponentially decaying as $y\to \infty$, so we can apply the Laplace's method (a version of the method of steepest descent) \cite[Theorem 7.1]{Olver2018} to conclude:
\begin{align*}
&|\nu|^{-1/2+s-\alpha}\int_{\gamma_3} e^{ i|\nu|^{1/2}f(z)}z^{2\alpha-2s}h(|\nu|^{-1/2}z+M)dz \\
&\qquad =i^{2\alpha-2s+1}|\nu|^{-1/2+s-\alpha}\int_{\varepsilon}^R e^{ |\nu|^{1/2}S(y)}y^{2\alpha-2s}h(|\nu|^{-1/2}iy+M)dy \\
&\qquad= i^{1/2}\sqrt{\frac{\pi}{|\nu|}}\lp(\frac{\mu}{\nu}\rp)^{1/4-s+\alpha}e^{-2\sqrt{|\mu \nu|}}\lp[h\lp(M+\sqrt{|\mu/\nu|}\rp)+O\lp(|\nu|^{-1/2}\rp)\rp]\,.
\end{align*} 
Expanding $h$ as a Taylor series, we obtain the result in the statement.
\end{proof}

Applying Lemma~\ref{lemma:technical-core-mode-stability-extremal} to $\tilde{g}$, we obtain 
\begin{lemma} \label{lemma:g-tilde-asymptotics-superrad} For any $\omega$ real such that $\omega\neq 0$ and $(\omega-m\upomega_+)\neq 0$, we have
\begin{align*}
\tilde{g}(x)&= e^{4\sigma\sqrt{-\beta\gamma}\lp(x/M\rp)^{1/2}}x^{-3/4-\alpha+s}\times \\
&\qquad \times \lp( \sqrt{\frac{ i \pi M}{2|\omega|}}\lp(\frac{2M^3(\omega-m\upomega_+)}{\omega}\rp)^{1/4-s+\alpha}e^{-2iM\omega}\lp[e^{-\frac{\beta}{r-M}}(r-M)^{-\alpha+2s} R\rp]_{r=M}\rp.\\
&\qquad \lp.\qquad +O(x^{-1/2})\rp) \text{~as~} x\to\infty\,,
\end{align*}
where $\sigma=+1$ if $\beta,\,\gamma<0$ and $\sigma=-1$ otherwise.
\end{lemma}
\begin{proof} 
The method of steepest descent suggests that the greatest contribution to the integral will come from the stationary point of the phase in the oscillatory integral defining $Z$, which will be very close to $r=M$ as $x\to\infty$.

Let $\chi(r)$ be a cutoff function that is identically one close to $r=M$ and vanishes for $[M+2,\infty)$. We allow $\chi$ to depend on $x$, but require $1\leq |\p_r\chi|\leq 2$. We can split the integral in  the definition of $\tilde{g}$ as follows:
\begin{align*}
\tilde{g}(x)&=e^{-2i\omega M}\int_{M}^\infty e^{Ax(r-M)}(r-M)^{2\alpha-2s}e^{2\beta(r-M)^{-1}} \lp(e^{-\frac{\beta}{r-M}}(r-M)^{-\alpha+2s}e^{\gamma r}\chi R\rp) dr \\
&\qquad +\lim_{y\to 0}e^{-2i\omega M}\int_{M}^\infty e^{A(x+iy)(r-M)}(r-M)^{2\alpha-2s}e^{2\beta(r-M)^{-1}} \times \numberthis \label{eq:splitting-asymptotics}\\
&\qquad\qquad\qquad\qquad \times \lp(e^{-\frac{\beta}{r-M}}(r-M)^{-\alpha+2s}e^{\gamma r}(1-\chi) R\rp) dr\,, 
\end{align*}
for $(x,y)\in[2M,\infty)\times\{y\in[-1,1]\colon y\Re\omega\geq 0 \text{~or~} y\omega>0\}$. For the second integral, we appeal to the integration by parts argument of Lemma~\ref{lemma:g-tilde-IBP} to write, for $y\neq 0$ and $z:=x+iy$,
\begin{align}
\begin{split}
&\int_{M}^{\infty} e^{A(z-M)(r-M)}(r-M)^{\alpha}e^{\beta(r-M)^{-1}}e^{\gamma r}R\lp(1-\chi\rp) dr \\
&\qquad =\frac{1}{\lp[A(z-M)\rp]^{k}}\int_{M}^{\infty} e^{A(z-M)(r-M)}\frac{d^k}{dr^k}\lp((r-M)^{\alpha}e^{\beta(r-M)^{-1}}e^{\gamma r}R\lp(1-\chi\rp)\rp) dr\,.
\end{split} \label{eq:polynomial-decay}
\end{align}
The right hand side of \eqref{eq:polynomial-decay} is certainly a convergent integral if $k\geq 2-2s$, in which case it is $O(x^{-2+2s})$ as $x\to \infty$. Moreover, as in Lemma \ref{lemma:reduced-Heun-z} we have shown that $\tilde{g}(x+iy)$ admits a $C_x^\infty([2M,\infty))$ extension, $\tilde{g}(x)$, \eqref{eq:polynomial-decay} must hold for $y=0$. We note that the integrand is $O(r^{-1-2s-k})$ for any $k$.

Set $i\nu=Ax$, $i\mu= 2\beta$. For the non-superradiant frequencies, the result now follows from Lemma~\ref{lemma:technical-core-mode-stability-extremal} by taking $\chi$ to be identically one in $[M,M+1]$ and letting $h(r)$ as the function in brackets in the first integral of \eqref{eq:splitting-asymptotics}. For the superradiant frequencies, we set $\chi$ to be identically one in $[M, M+\sqrt{(M|\beta|)/|\gamma|}x^{-1/2}/2]$ so that both integrals in \eqref{eq:splitting-asymptotics} will include the critical point of the phase; then we apply Lemma~\ref{lemma:technical-core-mode-stability-extremal} for the integral in \eqref{eq:splitting-asymptotics} and the integral \eqref{eq:polynomial-decay} with $h$ being the functions in brackets (note that these satisfy the decay assumptions of Lemma~\ref{lemma:technical-core-mode-stability-extremal} as one is compactly supported and the other has as much polynomial decay as we want by the previous considerations).
\end{proof}

\subsubsection{Proof of Propositions \ref{prop:ode-u-tilde}}
\label{sec:integral-transformation-smooth}

We are now ready to prove \cref{prop:ode-u-tilde}:
\begin{proof}[Proof of \cref{prop:ode-u-tilde}]
Recall that out integral transformation $\tilde{u}$ is given by (\ref{eq:extremal-transformation-with-g-tilde})
\begin{align*}
\tilde{u}(x)=(x^2+2M^2)^{1/2}(x-M)^{-s}(x-2M)^{\alpha}\tilde{g}(x)\,,
\end{align*}
with derivative
\begin{align} \label{eq:derivative-u-tilde}
\tilde{u}'=\frac{(x-2M)(x-M)}{x^2+2M^2}\lp[\frac{x}{x^2+2M^2}-\frac{s}{x-M}+\frac{\alpha}{x-2M}+\frac{1}{\tilde{g}}\frac{d\tilde{g}}{dx}\rp]\tilde{u}\,.
\end{align}
By \cref{lemma:reduced-Heun-z}, $\tilde{g}$ is a smooth function for $x\in(2M,+\infty)$, so $\tilde{u}$ must also be smooth in $x\in(2M,+\infty)$. The ODE \eqref{eq:ode-u-tilde} for $\tilde{u}$ can be computed directly from the ODE \eqref{eq:ode-g-tilde} for $\tilde{g}$. This shows statements 1 and 2.

By \cref{lemma:g-tilde-IBP}, $\tilde{g}(x)$ and $\p_x \tilde{g}(x)$  are bounded functions on $x\in[2M,+\infty)$, so $\tilde{u}$ and $\tilde{u}'$ must be bounded as well, which shows statement 3. In the previous section, we have computed the asymptotics for $\tilde{g}$ which, combined with \eqref{eq:derivative-u-tilde}, yield the boundary conditions for $\tilde{u}$ as in statement 4.

Finally, we must show injectivity; we will proceed as in \cite[proof of Theorems 1.5, 1.6]{Shlapentokh-Rothman2015}. Extend $R(r)$ to $r\in \mathbb{R}$ by 0.
\begin{itemize}
\item If $\Im\omega> 0$, we write the Fourier transform of  $(r-M)^{\alpha}e^{\beta(r-M)^{\gamma}}e^{\gamma r}R(r)$ as
\begin{align*}
\hat{R}(z):=\frac{1}{2|\omega|^2}\int_{-\infty}^\infty e^{2i|\omega|^2 z(r-M)}(r-M)^{\alpha}e^{\beta(r-M)^{\gamma}}e^{\gamma r}R(r)dr\,,
\end{align*}
which is a holomorphic function in the upper half plane (by the Paley--Wiener theorem, due to $R$ being an $L^2$ function of $r\in(M,\infty)\subseteq \mathbb{R}_+$). However, since
\begin{align*}
\int_{M}^\infty e^{\frac{2i\omega}{M}(z-M)(r-M)}(r-M)^{\alpha}e^{\beta(r-M)^{\gamma}}e^{\gamma r}R(r)dr \equiv 0\,, 
\end{align*}
for all $x\in(2M,+\infty)$, we conclude that $\hat{R}(z)$ vanishes along the line $\{z\in\mathbb{C}\,:\, z=y/\overline{\omega},\, y\in(2,+\infty)\}$. By analyticity, $\hat{R}$ and hence $R$ vanish identically.

\item If $\omega\in\mathbb{R}\backslash\{0\}$, then we write the Fourier transform of $(r-M)^{\alpha}e^{\beta(r-M)^{\gamma}}e^{\gamma r}R(r)$ as the $L^2_{y}(\mathbb{R})$ function given by
\begin{align*}
\hat{R}(y):=\frac{1}{2\omega}\int_{-\infty}^\infty e^{2i\omega y (r-M)}(r-M)^{\alpha}e^{\beta(r-M)^{\gamma}}e^{\gamma r}R(r)dr\,.
\end{align*}
Since
\begin{align*}
\int_{M}^\infty e^{\frac{2i\omega}{M}(z-M)(r-M)}(r-M)^{\alpha}e^{\beta(r-M)^{\gamma}}e^{\gamma r}R(r)dr \equiv 0\,,
\end{align*}
for all $x\in(2M,+\infty)$, we conclude that $\hat{R}(y)$ vanishes along the line $y\in(2,+\infty)$. However, the Fourier transform of a non-trivial $L^2$ function supported in $\mathbb{R}_+$ cannot vanish. Since $R$ is an $L^2$ function in its support contained in $(M,+\infty)\subseteq\mathbb{R}_+$, we conclude that $R$ vanishes identically.
\end{itemize}
Hence \eqref{eq:extremal-transformation} defines an injective map $R \mapsto \tilde{u}$, as in statement 5.
\end{proof}

\subsection{Whiting's integral radial transformation for subextremal Kerr \texorpdfstring{($|a|\leq M$)}{}}
\label{sec:integral-transformation-subextremal}

For a \textit{subextremal} Kerr background, $|a|<M$, we rewrite the radial ODE \eqref{eq:radial-ODE} as, dropping subscripts
\begin{equation}
\begin{gathered}
\frac{d}{dr}\lp[(r-r_+)(r-r_-)\frac{d}{dr}\rp]R(r)+s(r-r_++r-r_-)\frac{d}{dr}R(r)\\
+ \gamma \lp[(s-1)(r-r_++r-r_-)-2\eta(r-r_+)-2\xi(r-r_-)\rp]R(r)\\ + 2\gamma(1-s)(r-r_-)+\gamma(1-s)(r_+-r_-)-L-2\eta\xi\\
+\frac{\eta(s-\eta)(r-r_+)^2+\xi(s-\xi)(r-r_-)^2}{(r-r_-)(r-r_+)}-\gamma^2(r-r_+)(r-r_-)R(r)
=\Delta \hat{F}(r)\,,
\end{gathered}\label{eq:radial-ODE-eta-xi-gamma}
\end{equation}
where $\hat{F}(r)$ is a smooth inhomogeneity compactly supported away from $r=r_+$ and $r=\infty$ and where we have defined
\begin{align}
\xi := i\frac{am -2M r_+\omega}{r_+-r_-} \,,\quad  \eta := i\frac{2Mr_-\omega -am}{r_+-r_-}\,,\quad \gamma= -i\omega\,, \quad L:= \lambda+a^2\omega^2-2am\omega \,. \label{eq:def-eta-xi-gamma}
\end{align}

We will use Whiting's integral transformation for the radial ODE. The following proposition combines results in \cite{Whiting1989,Shlapentokh-Rothman2015} and extends them:
\begin{proposition} \label{prop:ode-u-tilde-sub} Fix $|a|< M$, $s\in \frac12 \mathbb{Z}_{\leq 0}$ and some frequency parameters $(\omega,m,\lambda)\in\mathbb{C}\times \mathbb{Z}\times \mathbb{C}$ with $\Im\omega> 0$ or $(\omega,m,\lambda)\in\mathbb{R}\backslash\{0\}\times \mathbb{Z}\times \mathbb{R}$. Let $R_{ml}^{[s],(a\omega)}$ be a solution to the radial ODE~\eqref{eq:radial-ODE-eta-xi-gamma} with boundary conditions as in Definition~\ref{def:mode-solution}. Dropping subscripts, define $\tilde{u}$ via the integral transform 
\begin{align}
\begin{split}
\tilde{u}(x) &:=\lim_{y\to 0}(x^2+a^2)^{1/2}(x-r_-)^{-s}(x-r_+)^{\xi+\eta}  \times\\
&\qquad\times \int_{r_+}^\infty e^{-\frac{2\gamma}{r_+-r_-}(x+iy-r_-)(r-r_-)}(r-r_-)^{\eta}(r-r_+)^\xi e^{-\gamma r}R(r) dr\,,
\end{split} \label{eq:subextremal-transformation}
\end{align}
with $\xi$, $\eta$ and $\gamma$ as in \eqref{eq:def-eta-xi-gamma}, where the limit is taken in the function space $L^2_x([r_+,+\infty))$. Introduce a new coordinate $x^*(x):(r_+,+\infty)\to (-\infty,\infty)$ by
\begin{align*}
\frac{dx^*}{dx}= \frac{x^2+a^2}{(x-r_-)(x-r_+)}\,,\quad x^*(3M)=0\,.
\end{align*}
Then the following hold:
\begin{enumerate}
\item $\tilde{u}(x)$ is in fact smooth for $x\in(r_+,+\infty)$;
\item $\tilde{u}$ satisfies the ODE
\begin{align}
\tilde{u}''+\tilde{V}\tilde{u}=\frac{(x-r_+)(x-r_-)}{x^2+a^2}\tilde{H}\,, \label{eq:ode-u-tilde-sub}
\end{align}
where the inhomogeneity $\tilde{H}$ is given by
\begin{align}
\begin{split}
\tilde{H}&:=(x^2+a^2)^{1/2}(x-r_-)^{-s}(x-r_+)^{\xi+\eta}  \times\\
&\qquad\times \int_{r_+}^\infty e^{-\frac{2\gamma}{r_+-r_-}(x-r_-)(r-r_-)}(r-r_-)^{\eta}(r-r_+)^\xi e^{\gamma r}\hat{F}(r) dr\,,
\end{split}\label{eq:H-tilde-sub}
\end{align}
and the potential is
\begin{align*}
\tilde{V}(x)&:= \frac{(x-r_-)(x-r_+)}{(x^2+a^2)^2}\lp[(x-r_+)(x-r_-)+\frac{4M^2(x-r_-)}{r_+-r_-}+a^2\rp]\omega^2\\
&\qquad+\frac{4M^2(x-r_-)^2}{(x^2+a^2)^2}\omega^2 -\frac{(x-r_-)(x-r_+)}{(x^2+a^2)^2}s^2\\
&\qquad-\frac{(x-r_-) (x-r_+)}{\left(x^2+a^2\right)^2}(\lambda+s)  -\frac{4 m \omega  (x-M) (x-r_+) (x-r_-)}{\left(x^2+a^2\right)^2}\numberthis \label{eq:V-tilde-sub} \\ 
&\qquad-\frac{(x-r_+)(x-r_-)}{(x^2+a)^4}\lp[2M^2 (x-r_-)(x-r_+)+2Mx(x^2-a^2)\rp]\,.
\end{align*}
\item $\tilde{u}$ and $\tilde{u}'$ are bounded for $x^*\in\mathbb{R}$;
\item $\tilde{u}$ and $\tilde{u}'$ satisfy the boundary conditions
\begin{enumerate}[label=(\alph*)]
\item  if $\Im\omega>0$, then 
\begin{enumerate}
\item $\tilde{u}',\tilde{u}=O\lp((r-r_+)^{2M\Im\omega}\rp)$ as $x\to r_+$,
\item $\tilde{u}',\tilde{u}=O\lp(e^{-\Im \omega\, x}\rp)$ as $x\to \infty$,
\end{enumerate}
\item if $\omega\in\mathbb{R}\backslash\{0\}$, then
\begin{enumerate}
\item $\tilde{u}'+ i\omega \frac{r_+-r_-}{r_+} \tilde{u}=O(r-r_+)$ as $x\to r_+$,
\item $\tilde{u}'-i\omega \tilde{u}= O(x^{-1})$ as $x\to\infty$ with
$$|\tilde{u}(\infty)|^2 = \frac{|\Gamma(2\xi-s+1)|^2(r_+-r_-)^2}{8Mr_+\omega^2}\lp(\frac{(r_+-r_-)^2}{2|\omega|}\rp)^{2-2s}\lp|(\Delta^{s/2}u)(-\infty)\rp|^2\,,$$
where $\Gamma(z):=\int_0^\infty e^{-t}t^{z-1}$
is the Gamma function;
\end{enumerate} 
\end{enumerate}
\item the integral transformation \eqref{eq:subextremal-transformation} defines an injective map $R\mapsto \tilde{u}$: if $\tilde{u}$ vanishes identically, then $R$ must also vanish identically.
\end{enumerate}
\end{proposition}

\begin{remark} \label{rmk:prop-potential-sub}
For our proof in \cref{sec:proof-s-negative}, it will be useful to highlight the following properties of $\tilde{V}(x)$ for $x\in(r_+,\infty)$:
\begin{enumerate}[label=(\roman*)]
\item $\omega^2$ has a positive coefficient;
\item $\lambda+s$ has a non-positive coefficient;
\item the $(\omega,m,\lambda)$-independent part of $\tilde{V}$ is non-negative;
\item $\tilde{V}$ is real whenever $\omega$ is real;
\item $\tilde{V}=\omega^2+O(x^{-2})$ as $x\to\infty$ and $\tilde{V}=\omega^2(r_+-r_-)^2/r_+^2+O(x-r_+)$ as $x\to r_+$\,.
\end{enumerate}
These properties follow easily from \eqref{eq:V-tilde-sub}.
\end{remark}

To prove Proposition~\ref{prop:ode-u-tilde-sub}, we will follow the same steps as for the identical result, Proposition~\ref{prop:ode-u-tilde}, for $|a|=M$. We break up \eqref{eq:subextremal-transformation} into smaller pieces which we will analyze separately: we define the auxiliary function
\begin{align}
g(r):=(r-r_-)^{-\eta+s} (r-r_+)^{-\xi+s}e^{-\gamma r} R(r)\,, \label{eq:def-g-sub}
\end{align}
such that, if $R(r)$ satisfies the outgoing boundary conditions in Definition \ref{def:outgoing-radial-solution}, we have
\begin{align}
\begin{alignedat}{3} \label{eq:g-bdry-sub}
g(r) &= \sum_{k=0}^\infty b_k (r-r_+)^k &\text{~ as } r\to r_+\,, \\
g(r) &= e^{-2\gamma r}r^{\xi+\eta-1}\lp[\sum_{k=0}^N c_k r^{-k}+O\lp(r^{-N-1}\rp)\rp] &\text{~ as } r\to \infty\,, 
\end{alignedat}
\end{align}
and the weighted integral of $g$
\begin{align}
\begin{split}
\tilde{g}(z) &:= \int_{r_+}^\infty e^{A(z-r_-)(r-r_-)}(r-r-)^{2\eta-s} (r-r_-)^{2\xi-s}e^{2\gamma r}g(r) dr \\
&=\int_{r_+}^\infty e^{A(z-r_-)(r-r_-)}(r-r_-)^{\eta}(r-r_+)^\xi e^{\gamma r}R(r) dr \,,
\end{split}\label{eq:def-g-tilde-sub}
\end{align}
where $A(r_+-r_-)=-2\gamma=i\omega$ and $z=x+iy$, with $z=x+iy$ for $(x,y)\in[r_+,\infty)\times\{y\in[-1,1]\colon y\Re\omega\geq 0 \text{~or~} y\omega>0\}$. With these definitions, the integral transformation \eqref{eq:subextremal-transformation} becomes simply
\begin{align}
\tilde{u}(x)=(x^2+a^2)^{1/2}(x-r_-)^{-s}(x-r_+)^{\xi+\eta}\lim_{y\to 0}\tilde{g}(x+iy)\,, \label{eq:subextremal-transformation-with-g-tilde}
\end{align}
where the limit is, \textit{a priori}, to be taken in the function space $L^2_x([r_+,+\infty))$.

We will prove \cref{prop:ode-u-tilde-sub} over the next sections: first, in \cref{sec:integral-transformation-well-defined-sub}, we show that $\tilde{g}(x)$ is well-defined and that it is $C^{1,1/2}_x([r_+,\infty))$ and bounded; then, in \cref{sec:ode-auxiliary-sub}, we show that it is in fact a smooth solution to a second order ODE. We also obtain precise asymptotics for $\tilde{g}(x)$ in Section~\ref{sec:asymptotics-g-tilde-sub}. Finally, in \cref{sec:integral-transformation-smooth-sub} we put these together to prove \cref{prop:ode-u-tilde-sub}. 

\subsubsection{Defining the integral transformation for real \texorpdfstring{$\omega$}{time frequency}}
\label{sec:integral-transformation-well-defined-sub}

In this section we will show that \eqref{eq:subextremal-transformation-with-g-tilde} is well-defined, by understanding the limit of $\tilde{g}(x+iy)$ as $y\to 0$ for $x\in[r_+,\infty)$. Note that the integrand in \eqref{eq:subextremal-transformation-with-g-tilde} is $O((r-r_+)^{-s})$ as $r\to r_+$, hence integrable (as $s<0$); as $r\to \infty$, since the term $e^{2\gamma r}$ exactly matches the exponential in the boundary condition of $g$ at $r=\infty$ in \eqref{eq:g-bdry-sub}, the integrand is $O(e^{-\Re \omega\,y r-\Im \omega\, x r}r^{-1-2s})$, which is not integrable if $\Im\omega=y=0$. To define $\tilde{g}$ properly in the limit $y\to 0$, we will integrate by parts to produce more decay near $r=\infty$:

\begin{lemma} \label{lemma:g-tilde-IBP-sub} Fix $s\leq 0$ and  $\omega\in\mathbb{C}$ such that $\omega\in\mathbb{R}\backslash\{0\}$ or $\Im\omega>0$. Let  $z:=x+iy$, where $z=x+iy$ for $(x,y)\in[r_+,\infty)\times\{y\in[-1,1]\colon y\Re\omega\geq 0 \text{~or~} y\omega>0\}$. Let $\varepsilon>0$ be arbitrary; we have
\begin{align*} 
\tilde{g}(z)&=\int_{r_+}^{r_++\varepsilon} e^{A(z-r_-)(r-r_-)}(r-r_-)^{2\eta-s}(r-r_+)^{2\xi-s}e^{2\gamma r} g(r) dr \numberthis\label{eq:g-tilde-IBP-sub}\\
&\qquad +\frac{1}{\lp[A(z-r_-)\rp]^{1-2s}}\int_{r_++\varepsilon}^{\infty} \lp\{e^{A(z-r_-)(r-r_-)}\times\rp.\\
&\qquad\qquad\qquad\qquad\qquad\qquad \lp. \times\lp(\frac{d}{dr}\rp)^{1-2s}\lp((r-r_-)^{2\eta-s}(r-r_+)^{2\xi-s}e^{2\gamma r} g(r)\rp)\rp\} dr \\
&\qquad+\sum_{k=1}^{1-2s} \frac{(-1)^k}{[A(z-r_-)]^k} \lp[e^{A(z-r_-)(r-r_-)}\rp.\times\\
&\qquad\qquad\qquad\qquad\qquad\qquad\lp.\times\lp(\frac{d}{dr}\rp)^{k-1}\lp((r-r_-)^{2\eta-s}(r-r_+)^{2\xi-s}e^{2\gamma r} g(r)\rp)\rp]_{r=r_++\varepsilon}\,,
\end{align*}
and similarly, for $j=1,2$,
\begin{align*} 
&\frac{1}{A^j}\frac{\p^j}{\p z^j}\tilde{g}(z)\\
&\quad= \int_{r_+}^{r_++\varepsilon} e^{A(z-r_-)(r-r_-)}(r-r_-)^{2\eta-s+j}(r-r_+)^{2\xi-s}e^{2\gamma r} g(r) dr \numberthis \label{eq:derivative-g-tilde-IBP-sub} \\
&\quad\qquad +\frac{1}{\lp[A(z-r_-)\rp]^{1+j-2s}}\int_{r_++\varepsilon}^{\infty}\lp\{ e^{A(z-r_-)(r-r_-)}\times\rp.\\
&\quad\qquad\qquad\qquad\qquad\qquad\qquad \lp. \times \lp(\frac{d}{dr}\rp)^{2+j-2s}\lp[(r-r_-)^{2\eta-s+j}(r-r_+)^{2\xi-s}e^{2\gamma r} g(r)\rp]\rp\} dr \\
&\quad\qquad+\sum_{k=1}^{1+j-2s} \frac{(-1)^k}{[A(z-r_-)]^k} \lp[e^{A(z-r_-)(r-r_-)}\times\rp.\\
&\quad\qquad\qquad\qquad\qquad\qquad\qquad \lp.\times \lp(\frac{d}{dr}\rp)^{k-1}\lp((r-r_-)^{2\eta-s+j}(r-r_+)^{2\xi-s}e^{2\gamma r} g(r)\rp)\rp]_{r=r_++\varepsilon}\,.
\end{align*}

\noindent Moreover, $\tilde{g}(x+iy)$ admits a unique extension to $(x,y)\in[r_+,\infty)\times\{y\in[-1,1]\colon y\Re\omega\geq 0\}$ such that
\begin{enumerate}[label=(\roman*)]
\item $\tilde{g}(x+iy)\to \tilde{g}(x)$ in $H^2_x([r_+,+\infty))$ and pointwise in $C^{1,1/2}_x([r_+,+\infty))$ as $y\to 0$;
\item $\tilde{g}(x)$ and its weak derivative, $\p_x\tilde{g}(x)$, are bounded.
\end{enumerate}
\end{lemma}
\begin{proof}
To obtain formulas \eqref{eq:g-tilde-IBP-sub} and \eqref{eq:derivative-g-tilde-IBP-sub}, we proceed in the same way as in the proof of Proposition~\ref{prop:ode-u-tilde} for $|a|=M$. In \eqref{eq:def-g-tilde-sub}, we split the integration range at $r_++\varepsilon$, obtaining a sum of an integral on $[r_++\varepsilon]$ and an integral on $[r_++\varepsilon,+\infty)$. The latter can be integrated by parts, noting that
\begin{align*}
&e^{-A(z-r_-)(r-r_-)} \frac{d}{dr}\lp(e^{A(z-r_-)(r-r_-)} \rp)=A(z-r_-)\,.
\end{align*}
The boundary terms generated by the integration by parts at infinity will vanish when either $\Im\omega>0$ or $y\omega >0$ due to the exponential decay that was similarly present in the identical procedure for $\tilde{g}$. One can repeat this procedure for $\Im\omega>0$ or $y\omega >0$ for any number of derivatives by first differentiating under the integral. 

We have now obtained formulas \eqref{eq:g-tilde-IBP-sub} and \eqref{eq:derivative-g-tilde-IBP-sub}. In these formulas,  since $|z|\geq x \geq r_+>r_-$, the right hand side is bounded and $O(x^{-1})$ as $x\to \infty$ for any $y\in\{y\in[-1,1]\colon y\Re\omega\geq 0\}$. We also note that $\tilde{g}$ is holomorphic for $y\neq 0$ in this range, so we can write
\begin{align}
\begin{split}\label{eq:bound-tilde-g-sub}
&|\tilde{g}(z)|^2+|\p_y\tilde{g}(z)|^2+|\p_x\tilde{g}(z)|^2+|\p_y\p_x\tilde{g}(z)|^2+|\p_x^2\tilde{g}(z)|^2+|\p_y\p_x^2\tilde{g}(z)|^2\\
&\qquad\lesssim |\tilde{g}(z)|^2+|\p_x\tilde{g}(z)|^2+|\p_x^2\tilde{g}(z)|^2+|\p_x^3\tilde{g}(z)|^2\lesssim \frac{1}{x^2}\,,
\end{split}
\end{align}
for any $(x,y)\in[r_+,\infty)\times\{y\in[-1,1]\colon y\Re\omega\geq 0\}$. By the same argument as in Proposition~\ref{prop:ode-u-tilde}, we can now infer (i) and (ii) from \eqref{eq:bound-tilde-g-sub}.
\end{proof}

\begin{remark} In our treatment of the subextremal case, we have chosen to interpret \eqref{eq:subextremal-transformation} as an oscillatory integral, the approach followed by \cite{Shlapentokh-Rothman2015}. An alternative subsequent approach is to substitute integration along the real axis by integration along a suitable contour in the complex plane, as in \cite{Andersson2017}; this then enables the authors to define the transformation for $s>0$ via integration along a (different) suitable complex path. 

As one of our goals is to provide a quantitative mode stability statement (Theorem~\ref{thm:quantitative-intro}) in the spirit of that in \cite{Shlapentokh-Rothman2015}, our choice in defining \eqref{eq:subextremal-transformation} is motivated by the desire to present a unified picture for the entire black hole parameter range $|a|\leq M$ that can be directly compared with the latter work.
\end{remark}

\subsubsection{Differential equations for the auxiliary functions}
\label{sec:ode-auxiliary-sub}

Let $\mc{T}_r$ be the confluent Heun operator \cite{Ronveaux1995} given by
\begin{align}
\begin{split}
\mc{T}_r &:= (r-r_+)(r-r_-)\frac{d^2}{dr^2}+2\gamma(1-2s)(r-r_-)+2\gamma (1-s)r_- -2s-L\\
&\qquad+\lp[(2\eta+1-s)(r-r_+)+(2\xi+1-s)(r-r_-) +2\gamma(r-r_+)(r-r_-)\rp]\frac{d}{dr}\,.
\end{split}\label{eq:confluent-operator}
\end{align}

Given the definition of $g$ in \eqref{eq:def-g-sub}, since $R(r)$ is a solution to the radial ODE~\eqref{eq:radial-ODE-eta-xi-gamma} with inhomogeneity $\hat{F}$, we find that $g$ satisfies $\mc{T}_r g =G$ where
\begin{align*}
G(r):=(r-r_-)^{-\eta+s}(r-r_+)^{-\xi+s}e^{\gamma r} \hat{F}(r)
\end{align*}
and $\xi$, $\eta$ and $\gamma$ are defined in \eqref{eq:def-eta-xi-gamma}. 

Following \cite{Shlapentokh-Rothman2015}, we would now like to show that $\tilde{g}$ also satisfies a confluent Heun equation with some new parameters (see also \cite{Kazakov1996,Schmidt1995}). Since
\begin{align*}
&e^{-A(z-r_-)(r-r_-)}\mc{T}_r e^{A(z-M)(r-M)}\\
&\quad = (z-r_+)(z-r_-)A^2 (r-r_-)^2\\
&\quad\qquad +[A(r_+-r_-)+p][A(r-r_+)(r-r_-)(z-r_-)+(1-2s)2\gamma] \\
&\quad\qquad+\lp[(1-2s)(z-r_+)+(1+2\xi+2\eta)(z-r_-)\rp]A(r-r_-)\\
&\quad\qquad -A(r_+-r_-)(z-r_+)(z-r_-)A(r-r_-) \\
&\quad\qquad-A(r_+-r_-)(2\eta+1-s)(z-r_-) +2\gamma(1-s)r_--2s-L
\end{align*}
if we pick $A=-2\gamma/(r_+-r_-)$, we have $\tilde{\mc{T}}_x e^{A(z-r_-)(r-r_-)} = \mc{T}_r e^{A(z-r_-)(r-r_-)}$, where
\begin{align}
\begin{split}
\tilde{\mc{T}}_x &:= (r-r_+)(r-r_-)\frac{d^2}{dx^2} +2\gamma(2\eta+1-s)(x-r_-)+2\gamma (1-s)r_- -2s-L\\
&\qquad+\lp[(1-2s)(x-r_+)+(1+2\xi+2\eta)(x-r_-) +2\gamma(x-r_+)(x-r_-)\rp]\frac{d}{dx} 
\end{split} \label{eq:confluent-operator-tilde}
\end{align}
is another confluent Heun operator with different parameters.

\begin{lemma}\label{lemma:Heun-z} Suppose  $\Im\omega>0$ or $\omega\in\mathbb{R}\backslash\{0\}$. Let $z=x+iy$ with $(x,y)\in[2M,\infty)\times\{y\in[-1,1]\colon y\Re\omega\geq 0 \text{~or~} y\omega>0\}$.  For $y\neq 0$, we have $\tilde{\mc{T}}_x\tilde{g}(z)=\tilde{G}(z)$, where $\tilde{\mc{T}}_x$ is defined by \eqref{eq:confluent-operator-tilde} and 
\begin{align*}
\tilde{G}(z) := \int_{r_+}^\infty e^{A(z-r_-)(r-M)}(r-r_-)^{2\eta-s} (r-r_+)^{2\xi-s}e^{2\gamma r}G(r) dr\,.
\end{align*}
Moreover, if $y=0$, $\tilde{g}(x)$ is smooth for $x\in(r_+,+\infty)$ and satisfies
\begin{equation}
\tilde{\mc{T}}_x\tilde{g}(x)= \tilde{G}(x) \label{eq:ode-g-tilde-sub}
\end{equation}
classically.
\end{lemma}
\begin{proof}
If $\Im\omega>0$ with $y\Re\omega\geq 0$ or $y\omega>0$, the integral is absolutely convergent and we can thus differentiate under the integral to obtain
\begin{align*}
\tilde{\mc{T}}_x\tilde{g} &= \int_{r_+}^\infty \tilde{\mc{T}}_x \lp(e^{A(z-r_-)(r-r_-)}\rp) (r-r_-)^{2\eta-s} (r-r_+)^{2\xi-s}e^{2\gamma r}g(r) dr \\
&=\int_{r_+}^\infty \mc{T}_r \lp(e^{A(x-r_-)(r-r_-)}\rp) (r-r_-)^{2\eta-s} (r-r_-)^{2\xi-s}e^{2\gamma r}g(r) dr \\
&= \tilde{G}(z)+ \lp[\lp(A(x-r_-)g-\frac{dg}{dr}\rp)(r-r_-)^{2\eta-s+1}(r-r_+)^{2\xi-s+1}e^{2\gamma r}e^{A(z-r_-)(r-r_-)}\rp]_0^\infty\,,
\end{align*}
where we have made use of the properties of the kernel, the equation for $g$ and the lemma
\begin{lemma} Let $\mc{T}_r$ be a differential operator as defined in \eqref{eq:confluent-operator}. Then, for sufficiently regular functions $f$ and $h$, after integration by parts,
\begin{align*}
&\int_{A_1}^{A_2} \lp(h\mc{T}_r f-f\mc{T}_r h \rp)(r-r_-)^{2\eta-s}(r-r_+)^{2\xi-s}e^{2\gamma r} dr \\
&\quad=\lp[\lp(h\frac{df}{dr}-f\frac{dh}{dr}\rp)(r-r_-)^{2\eta-s+1}(r-r_+)^{2\xi-s+1}e^{2\gamma r}\rp]_{r=A_1}^{r=A_2}\,.
\end{align*}
\end{lemma}
\noindent This lemma justifies the introduction of weights $(r-r_-)^{2\eta-s}(r-r_+)^{2\xi-s}e^{2\gamma r}$: when evaluated against such a measure, $\mc{T}_r$ is self-adjoint. 

We still need to show that the boundary terms vanish. This is clear at $r=\infty$, due to the decay of $g$ and $dg/dr$, together with the exponential decay brought in by either $\Re\omega y>0$ or $\Im \omega >0$. At the horizon, the extra factor of $(r-r_+)$ makes the boundary term vanish, due to the boundary conditions (\ref{eq:g-bdry-sub}) of $g$.

To consider the case $y=0$, first note that $\tilde{G}(x+iy)$ is smooth even for $y=0$, by compact support\footnote{Note that $\tilde{G}(x+iy)$ can be rewritten as the Fourier transform of a compactly supported function; see the proof of statement 5 of Proposition~\ref{prop:ode-u-tilde-sub} in Section~\ref{sec:integral-transformation-smooth-sub}.} of $\hat{F}$. 
We have shown that $\tilde{g}(x+iy)$, for $y\neq 0$, satisfies the ODE classically, hence, it also satisfies the ODE  weakly as an $H^1_x([r_+,\infty))$ function. Thus, for any smooth $\varphi$ compactly supported in $(r_+,\infty)$,
\begin{gather*}
\begin{gathered}
\int_{r_+}^\infty \frac{d}{dx}\tilde{g}(x)\lp[2(\xi+\eta)(x-r_-)-2s(x-r_+)+2\gamma(x-r_+)(x-r_-)\rp]\varphi(x)dx\\
+\int_{r_+}^\infty \frac{d}{dx}\tilde{g}(x)(x-r_-)(x-r_+)\frac{d\varphi}{dx}(x)dx\\
+\int_{r_+}^\infty g(x)\lp[2\gamma(1-s)r_--2s-L+2\gamma(2\eta+1-s)(x-r_-)\rp]\varphi(x)dx
-\int_{r_+}^\infty \tilde{G}(x)\varphi(x)dx \\
=\int_{r_+}^\infty \lp[\frac{d}{dx}\tilde{g}(x)-\frac{d}{dx}\tilde{g}(x+iy)\rp]\lp[2(\xi+\eta)(x-r_-)-2s(x-r_+)\rp]\varphi(x)dx\\
+\int_{r_+}^\infty \lp[\frac{d}{dx}\tilde{g}(x)-\frac{d}{dx}\tilde{g}(x+iy)\rp]\lp[(x-r_+)(x-r_-)\frac{d\varphi}{dx}(x)+2\gamma(x-r_+)(x-r_-)\varphi(x)\rp]dx\\
+\int_{r_+}^\infty \lp[g(x)-g(x+iy)\rp]\lp[2\gamma(1-s)r_--2s-L+2\gamma(2\eta+1-s)(x-r_-)\rp]\varphi(x)dx\\
-\int_{r_+}^{\infty}\lp[\tilde{G}(x)-\tilde{G}(x+iy)\rp]\varphi(x)dx\,.
\end{gathered}
\end{gather*}
As the left hand side is independent of $y$, we can take a limit as $y\to 0$ of the right hand side. Recall that $\tilde{g}(x+iy)\to\tilde{g}(x)$ in $C^{1,1/2}_x([r_+,+\infty))$, by Lemma~\ref{lemma:g-tilde-IBP-sub}, with uniform convergence in the compact support of the test function $\varphi$. Hence, we can exchange the limit with the integral, which shows that $\tilde{g}(x)$ is a weak $H^1_x([r_+,+\infty))$ solution to the ODE as well. As $\tilde{G}$ is smooth, by elliptic regularity of $\tilde{\mc{T}}_x$ we conclude that $\tilde{g}(x)$ is in fact smooth for $x\in(r_+,+\infty)$ and satisfies the ODE classically.
\end{proof}

\subsubsection{Asymptotics of \texorpdfstring{$\tilde{g}$}{tilded g} for large \texorpdfstring{$x$}{x}}
\label{sec:asymptotics-g-tilde-sub}

By Lemma~\ref{lemma:g-tilde-IBP-sub}, $\tilde{g}(x)$ is bounded and  $C^2_x([r_+,\infty))$, which uniquely determines the behavior at $r=r_+$:
\begin{align*}
\tilde{g}(x)&= \sum_{k=0}^\infty b_k (x-r_+)^k \text{~ as } x\to r_+\,.
\end{align*}
As $r\to \infty$, since $\tilde{g}$ satisfies \eqref{eq:ode-g-tilde-sub} by Lemma~\ref{lemma:Heun-z}, if the inhomogeneity vanishes, $\tilde{g}$ must be superposition of, for $N\in\mathbb{N}$,
\begin{gather*}
\exp\lp(-\gamma x\rp) x^{-1+s-2\xi}\lp[\sum_{k=0}^N c_kx^{-k} +O(x^{-N-1})\rp] \text{~ as } r\to \infty\,,\\
  x^{-\frac{2\gamma(2\eta+1-s)}{4|\omega|^2}}\lp[\sum_{k=0}^N c_kx^{-k}+O\lp(x^{-N-1}\rp)\rp]  \text{~ as } r\to \infty\,.
\end{gather*}
Boundedness of $\tilde{g}$ can be used to rule out the second behavior. However, the asymptotics at this end can be derived in several other ways for which we do not require the assumption $\tilde{H}\equiv 0$. For instance, following \cite[Lemma 4.4]{Shlapentokh-Rothman2015}, one can use the formulas in Lemma \ref{lemma:g-tilde-IBP-sub} to show that $\tilde{g}\sim O(x^{-1+s})$ as $x\to\infty$ and
\begin{align*}
\frac{d}{dx}\tilde{g}-A(r_+-r_-)\tilde{g} =  O(x^{-2+s}) \text{~as~}x\to\infty\,.
\end{align*}
Alternatively, on the real axis, we can generalize \cite[Lemma 4.10]{Shlapentokh-Rothman2015} to obtain the precise asymptotics for spin $s$:
\begin{lemma} \label{lemma:sub-g-tilde-asymptotics} As $x\to\infty$, if $\omega\in\mathbb{R}\backslash\{0\}$, we have
\begin{align*}
\tilde{g}(x)&=\Gamma(2\xi-s+1)\exp\lp(\frac{i\pi}{2}(2\xi-s+1)\rp)\lp(\frac{2|\omega|}{r_+-r_-}\rp)^{-2\xi+s-1}e^{-i\omega (M+r_-)}\\
&\quad\times(r_+-r_-)^\eta\lp((r-r_+)^{-\xi+s}R(r)\rp)\Big|_{r=r_+}e^{-2\gamma x}x^{-2\xi+s-1} +O\lp(x^{-2+s}\rp)\,.
\end{align*}
\end{lemma}
\begin{proof} Let $\chi(r)$ be a cutoff function that is identically one in $[r_+,r_++1]$ and vanishes for $[r_++2,\infty)$. We can split the integral in  the definition of $\tilde{g}$ as follows:
\begin{align}
\tilde{g}(x)&=\int_{r_+}^\infty e^{Axr}(r-r_+)^{2\xi-s} \lp(e^{-A(r-r_-)r_-}e^{-Axr_-}e^{-i\omega r}(r-r_-)^\eta(r-r_+)^{-\xi+s}R(r)\chi(r)\rp) dr \nonumber\\
&\quad +\lim_{y\to 0}\int_{r_+}^\infty e^{A(z-r_-)(r-r_-)}(r-r_-)^\eta(r-r_+)^{\xi} e^{-i\omega r}R(r)\lp(1-\chi(r)\rp) dr \label{eq:g-split-sub}
\end{align}

For the second integral, we appeal to the argument of Lemma~\ref{lemma:g-tilde-IBP-sub} to write, for $\Im\omega>0$ or $y\Re\omega>0$,
\begin{align*}
&\int_{r_+}^{\infty} e^{A(z-r_-)(r-r_-)}(r-r_-)^{\eta}(r-r_+)^{\xi}e^{-i\omega r} R(r)(1-\chi(r)) dr \\
&\qquad =\frac{1}{\lp[A(z-r_-)\rp]^{k}}\int_{r_+}^{\infty} e^{A(z-r_-)(r-r_-)}\frac{d^k}{dr^k}\lp((r-r_-)^{\eta}(r-r_+)^{\xi}e^{-i\omega r} R(r)(1-\chi(r))\rp) dr\,,
\end{align*}
where, in the limit $y\to 0$, the right hand side is certainly a convergent integral if $k= 2-2s$; it is also of $O(x^{-2+2s})$ as $x\to \infty$.  

For the first integral, consider the following lemma:
\begin{lemma}\label{lemma:technical-core-mode-stability} Let $h$ be a smooth function on $[r_+,\infty)$ which vanishes on $[r_++2,\infty)$. Then, for $\nu\in\mathbb{R}\backslash\{0\}$ or $\Im\nu>0$,
\begin{align*}
Z(\nu)&=\int_{r_+}^\infty e^{i\nu r}(r-r_+)^{2\xi-s} h(r)dr \\
&=\Gamma(2\xi-s+1)\exp\lp(\frac{i\pi}{2}(2\xi-s+1)\rp)h(r_+)e^{i\nu r_+}\nu^{-2\xi+s-1}+O\lp(|\nu|^{-1+s}\rp)\,,
\end{align*} 
as $\nu\to\infty$.
\end{lemma}
\begin{proof} We want to compute $Z(\nu)$ by  integrating by parts, so it will be useful to obtain a formula for the anti-derivative of $e^{i\nu r}(r-r_+)^{2\xi-s}$. In order to do this for, we extend $r$ to $w\in\mb{C}\backslash (-\infty,r_+]$ and let 
\begin{align*}
l_0(r,\nu)&:=e^{i\nu w}(w-r_+)^{2\xi-s}
\end{align*}
where $(w-r_+)^{2\xi-s}=\exp\lp[(2\xi-s)\log(w-r_+)\rp]$ and we are taking the principal branch of the logarithm. Clearly, $\exp\lp[(2\xi-s)\log(w-r_+)\rp]$ is uniformly bounded for $\{w:\Re w\in[r_+,r_++2)\}$. Assume $\nu>0$ first; the exponential decay from $e^{i\nu w}$ as $w\to r+i\infty$ means we can define the anti-derivative of $l_0$ as
\begin{align*}
l_{-1}(r,\nu)&:=-\int_r^{r+i\infty}e^{i\nu w}(w-r_+)^{2\xi-s}dw\,,
\end{align*}
on any contour connecting $r$ and $r+i\infty$ which avoids the branch cut (and thus passes only through points where the integrand is holomorphic), by Cauchy's theorem. Take the contour $w(t)=r+it$ for $t\in[0,\infty)$; then we have
\begin{align*}
l_{-1}(r,\nu)&=-ie^{i\nu r}\int_0^{\infty}e^{-\nu t}(r-r_++it)^{2\xi-s}dt \text{~ integrating along~}t\mapsto r+it\\
&=-ie^{i\nu r}\nu^{-1}\int_0^{\infty}e^{-t}\lp(r-r_++\frac{it}{\nu}\rp)^{2\xi-s}dt\,. \numberthis \label{eq:l-n-1}
\end{align*}
For $\nu<0$, one repeats this argument along a path finishing at $r-i\infty$ instead, so as to obtain the exponential decay that allows for convergence of the integral.


Now, integration by parts once allows us to rewrite $Z$ as
\begin{align*}
Z(\nu) = S_0(\nu) +\mc{S}_1(\nu)\,, \quad
S_0(\nu)=- h(r_+)l_{-1}(r_+)\,, ~~
S_1(\nu)=\int_{r_+}^\infty h'(r) l_{0}(r) dr\,.
\end{align*}
By (\ref{eq:l-n-1}), at $r=r_+$,
\begin{gather*}
l_{-1}(r_+,\nu)=-i^{2\xi-s+1}e^{i\nu r_+}\nu^{-2\xi+s-1}\Gamma(2\xi-s+1)\\
\Rightarrow s_0 =\Gamma(2\xi-s+1)i^{2\xi-s+1}e^{i\nu r_+}|\nu|^{-2\xi+s-1} h(r_+)\,.
\end{gather*}
For the remainder term, we note that 
\begin{align*}
& \int_{r_+}^\infty \frac{d}{dr}h(r)\lp(\nu(r-r_+)+it\rp)^{2\xi-s}dr \\
&\qquad= \int_{r_+}^{r_++\nu^{-1}} \frac{d}{dr}h(r)\lp(\nu(r-r_+)+it\rp)^{2\xi-s} dr\\
 &\qquad\qquad +(i\nu)^{-1}\lp(1+it\rp)^{2\xi-s}h(r_++\nu^{-1}) \\
 &\qquad\qquad-(i\nu)^{-1} \int_{r_++\nu^{-1}}^\infty \lp(\frac{d}{dr}\rp)\lp[\lp(\nu(r-r_+)+it\rp)^{2\xi-s}h(r)\rp]dr\,,
\end{align*}
which is $O(|\nu|^{-1})$ by compact support of $h$. Thus, after Fubini
\begin{align*}
|S_1(\nu)|
&\leq |\nu|^{-1+s}\int_0^\infty e^{-t} \lp|\int_{r_+}^\infty h'(r)\lp(\nu(r-r_+)+it\rp)^{2\xi-s}dr\rp|dt\lesssim |\nu|^{-2+s}\,,
\end{align*}
for $\nu\in\mathbb{R}\{0\}$ and similarly for $\Im\omega>0$.
\end{proof}

Applying this lemma with $i\nu=Ax$ and $h(r)$ as the function in brackets in the first line of \eqref{eq:g-split-sub} yields the statement.
\end{proof}

\subsubsection{Proof of Proposition \ref{prop:ode-u-tilde-sub}}
\label{sec:integral-transformation-smooth-sub}

We are now ready to prove \cref{prop:ode-u-tilde-sub}:
\begin{proof}[Proof of \cref{prop:ode-u-tilde-sub}]
Recall that out integral transformation $\tilde{u}$ is given by (\ref{eq:subextremal-transformation-with-g-tilde})
\begin{align*}
\tilde{u}(x)=(x^2+a^2)^{1/2}(x-r_-)^{-s}(x-r_+)^{\xi+\eta}\tilde{g}(x)\,,
\end{align*}
with derivative
\begin{align} \label{eq:derivative-u-tilde-sub}
\tilde{u}'=\frac{(x-r_+)(x-r_-)}{x^2+a^2}\lp[\frac{x}{x^2+a^2}-\frac{s}{x-r_-}+\frac{\xi+\eta}{x-r_+}+\frac{1}{\tilde{g}}\frac{d\tilde{g}}{dx}\rp]\tilde{u}\,.
\end{align}
By \cref{lemma:Heun-z}, $\tilde{g}(x)$ is a smooth function for $x\in(r_+,\infty)$, so $\tilde{u}$ must also be smooth in $x\in(r_+,\infty)$. The ODE \eqref{eq:ode-u-tilde-sub} for $\tilde{u}$ can be computed directly from the ODE \eqref{eq:ode-g-tilde-sub} for $\tilde{g}$. This shows statements 1 and 2.

By \cref{lemma:g-tilde-IBP-sub}, $\tilde{g}(x)$ and $\p_x \tilde{g}(x)$  are bounded functions on $x\in[r_+,\infty)$, so $\tilde{u}$ and $\tilde{u}'$ must be bounded as well, which shows statement 3. In \cref{sec:asymptotics-g-tilde-sub}, we computed the asymptotics of $\tilde{g}$, from which, together with \eqref{eq:derivative-u-tilde-sub}, the boundary conditions as stated in statement 4 follow.

Finally, we must show injectivity; we will proceed as in \cite[proof of Theorems 1.5, 1.6]{Shlapentokh-Rothman2015}. Extend $R(r)$ to $r\in \mathbb{R}$ by 0.
\begin{itemize}
\item If $\Im\omega> 0$, we write the Fourier transform of  $(r-r_-)^{\eta}(r-r_+)^\xi e^{\gamma r}R(r)$ as
\begin{align*}
\hat{R}(z):=\frac{1}{2|\omega|^2}\int_{-\infty}^\infty e^{2i|\omega|^2 z(r-r_-)}(r-r_-)^{\eta}(r-r_+)^\xi e^{\gamma r}R(r)dr\,,
\end{align*}
which is a holomorphic function in the upper half plane (by the Paley--Wiener theorem, due to $R$ being an $L^2$ function of $r\in(M,\infty)\subseteq \mathbb{R}_+$). However, since
\begin{align*}
\int_{r_+}^\infty e^{\frac{2i\omega}{r_+-r_-}(z-r_-)(r-r_-)}(r-r_-)^{\eta}(r-r_+)^\xi e^{\gamma r}R(r)dr \equiv 0\,, 
\end{align*}
for all $x\in(r_+,+\infty)$, we conclude that $\hat{R}(z)$ vanishes along the line $\{z\in\mathbb{C}\,:\, z=y/\overline{\omega},\, y\in(1,+\infty)\}$. By analyticity, $\hat{R}$ and hence $R$ vanish identically.

\item If $\omega\in\mathbb{R}\backslash\{0\}$, then we write the Fourier transform of $(r-r_-)^{\eta}(r-r_+)^\xi e^{\gamma r}R(r)$ as the $L^2_{y}(\mathbb{R})$ function given by
\begin{align*}
\hat{R}(y):=\frac{1}{2\omega}\int_{-\infty}^\infty e^{2i\omega y (r-r_-)}(r-r_-)^{\eta}(r-r_+)^\xi e^{\gamma r}R(r)dr\,.
\end{align*}
Since
\begin{align*}
\int_{r_+}^\infty e^{\frac{2i\omega}{r_+-r_-}(z-r_-)(r-r_-)}(r-r_-)^{\eta}(r-r_+)^\xi e^{\gamma r}R(r)dr \equiv 0\,,
\end{align*}
for all $x\in(r_+,+\infty)$, we conclude that $\hat{R}(y)$ vanishes along the line $y\in(1,\infty)$. However, the Fourier transform of a non-trivial $L^2$ function supported in $\mathbb{R}_+$ cannot vanish. Since $R$ is $L^2$ in its support contained in $(r_+,+\infty)\subseteq\mathbb{R}_+$, we conclude that $R$ vanishes identically.
\end{itemize}
Hence \eqref{eq:subextremal-transformation} defines an injective map $R \mapsto \tilde{u}$, as in statement 5.
\end{proof}

\section{Proof of mode stability on Kerr backgrounds}
\label{sec:proof}

In this section, we will prove mode stability for $\omega$ in the upper half-plane and on the real axis for Kerr black hole spacetimes, including extremality. Let us begin by giving the precise statements we will prove. For convenience in the later sections, we separate the case of $\omega$ real and $\omega$ in the upper half-plane.

\begin{theorem}\label{thm:mode-stability-real-axis}  Fix $M>0$, $|a|\leq M$ and $s\in\frac12\mathbb{Z}$. There are no nontrivial outgoing solutions to the homogeneous radial ODE~\eqref{eq:radial-ODE}, in the sense of Definition~\ref{def:outgoing-radial-solution}, for any admissible frequency triple $(\omega,m,\lambda)$ with respect to $s$ and $a$ where $\omega$ is real.
\end{theorem}
\begin{theorem}\label{thm:mode-stability-upper-half}  Fix $M>0$, $|a|\leq M$ and $s\in\frac12\mathbb{Z}$. There are no nontrivial outgoing solutions to the homogeneous radial ODE~\eqref{eq:radial-ODE}, in the sense of Definition~\ref{def:outgoing-radial-solution}, for any admissible frequency triple $(\omega,m,\lambda)$ with respect to $s$ and $a$ where $\Im\omega>0$.
\end{theorem}

We note that Theorems \ref{thm:mode-stability-real-axis}  and \ref{thm:mode-stability-upper-half} do not require $\lambda$ to be a separation constant arising from separating the angular and radial components of a mode solution. However, whereas for mode stability on the real axis, $\lambda$ can take any real value, for mode stability in the upper half-plane we require the more restrictive condition of $\Im(\lambda\overline{\omega})<0$. We recall that both these conditions are fulfilled by an eigenvalue of the angular ODE \eqref{eq:angular-ode}, by statement 1 of Proposition~\ref{def:angular-ode}. Hence,  Theorems \ref{thm:mode-stability-real-axis} and \ref{thm:mode-stability-upper-half} imply Theorems~\ref{thm:mode-stability-subextremal} and \ref{thm:mode-stability-full}, discussed in the introduction. 

The structure of this section is as follows. In \cref{sec:unique-continuation-lemma}, we introduce a unique continuation lemma in the style of \cite[Lemma 6.1]{Shlapentokh-Rothman2015}. This plays a key role in \cref{sec:proof-s-negative}, where we show mode stability for $s\leq 0$. In \cref{sec:proof-s-positive}, we prove that mode stability for $s>0$ follows from the mode stability for $s<0$ via the Teukolsky--Starobinsky identities introduced in \cref{sec:teukolsky-starobinsky}. Finally, in \cref{sec:continuity-argument}, we give an alternative proof of Theorem  \ref{thm:mode-stability-upper-half} in the extremal Kerr ($|a|=M$) case.

\subsection{A unique continuation lemma}
\label{sec:unique-continuation-lemma}
For the proof of mode stability for $s\leq 0$ on the real axis, it will be useful to introduce the following unique continuation lemma, inspired by that of \cite[Lemma 5.1]{Shlapentokh-Rothman2015}:
\begin{lemma}[Unique continuation lemma]\label{lemma:unique-continuation} Suppose that we have a solution $u(r^*):(-\infty,\infty)\to \mathbb{C}$ to the ODE
\begin{align*}
u''+V u =0
\end{align*}
such that $u\in L^\infty(-\infty,\infty)$ with $(|u'|+|u|)(-\infty)=0$, $V$ is real and bounded with either
\begin{enumerate}
\item $V(\infty)=\omega^2\in\mathbb{R}\backslash\{0\}$ and $V-\omega^2$ decaying at an integrable rate as $r^*\to \infty$, or
\item $V(-\infty)=\omega_0^2\in\mathbb{R}\backslash\{0\}$ and $\hat{V}:=\omega_0^2-V$ decaying at an integrable rate as $r^*\to -\infty$.
\end{enumerate}
Then $u$ is identically zero.
\end{lemma}

\begin{proof}[Proof of Lemma~\ref{lemma:unique-continuation}]
For a continuous, piecewise continuously differentiable function $y$, we define the $y$-current as 
\begin{align*}
Q^y(r^*):=y|u'|^2+y V|u|^2\,,\quad (Q^y)'(r^*) = y'|u'|^2+(y V)'|u|^2 \,.
\end{align*}

Note that statement 1 is precisely \cite[Lemma 5.1]{Shlapentokh-Rothman2015}. As the proof is very similar to that of statement 2, we omit it here. For statement 2, define
\begin{align*}
y(r^*) := -\exp\lp(-C\int_{-\infty}^{r^*} \zeta dr^*\rp)\,,
\end{align*}
where $\zeta$ is a fixed positive function which is identically 1 near $r=\infty$ and is set to have the same (integrable) decay as $|\hat{V}|$ as $r^*\to-\infty$. In particular, $y'>0$ in $(-\infty,\infty)$, $y(+\infty)=0$ and $y(-\infty)=-1$. It is clear from the assumptions on $y$ and $u$ that $Q^y(\pm \infty)=0$, so, writing $V(r^*)=\omega_0^2-\hat{V}(r^*)$, the fundamental theorem of calculus implies
\begin{align*}
\int_{-\infty}^\infty \lp(y'|u'|^2+y'\omega_0^2-(y\hat{V})'|u|^2\rp)dr^* =0\,.
\end{align*}
The only term which can threaten this estimate is  $(y\hat{V})'|u|^2$; however,
\begin{align*}
\lp|\int_{-\infty}^\infty (y\hat{V})'|u|^2 dr^* \rp|&\leq \int_{-\infty}^\infty  2|y||u'||u| dr^* \leq \frac12 \int_{-\infty}^\infty \lp[y'|u'|^2+y'\omega_0^2 \lp(\frac{2\hat{V}}{C\omega_0 \zeta}\rp)^2|u|^2 \rp]dr^* \\
&\leq \frac12 \int_{-\infty}^\infty \lp(y'|u'|^2+y'\omega_0^2|u|^2 \rp)dr^* 
\end{align*}
if we make $C=C(\omega_0,\lVert \hat{V}\zeta^{-1}\rVert_\infty)$ sufficiently large. Thus, this term cam be absorbed into the left hand side of the previous identity and the conclusion follows from an application of the fundamental theorem of calculus.
\end{proof}

\begin{remark} We note that an analogue of Lemma~\ref{lemma:R-general-asymptotics}, based solely on an asymptotic analysis, could be used in lieu of Lemma~\ref{lemma:unique-continuation}. However, the latter is slightly stronger in that it does not require knowledge of the explicit  potential, but only its decay properties at one end. Thus, whereas for the original Teukolsky radial ODE~\eqref{eq:radial-ODE} we have sketched how to apply Lemma~\ref{lemma:R-general-asymptotics} to establish non-existence of solutions for non-superradiant frequencies (see Section~\ref{sec:superradiance}), in the following section we demonstrate how Lemma~\ref{lemma:unique-continuation} can be used to derive non-existence of solutions to the transformed equations \eqref{eq:ode-u-tilde} and \eqref{eq:ode-u-tilde-sub} for any admissible frequencies.
\end{remark}

\subsection{Mode stability \texorpdfstring{for $s\leq 0$}{for non-positive spin}}
\label{sec:proof-s-negative}

\begin{proof}[Proof of Theorems \ref{thm:mode-stability-real-axis} and \ref{thm:mode-stability-upper-half} for $s\leq 0$]

Fix $M>0$, $|a|\leq M$ and $s\in\frac12\mathbb{Z}_{\leq 0}$. Fix an admissible frequency triple $(\omega,m,\lambda)$.

Let $R(r):=R_{m\lambda}^{[s],\,(a\omega)}(r)$ be an outgoing solution to the radial ODE \eqref{eq:radial-ODE}  (see Definition~\ref{def:outgoing-radial-solution}) with vanising inhomogeneity. Let $\tilde{u}$ be defined by \eqref{eq:extremal-transformation} if $|a|=M$ and by \eqref{eq:subextremal-transformation} if $|a|<M$. Then $\tilde{u}$ is smooth in its domain  and satisfies an ODE of the form 
\begin{align}
\tilde{u}''+\tilde{V}\tilde{u}=0\,, \label{eq:model-u-tilde}
\end{align}
with respect to $x^*\in(-\infty,\infty)$, where $\tilde{V}$ was explicitly computed in Propositions \ref{prop:ode-u-tilde} and \ref{prop:ode-u-tilde-sub}, respectively for $|a|=M$ and $|a|<M$. Define the $T$-current
\begin{align}
\tilde{Q}^T := \Im \lp(\tilde{u}'\overline{\omega \tilde{u}}\rp)\,, \quad -\lp(\tilde{Q}^T\rp)' := \Im \omega|\tilde{u}'|^2 -\Im(\overline{\omega}\tilde{V})|\tilde{u}|^2\,, \label{eq:T-current-tilde}
\end{align}
where the last equality is obtained using \eqref{eq:model-u-tilde}.

For $\omega$ in the upper half-plane and on the real axis, we will show that $\tilde{u}$ vanishes identically. By the injectivity of the transformation $R\mapsto\tilde{u}$ (statement 5 in Propositions \ref{prop:ode-u-tilde-sub} and \ref{prop:ode-u-tilde}), $R$ must vanish identically as well. This is in contradiction with our \cref{def:mode-solution}.

\medskip
\noindent \textit{The case $\Im\omega >0$ (Theorem~\ref{thm:mode-stability-upper-half})}. From the boundary conditions listed in statement 4 of Propositions \ref{prop:ode-u-tilde} and \ref{prop:ode-u-tilde-sub}, it follows that $\tilde{Q}^T(\pm \infty)=0$. An application of the fundamental theorem of calculus yields
\begin{align*}
0&=\int_{-\infty}^\infty \lp(-\tilde{Q}^T\rp)'dx^* \geq \Im\omega \int_{-\infty}^\infty \lp(|\tilde{u}'|^2+\tilde{V}_{\omega^2}|\omega|^2|\tilde{u}|^2 \rp)dx^*\,,
\end{align*}
where the inequality is obtained using properties (ii) and (iii) of $\tilde{V}$ in Remarks \ref{rmk:prop-potential} and  \ref{rmk:prop-potential-sub} and the assumption that $\Im(\omega\overline{\lambda})\geq 0$. Here, $\tilde{V}_{\omega^2}$  denotes the coefficient of $\omega^2$ in $\tilde{V}$; as, by property (i) of $\tilde{V}$ in Remarks \ref{rmk:prop-potential-sub} and \ref{rmk:prop-potential}, $\tilde{V}_{\omega^2}>0$ for $x^*\in(-\infty,\infty)$ and $\Im\omega>0$ by assumption, it follows that $\tilde{u}$ vanishes identically. From injectivity of the map $R\mapsto \tilde{u}$, we conclude that $R$ must also vanish identically, which contradicts the assumption that it is the radial component of a (nontrivial) mode solution.

\medskip
\noindent \textit{The case $\omega\in\mathbb{R}\backslash\{0\}$ and, if $|a|=M$, $\omega-m\upomega_+\neq 0$ (Theorem~\ref{thm:mode-stability-real-axis})}. In this case, $\tilde{V}$ is  real (see property (iv) of Remarks \ref{rmk:prop-potential} and \ref{rmk:prop-potential-sub}), so the $\tilde{Q}^T$ current defined in \eqref{eq:T-current-tilde} is conserved. 

In both the subextremal and extremal cases, an application of the fundamental theorem of calculus yields $0= \tilde{Q}^T(-\infty)-\tilde{Q}^T(+\infty)$. Thus, if $|a|<M$
\begin{align*}
0&= \tilde{Q}^T(-\infty)-\tilde{Q}^T(+\infty)\\
&=\frac12\lp( \frac{r_+}{r_+-r_-}|\tilde{u}'(-\infty)|^2+\frac{r_+-r_-}{r_+}\omega^2 |\tilde{u}(-\infty)|^2\rp)+\frac12\lp( |\tilde{u}'(+\infty)|^2+\omega^2 |\tilde{u}(+\infty)|^2\rp)\,,
\end{align*}
by the boundary conditions of $\tilde{u}$ obtained in statement 4 of Proposition \ref{prop:ode-u-tilde-sub}, and, for $|a|=M$, if additionally $\omega(\omega-m\upomega_+)<0$
\begin{align*}
0&= \tilde{Q}^T(-\infty)-\tilde{Q}^T(+\infty)= \frac12\lp(3|\tilde{u}'(-\infty)|^2+\frac{1}{3}\omega^2 |\tilde{u}(-\infty)|^2\rp)\,,\end{align*}
or if $\omega(\omega-m\upomega_+)>0$
\begin{align*}
0&= 2\tilde{Q}^T(-\infty)-2\tilde{Q}^T(+\infty)\\
&= \frac{1}{4|\omega|\sqrt{2M\omega(\omega-m\upomega_+)}}|(x^{1/4}\tilde{u}')(+\infty)|^2+4|\omega|\sqrt{2M\omega(\omega-m\upomega_+)}|(x^{-1/4}\tilde{u})(+\infty)|^2\\
&\qquad+3|\tilde{u}'(-\infty)|^2+\frac{1}{3}\omega^2 |\tilde{u}(-\infty)|^2\,,
\end{align*}
 by the boundary conditions of $\tilde{u}$ obtained in statement 4 of Proposition \ref{prop:ode-u-tilde}. Since $\tilde{u}$ and $\tilde{u}'$ are bounded functions (statement 3 of Propositions \ref{prop:ode-u-tilde} and \ref{prop:ode-u-tilde-sub}), we can apply \cref{lemma:unique-continuation} to conclude that $\tilde{u}$ vanishes identically. From injectivity of the map $R\mapsto \tilde{u}$, it follows that $R$ must vanish identically, which contradicts the assumption that it is the radial component of a (nontrivial) mode solution.
\end{proof}

\subsection{Mode stability for \texorpdfstring{$s>0$}{positive spin}}
\label{sec:proof-s-positive}

\begin{proof}[Proof of Theorems \ref{thm:mode-stability-real-axis} and \ref{thm:mode-stability-upper-half} for $s>0$]
Fix $M>0$, $|a|\leq M$ and $s\in\frac12 \mathbb{Z}_{> 0}$. Fix an admissible frequency triple $(\omega,m,\lambda)$. 

Let $R_{m\lambda}^{[s],\,(a\omega)}(r)$ be a nontrival outgoing solution to the radial ODE \eqref{eq:radial-ODE} with inhomogeneity $\hat{F}_{m\lambda}^{[s],\,(a\omega)}\equiv 0$ (see Definition~\ref{def:outgoing-radial-solution}). Then, we can generate a nonzero solution to the radial ODE of opposite sign spin via the Teukolsky--Starobinsky identities (\ref{eq:TS-radial}), according to \cref{lemma:TS-radial-boundary-conditions}. This contradicts the mode stability result for $s< 0$ obtained in \cref{sec:proof-s-negative}.
\end{proof}

\subsection{An alternative proof of mode stability \texorpdfstring{for $\Im\omega>0$}{in the upper half plane} for extremal Kerr}
\label{sec:continuity-argument}

Indeed, in this section, we will show that Theorem \ref{thm:mode-stability-upper-half} for $|a|<M$ implies Theorem~\ref{thm:mode-stability-upper-half} for $|a|=M$, circumventing the need, in this case, for the novel integral transformation we have previously introduced. We note that the the continuity argument we present makes precise the argument suggested at the end of \cite[Section VIII-B]{Casals2019}.

\begin{proof}[Alternative proof of Theorem~\ref{thm:mode-stability-upper-half} for $|a|=M$]
Fix $M>0$ and $s\in\frac12\mathbb{Z}$. Let $(\omega,m,\lambda)$ be an admissible frequency triple with $\Im\omega>0$ and $R:=R_{m\lambda}^{[s],\,(a\omega)}$ be a nontrivial solution to the radial ODE~\eqref{eq:radial-ODE} with the boundary condition at the horizon in Definition~\ref{def:outgoing-radial-solution}. 

By \eqref{eq:R-general-asymptotics}, for fixed $s$ and $m$, dropping superscripts,
\begin{align*}
R(r)= \underline{Z}(a,\omega,\lambda)R_{\mc{I}^+}(r)+Z(a,\omega,\lambda)R_{\mc{I}^-}(r)
\end{align*}
for some complex $Z$ and $\underline{Z}$. From the proofs given in \cite{Hartle1974}, which can be trivially extended to yield statements in the full range $|a|\leq M$, it follows that
\begin{theorem}[\cite{Hartle1974}] \label{thm:hartle-Z} Fix $M>0$, $s\in\frac12\mathbb{Z}$ and an admissible frequency parameter $m$. Then, for $(a,\omega,\lambda)\in[-M,M]\times\mathbb{C}\times\mathbb{C}$ such that
\begin{align*}
2M\lp(M^2+\sqrt{M^2-a^2}\rp)\omega\neq am -i|s-1|\sqrt{M^2-a^2}-it\,,\text{~for any~} t\in[0,\infty)\,,
\end{align*}
$Z$ is analytic in $\lambda$ (fixing $a$ and $\omega$), analytic in $\omega$ (fixing $a$ and $\lambda$) and continuous in $a$ (fixing $\lambda$ and $\omega$).
\end{theorem}

For $R$ to be the radial component of a mode solution (in particular, for it to be an outgoing solution of the radial ODE), $a$ and $\omega$ must be such that $Z$ vanishes. Suppose there is a zero of $Z_{a_0}$ for $|a_0|=M$ at $(\omega,\lambda)=(\omega_0,\lambda_0)$ with $\Im\omega_0>0$ and $\Im(\overline{\lambda_0}\omega_0)>\sigma$. In any compact set of parameters $(a,\omega)$ for which the conclusion of Theorem~\ref{thm:hartle-Z} holds, $Z$ is a uniformly continuous function of the two variables: thus, for any $\varepsilon>0$ fixed independently of $a$ and $\omega$,
\begin{align*}
\exists \delta>0\,\,\, \text{such that}\,\, |a-a_0|<\delta \text{~and~} \omega\in B_{|\omega_0|/2}(\omega_0) \text{~implies~} |Z_{a_0}(\omega)-Z_{a}(\omega)|< \varepsilon\,.
\end{align*}
Now, since $Z_{a_0}(\omega)$ is an analytic function of $\omega$, its zeros are isolated points. Thus, we can certainly find some $\sigma' < |\omega_0|/2$ such that $|Z_{a_0}|>\varepsilon>|Z_{a_0}-Z_{a}|$ for $\omega\in\p B_{\sigma'}(\omega_0)$. We further tighten the ball by ensuring that $(\omega,m,\lambda_0)$ is still an admissible frequency triple for any $\omega\in\p B_{\sigma'}(\omega_0)$:
\begin{align*}
\Im\lp(\overline{\lambda_0}\omega\rp)>\sigma-|\lambda_0|\sigma'>0\,.
\end{align*}
 By Rouché's theorem (see e.g.\ \cite{Fischer2012}), $Z_{a}$ will also have a zero in this ball. We conclude that, if there is a mode in the upper half plane for $|a_0|=M$, then there must exist a mode in the upper half plane for $|a|=M-\delta/2$, which contradicts the statement of Theorem \ref{thm:mode-stability-upper-half} in the case $|a|<M$ (proven in \cite{Whiting1989}).  
\end{proof}

\begin{remark} Note that in the previous proof, we have made use of the fact that the mode stability results we can prove hold for more general frequency parameters $\lambda$ than the angular eigenvalue arising from the separation constant. Hence, we have not shown that, for $\Im\omega\geq 0$, Theorem~\ref{thm:mode-stability-subextremal} implies Theorem~\ref{thm:mode-stability-full}. Such an implication would appear to be significantly more difficult to obtain, given that, for complex $\omega$, the angular eigenvalues, and hence $Z$, become multi-valued complex functions of $a$ and $\omega$.
\end{remark}

\section{Quantitative mode stability on the real axis and a scattering theory for the Teukolsky equation}
\label{sec:quantitative}

As in \cite{Shlapentokh-Rothman2015}, mode stability for the Teukolsky equation can be made into a quantitative statement, which is the goal of this section.  In Section~\ref{sec:quantitative-statement}, we state the main result, Theorem~\ref{thm:quantitative}, which we prove in Section~\ref{sec:quantitative-proof}, with the aid of some estimates, given in Section~\ref{sec:estimates-u-tilde}, on the ODEs for the previously introduced integral tranformations. We also provide an application in the context of scattering theory in Section~\ref{sec:scattering}. 

\subsection{Statement of the main theorem}
\label{sec:quantitative-statement}

Recall Definition~\ref{def:uhor-uout} and \eqref{eq:uout-uhor} for $\swei{u}{s}_{\mc{H}^+}$ and $\swei{u}{s}_{\mc{I}^+}$, which are uniquely defined solutions of the homogeneous radial ODE~\eqref{eq:u-Schrodinger}. 

\begin{definition} \label{def:Wronskian} Fix $M>0$, $|a|\leq M$ and $s\in\frac12\mathbb{Z}$. For an admissible frequency triple $(\omega,m,\lambda)$ with real $\omega$, we define $\swei{\mathfrak{W}}{s}(\omega,m,\lambda)$ to be the Wronskian of $\swei{u}{s}_{\mc{H}^+}$ and $\swei{u}{s}_{\mc{I}^+}$ 
\begin{equation} 
\swei{\mathfrak{W}}{s}(\omega,m,\lambda):=\lp(\swei{u}{s}_{\mc{I}^+}\rp)'\cdot\lp(\swei{u}{s}_{\mc{H}^+}\rp)- \lp(\swei{u}{s}_{\mc{H}^+}\rp)'\cdot \lp(\swei{u}{s}_{\mc{I}^+}\rp)\,.
\end{equation}
\end{definition}

\begin{remark} Note that $\swei{\mathfrak{W}}{s}$ in Definition~\ref{def:Wronskian} is independent of $r^*$ due to the fact that $\swei{u}{s}_{\mc{I}^+}$ and $\swei{u}{s}_{\mc{H}^+}$ satisfy the same second order ODE which does not involve first order derivatives in $r^*$. 
\end{remark}

There is  a non-trivial solution to the homogeneous radial ODE (\ref{eq:u-Schrodinger}) (i.e.\ $u_{\mc{H}^+}$ and $u_{\mc{I}^+}$ are linearly dependent) if and only if $\mathfrak{W}=0\Leftrightarrow \lp|\mathfrak{W}^{-1}\rp|=\infty$. The results shown in the previous section immediately imply
\begin{corollary}[of Theorem \ref{thm:mode-stability-real-axis}] \label{cor:nonexplicit-quantitative-mode-stability} Fix $M>0$, $|a|\leq M$ and $s\in\frac12\mathbb{Z}$. For any admissible frequency triple $(\omega,m,\lambda)$ with respect to $s$ and $a$,
$$\swei{\mathfrak{W}}{s}(\omega,m,\lambda)\neq 0\,.$$
Consequently, on any compact set of frequencies where Theorem \ref{thm:mode-stability-real-axis} holds, there is a real number $G>0$ such that
$$\lp|\swei{\mathfrak{W}}{s}\rp|^{-1}\leq G< \infty\,.$$
\end{corollary}

Quantitative mode stability consists of showing that, in a bounded frequency regime, the constant $G$ can be explicitly determined in terms of the compact set of frequencies, the spin $s$ characterizing the Teukolsky equation, and the black hole parameters. For subextremal Kerr black hole spacetimes ($|a|< M$), such a bound was attained in \cite{Shlapentokh-Rothman2015}, with a constant that degenerated as $|a|\to M$. A careful manipulation of the integral transformations introduced in Section \ref{sec:integral-transformation} will enable us to upgrade Theorem~\ref{thm:mode-stability-real-axis} to a quantitative statement which is uniform in the specific angular momentum of a Kerr black hole:

\begin{theorem} Fix $M>0$ and $s\in\frac12\mathbb{Z}$. Let $\mc{A}$  be a set of frequency parameters $(\omega,m,\lambda)$ admissible with respect to $s$ such that $\omega$ is real and
\begin{align*}
C_\mc{A}&:=\sup_{(a,\omega,m,l)\in[-M,M]\times\mc{A}}\lp(|\omega|+|\omega|^{-1}+\lp|\omega-\frac{a m}{2M^2}\rp|^{-1}\delta_{|a|,M}+|m|+\lp|\lambda\rp|\rp)<\infty\,,
\end{align*}
where $\delta_{|a|,M}$ is the Kronecker delta. Then
\begin{align}
\sup_{(\omega,\,m,\,\lambda)\in\mc{A}}\lp|\swei{\mathfrak{W}}{s}\rp|^{-1}\leq G(C_\mc{A},M,|s|)<\infty\,, \label{eq:Wronskian-bounded}
\end{align}
where $G$ will be given in an explicitly computable manner by \eqref{eq:Wronskian-explicit-bound}, respectively.
\label{thm:quantitative}
\end{theorem}

Together with Theorem \ref{thm:hartle-Z}, as $\swei{\mathfrak{W}}{s}=2i\omega Z$, we can infer that, away from any of the problematic frequencies already identified, the Wronskian is continuous in the extremal limit:
\begin{corollary}[of Theorem \ref{thm:hartle-Z} and \ref{thm:quantitative}]\label{cor:hartle} Fix $M>0$, $|a|\leq M$ and $s\in\frac12\mathbb{Z}$. Then, for each real $m$ and $\lambda$ admissible with respect to $s$, the Wronskian $\swei{\mathfrak{W}}{s}(\omega,m,\lambda)$ is continuous in $\omega$ for every $\omega\in\mathbb{R}$ admissible with respect to $a$.
\end{corollary}

\subsection{Application to scattering theory}
\label{sec:scattering}

A simple corollary from Theorem \ref{thm:quantitative} concerns the transmission and reflection coefficients for the radial ODE. We begin with a lemma to define these objects:
\begin{lemma}[Transmission and reflection coefficients] \label{def:t-r-coeffs} For admissible frequency triples $(\omega,m,\lambda)$ with $\omega\in\mathbb{R}\backslash\{0, m\upomega_+\}$ for which $\swei{\mathfrak{W}}{s}(\omega,m,\lambda)\neq 0$, there exists a unique set of complex numbers $\mathfrak{T}^{[s]}(\omega,m,\lambda)$, $\mathfrak{\tilde{T}}^{[s]}(\omega,m,\lambda)$, $\mathfrak{R}^{[s]}(\omega,m,\lambda)$ and $\mathfrak{\tilde{R}}^{[s]}(\omega,m,\lambda)$ such that, if $s\geq 0$
\begin{equation} \label{eq:def-t-r-1}
\begin{split}
\frac{\swei{\mathfrak{T}}{s}}{-i(\omega-m\upomega_+)}\swei{u}{s}_{\mc{H}^+}&=\frac{\swei{\mathfrak{R}}{s}}{i\omega}\swei{u}{s}_{\mc{I}^+} +\frac{1}{i\omega}\swei{u}{s}_{\mc{I}^-}\,, \\
\frac{\swei{\mathfrak{\tilde{T}}}{s}}{i\omega}\swei{u}{s}_{\mc{I}^+}&=\frac{\swei{\mathfrak{\tilde{R}}}{s}}{-i(\omega-m\upomega_+)}\swei{u}{s}_{\mc{H}^+} +\frac{1}{-i(\omega-m\upomega_+)}\swei{u}{s}_{\mc{H}^-}\,,
\end{split}
\end{equation}
and, if $s<0$,
\begin{equation} \label{eq:def-t-r-1-+}
\begin{split}
\frac{\swei{\tilde{\mathfrak{T}}}{s}}{-i(\omega-m\upomega_+)}\swei{u}{s}_{\mc{H}^+}&=\frac{\swei{\tilde{\mathfrak{R}}}{s}}{i\omega}\swei{u}{s}_{\mc{I}^+} +\frac{1}{i\omega}\swei{u}{s}_{\mc{I}^-}\,, \\
\frac{\swei{\mathfrak{{T}}}{s}}{i\omega}\swei{u}{s}_{\mc{I}^+}&=\frac{\swei{\mathfrak{{R}}}{s}}{-i(\omega-m\upomega_+)}\swei{u}{s}_{\mc{H}^+} +\frac{1}{-i(\omega-m\upomega_+)}\swei{u}{s}_{\mc{H}^-}\,.
\end{split}
\end{equation}
We say that $\mathfrak{{R}}^{[s]}$ and $\mathfrak{\tilde{R}}^{[s]}$ are {\normalfont reflection coefficients} and that $\mathfrak{{T}}^{[s]}$ and $\mathfrak{\tilde{T}}^{[s]}$ are {\normalfont transmission coefficients}.
\end{lemma}

\begin{proof}
If $\swei{\mathfrak{W}}{s}\neq 0$, we can solve these equations for the transmission and reflection coefficients using the definition of $\swei{\mathfrak{W}}{s}$. We obtain
\begin{equation}\label{eq:def-t-r-2}
\begin{split}
\swei{\mathfrak{T}}{s}(\omega,m,\lambda)&=-\frac{\omega-m\upomega_+}{\omega}\frac{1}{\swei{\mathfrak{W}}{s}}\,W\lp[\swei{u}{s}_{\mc{I}^-},\swei{u}{s}_{\mc{I}^+}\rp]=-\frac{2i(\omega-m\upomega_+)}{\swei{\mathfrak{W}}{s}}\,,\\
\swei{\tilde{\mathfrak{T}}}{s}(\omega,m,\lambda)&=\frac{\omega}{\omega-m\upomega_+}\frac{1}{\swei{\mathfrak{W}}{s}}\,W\lp[\swei{u}{s}_{\mc{H}^-},\swei{u}{s}_{\mc{H}^+}\rp]=-\frac{2i\omega}{\swei{\mathfrak{W}}{s}}\,,\\
\swei{\mathfrak{R}}{s}(\omega,m,\lambda)&=\frac{1}{\swei{\mathfrak{W}}{s}}\,W\lp[\swei{u}{+s}_{\mc{I}^-},\swei{u}{s}_{\mc{H}^+}\rp]=\frac{1}{\swei{\mathfrak{W}}{s}}\,W\lp[\overline{\swei{u}{-s}_{\mc{I}^+}},\swei{u}{s}_{\mc{H}^+}\rp]\,,\\
\swei{\tilde{\mathfrak{R}}}{s}(\omega,m,\lambda)&=-\frac{1}{\swei{\mathfrak{W}}{s}}\,W\lp[\swei{u}{s}_{\mc{H}^-},\swei{u}{s}_{\mc{I}^+}\rp]=-\frac{1}{\swei{\mathfrak{W}}{s}}\,W\lp[\overline{\swei{u}{-s}_{\mc{H}^+}},\swei{u}{s}_{\mc{I}^+}\rp]\,,
\end{split}
\end{equation}
where $W[y_1,y_2]:=y_1y_2'-y_2y_1'$ denotes a Wronskian, for $s\geq 0$. For $s<0$, the same hold interchanging the tilded with non-tilded coefficients. We have further simplified these expressions by computing the Wronskians $W$ in the definitions of the transmission coefficients directly from the asymptotic expansions of $\swei{u}{s}_{\mc{I}^+}$ and $\swei{u}{s}_{\mc{H}^+}$ and Definition~\ref{def:uhor-uout}:
\begin{align*}
W\lp[u^{[s]}_{\mc{I}^-},\swei{u}{s}_{\mc{I}^+}\rp]&= 2i\omega\,, \\
W\lp[u^{[s]}_{\mc{H}^-},\swei{u}{s}_{\mc{H}^+}\rp]&= -2i(\omega-m\upomega_+)\,.
\end{align*}

By the considerations in Section \ref{sec:superradiance-s}, $W\lp[\overline{\swei{u}{-s}_{\mc{I}^+}},\swei{u}{s}_{\mc{H}^+}\rp]$ and $W\lp[\overline{\swei{u}{-s}_{\mc{H}^+}},\swei{u}{s}_{\mc{I}^+}\rp]$ are conserved in $r^*$. Hence, the transmission and reflection coefficients are independent of $r^*$.
\end{proof}

We recall that, by Theorem~\ref{thm:mode-stability-real-axis}, there is a lower bound on $\mathfrak{W}$ in a compact range of frequencies where the theorem applies, hence from \eqref{eq:def-t-r-2} one can infer immediately that there is an upper bound for the transmission and reflection coefficients in such a frequency range. The goal of this section is to make the latter bound \textit{explicit} in the frequency range and the black hole parameters:

\begin{corollary}[of Theorem~\ref{thm:quantitative}]\label{cor:bddness-scattering} Fix $M>0$ and $s\in \lp\{0,\pm \frac12, \pm 1, \pm \frac32, \pm 2\rp\}$. Let $\mc{B}\subset\mathcal{A}$ and $\tilde{\mc{B}}\subset\mathcal{A}$ be sets of frequency parameters $(\omega,m,\lambda)$ which are admissible with respect to $s$ such that $\omega$ is real, $\mathfrak{C}_{s}(\omega,m,\lambda)>0$ and
\begin{align*}
C_\mc{B}&:=\sup_{(a,\omega,m,l)\in[-M,M]\times\mathcal{B}}\lp(|\omega|+|\omega|^{-1}+|\omega-m\upomega_+|^{-1}+|m|+|\lambda|+|\mathfrak{C}_{s}(\omega,m,\lambda)|^{-1}\rp)<\infty\,,\\
C_{\tilde{\mc{B}}}&:=\sup_{(a,\omega,m,l)\in[-M,M]\times\tilde{\mathcal{B}}}\lp(|\omega|+|\omega|^{-1}+|\omega-m\upomega_+|^{-1}+|m|+|\lambda|\rp)<\infty\,,
\end{align*}
where $\mathfrak{C}_s(\omega,m,\lambda)$ is the radial Teukolsky--Starobinsky constant defined by Proposition~\ref{prop:eigenfunctions}. Then, 
\begin{align*}
\lp|\swei{\mathfrak{{T}}}{s}\rp|^2+\lp|\swei{\mathfrak{R}}{s}\rp|^2\leq G(C_\mc{B},M,|s|)<\infty\,,\\
\lp|\swei{\tilde{\mathfrak{T}}}{s}\rp|^2+\lp|\swei{\tilde{\mathfrak{R}}}{s}\rp|^2\leq G(C_{\tilde{\mc{B}}},M,|s|)<\infty\,,
\end{align*}
where $G$ can be given explicitly. For $|a|=0$ or $s$ half-integer, $G\equiv 1$.
\end{corollary}
\begin{remark} In general, $\mc{B}\subsetneq\tilde{\mc{B}}$. However, when $|s|=2$, if we take $\lambda$ to be a separation constant, $\lambda=\bm\uplambda_{ml}^{[s],\,(a\omega)}$ for some admissible $l$ (see Proposition~\ref{def:angular-ode}), by \eqref{eq:TS-constant-sign} in Remark~\ref{rmk:TS-constant-sign}, we must have $|\mathfrak{C}_{2}|>0$ if $\omega\neq 0$, so $\mc{B}=\tilde{\mc{B}}$. We recall that, in the case $|s|=2$, the Teukolsky equation describes the dynamics of the extremal curvature components under a linearization of the Einstein equation around the Kerr black hole solution, making this spin especially meaningful in the study of the Kerr black hole spacetimes.
\end{remark}

\begin{figure}[htbp]
    \centering
    \begin{subfigure}[t]{0.49\textwidth}
        \centering
 \includegraphics[scale=1]{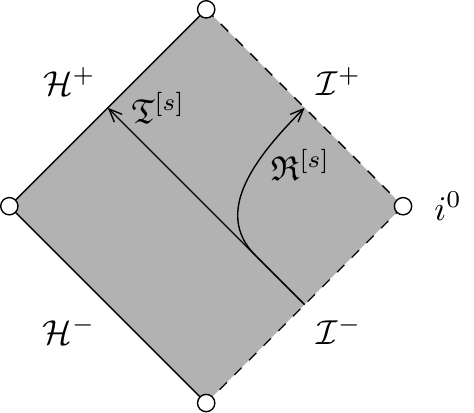}
        \caption{}
    \end{subfigure}%
    ~ 
        \begin{subfigure}[t]{0.49\textwidth}
        \centering
 \includegraphics[scale=1]{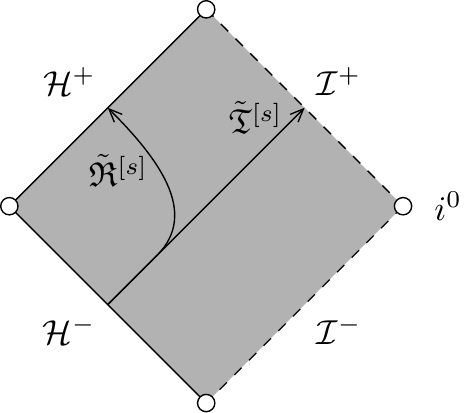}
        \caption{}
    \end{subfigure}%
    \caption{These Penrose diagrams for the manifold with boundary $\mc{R}\cup\mc{H}^-$ (see Section~\ref{sec:kerr-spacetime}) are meant to illustrate the meaning of the reflection and transmission coefficients for $s\geq 0$. In (a) we consider a forward scattering problem where data of energy 1 is placed at $\mc{I}^{-}$ and the reflection and transmission coefficients measure the fraction of energy that is scattered to $\mc{I}^+$ and $\mc{H}^+$, respectively. In (b), we consider a different scattering problem where data of energy 1 is placed at $\mc{H}^{-}$ and the reflection and transmission coefficients measure the fraction of energy that is scattered to $\mc{H}^+$ and $\mc{I}^+$, respectively. In the case $s<0$, the arrows in (a) represent the tilded quantities and those in (b) represent the non-tilded quantities.\\    The difference in our notation can be thought of reflecting the complementary character that one expects from suitable data posed for the opposing signs of the Teukolsky equation: consistency with the the Teukolsky--Starobinsky identities means that, at each past boundary only data for one sign of $s$ needs to be specified. We refer the reader to the upcoming \cite{Masaood} for further insight on this point, specifically in the case $|s|=2$ and $a=0$.}
    \label{fig:scattering-r-t}
\end{figure}

\begin{proof}[Proof of Corollary~\ref{cor:bddness-scattering}]
We begin by using the energy identities in Propositions \ref{prop:T-identity} and  \ref{prop:wronskian-identity} to relate the transmission and reflection coefficients. For the non-tilded map (Figure~\ref{fig:scattering-r-t}a), set
\begin{itemize}
\item if $|a|=M$ and $s$ is an integer or if $s=0$,
\begin{gather*}
\swei{a}{ s}_{\mc{H}^{-}}=0\,,\quad \lp|\swei{a}{ s}_{\mc{I}^{-}}\rp|=1\,, \quad \lp|\swei{a}{ s}_{\mc{I}^{+}}\rp|^2=\frac{\mathfrak C_s}{(2\omega)^{4s}}\lp|\swei{\mathfrak{R}}{ s}\rp|^2\,, \\
\quad\lp|\swei{a}{s}_{\mc{H}^{+}}\rp|^2=\lp[\frac{\omega}{2M^2(\omega-m\upomega_+)}\rp]^{2s}\lp|\swei{\mathfrak{T}}{s}\rp|^2\,,\\
\lp|\swei{a}{- s}_{\mc{H}^{-}}\rp|=1\,, \quad\swei{a}{- s}_{\mc{I}^{-}}=0\,, \quad\lp|\swei{a}{-s}_{\mc{I}^{+}}\rp|^2=\lp(\frac{\omega-m\upomega_+}{\omega}\rp)^2\lp[\frac{2M^2(\omega-m\upomega_+)}{\omega}\rp]^{2s}\lp|\swei{{\mathfrak{T}}}{-s}\rp|^2\,, \\
\lp|\swei{a}{- s}_{\mc{H}^{+}}\rp|^2=\frac{\mathfrak C_s}{[4M^2(\omega-m\upomega_+)]^{4s}}\lp|\swei{\mathfrak{R}}{- s}\rp|^2\,;
\end{gather*}
\item if $|a|=M$ and $s$ is a half-integer,
\begin{gather*}
\swei{a}{ s}_{\mc{H}^{-}}=0\,,\quad \lp|\swei{a}{ s}_{\mc{I}^{-}}\rp|=1\,, \quad \lp|\swei{a}{ s}_{\mc{I}^{+}}\rp|^2=\frac{\mathfrak C_s}{(2\omega)^{4s}}\lp|\swei{\mathfrak{R}}{ s}\rp|^2\,, \\
\quad\lp|\swei{a}{s}_{\mc{H}^{+}}\rp|^2=2M^2\lp[\frac{\omega}{2M^2(\omega-m\upomega_+)}\rp]^{2s+1}\lp|\swei{\mathfrak{T}}{s}\rp|^2\,,\\
\lp|\swei{a}{- s}_{\mc{H}^{-}}\rp|=1\,, \quad\swei{a}{- s}_{\mc{I}^{-}}=0\,, \quad\lp|\swei{a}{-s}_{\mc{I}^{+}}\rp|^2=2M^2\lp[\frac{\omega}{2M^2(\omega-m\upomega_+)}\rp]^{2s+1}\lp|\swei{{\mathfrak{T}}}{-s}\rp|^2\,,\quad \\
\lp|\swei{a}{- s}_{\mc{H}^{+}}\rp|^2=\frac{\mathfrak C_s}{[4M^2(\omega-m\upomega_+)]^{4s}}\lp|\swei{\mathfrak{R}}{- s}\rp|^2\,;
\end{gather*}
\item if $|a|<M$ and $s$ is an integer, setting products denoted by $\prod$ to be the identity whenever $s=1$,
\begin{gather*}
\swei{a}{ s}_{\mc{H}^{-}}=0\,,\quad \lp|\swei{a}{ s}_{\mc{I}^{-}}\rp|=1\,, \quad \lp|\swei{a}{ s}_{\mc{I}^{+}}\rp|^2=\frac{\mathfrak C_s}{(2\omega)^{4s}}\lp|\swei{\mathfrak{R}}{ s}\rp|^2\,, \\
\quad\lp|\swei{a}{s}_{\mc{H}^{+}}\rp|^2=\frac{4|\xi|^2\prod_{j=1}^{s-1}\lp[4|\xi|^2+(s-j)^2\rp]}{(2\omega)^{2s}}\lp|\swei{\mathfrak{T}}{s}\rp|^2\,,\\
\lp|\swei{a}{- s}_{\mc{H}^{-}}\rp|=1\,, \quad\swei{a}{- s}_{\mc{I}^{-}}=0\,, \quad\lp|\swei{a}{-s}_{\mc{I}^{+}}\rp|^2=\frac{(2\omega)^{2s-2}}{\prod_{j=1}^{s-1}\lp[4|\xi|^2+(s-j)^2\rp]}\lp(\frac{r_+-r_-}{2Mr_+}\rp)^2\lp|\swei{{\mathfrak{T}}}{-s}\rp|^2\,,\quad \\
\lp|\swei{a}{- s}_{\mc{H}^{+}}\rp|^2=\frac{\mathfrak C_s}{4|\xi|^2(4|\xi|^2+s^2)\prod_{j=1}^{s-1}\lp[4|\xi|^2+(s-j)^2\rp]^2}\lp|\swei{\mathfrak{R}}{- s}\rp|^2\,;
\end{gather*}
\item if $|a|<M$ and $s$ is a half-integer, setting products denoted by $\prod$ to be the identity whenever $s=1/2$,
\begin{gather*}
\swei{a}{ s}_{\mc{H}^{-}}=0\,,\quad \lp|\swei{a}{ s}_{\mc{I}^{-}}\rp|=1\,, \quad \lp|\swei{a}{ s}_{\mc{I}^{+}}\rp|^2=\frac{\mathfrak C_s}{(2\omega)^{4s}}\lp|\swei{\mathfrak{R}}{ s}\rp|^2\,, \\
\quad\lp|\swei{a}{s}_{\mc{H}^{+}}\rp|^2=\frac{\prod_{j=1}^{\lfloor s\rfloor}\lp[4|\xi|^2+(s-j)^2\rp]}{(2\omega)^{2s-1}}\frac{2Mr_+}{r_+-r_-}\lp|\swei{\mathfrak{T}}{s}\rp|^2\,,\\
\lp|\swei{a}{- s}_{\mc{H}^{-}}\rp|=1\,, \quad\swei{a}{- s}_{\mc{I}^{-}}=0\,, \quad\lp|\swei{a}{-s}_{\mc{I}^{+}}\rp|^2=\frac{(2\omega)^{2s-1}}{\prod_{j=1}^{\lfloor s\rfloor}\lp[4|\xi|^2+(s-j)^2\rp]}\frac{r_+-r_-}{2Mr_+}\lp|\swei{{\mathfrak{T}}}{-s}\rp|^2\,,\quad \\
\lp|\swei{a}{- s}_{\mc{H}^{+}}\rp|^2=\frac{\mathfrak C_s}{(4|\xi|^2+s^2)\prod_{j=1}^{\lfloor s \rfloor}\lp[4|\xi|^2+(s-j)^2\rp]^2}\lp|\swei{\mathfrak{R}}{- s}\rp|^2\,.
\end{gather*}
\end{itemize}
For the tilded map (Figure~\ref{fig:scattering-r-t}b), set
\begin{itemize}
\item if $|a|=M$ and $s$ is an integer or if $s=0$,
\begin{gather*}
\swei{a}{ s}_{\mc{H}^{-}}=0\,,\quad \lp|\swei{a}{ s}_{\mc{I}^{-}}\rp|=1\,, \quad \lp|\swei{a}{ s}_{\mc{I}^{+}}\rp|^2=\frac{\mathfrak C_s}{[4M^2(\omega-m\upomega_+)]^{4s}}\lp|\swei{\tilde{\mathfrak{R}}}{ s}\rp|^2\,, \\
\quad\lp|\swei{a}{s}_{\mc{H}^{+}}\rp|^2=\lp[\frac{2M^2(\omega-m\upomega_+)}{\omega}\rp]^{2s}\lp|\swei{\tilde{\mathfrak{T}}}{s}\rp|^2\,,\\
\lp|\swei{a}{- s}_{\mc{H}^{-}}\rp|=1\,, \quad\swei{a}{- s}_{\mc{I}^{-}}=0\,, \quad\lp|\swei{a}{-s}_{\mc{I}^{+}}\rp|^2=\lp[\frac{\omega}{2M^2(\omega-m\upomega_+)}\rp]^{2s}\lp|\swei{\tilde{\mathfrak{T}}}{-s}\rp|^2\,,\quad \\
\lp|\swei{a}{- s}_{\mc{H}^{+}}\rp|^2=\frac{\mathfrak C_s}{(2\omega)^{4s}}\lp|\swei{\tilde{\mathfrak{R}}}{- s}\rp|^2\,;
\end{gather*}
\item if $|a|=M$ and $s$ is a half-integer,
\begin{gather*}
\swei{a}{ s}_{\mc{H}^{-}}=0\,,\quad \lp|\swei{a}{ s}_{\mc{I}^{-}}\rp|=1\,, \quad \lp|\swei{a}{ s}_{\mc{I}^{+}}\rp|^2=\frac{\mathfrak C_s}{[4M^2(\omega-m\upomega_+)]^{4s}}\lp|\swei{\tilde{\mathfrak{R}}}{ s}\rp|^2\,, \\
\quad\lp|\swei{a}{s}_{\mc{H}^{+}}\rp|^2=\frac{1}{2M^2}\lp[\frac{2M^2(\omega-m\upomega_+)}{\omega}\rp]^{2s+1}\lp|\swei{\tilde{\mathfrak{T}}}{s}\rp|^2\,,\\
\lp|\swei{a}{- s}_{\mc{H}^{-}}\rp|=1\,, \quad\swei{a}{- s}_{\mc{I}^{-}}=0\,, \quad\lp|\swei{a}{-s}_{\mc{I}^{+}}\rp|^2=\frac{1}{2M^2}\lp[\frac{\omega}{2M^2(\omega-m\upomega_+)}\rp]^{2s+1}\lp|\swei{\tilde{\mathfrak{T}}}{-s}\rp|^2\,,\quad \\
\lp|\swei{a}{- s}_{\mc{H}^{+}}\rp|^2=\frac{\mathfrak C_s}{(2\omega)^{4s}}\lp|\swei{\tilde{\mathfrak{R}}}{- s}\rp|^2\,;
\end{gather*}
\item if $|a|<M$ and $s$ is an integer, setting products denoted by $\prod$ to be the identity whenever $s=1$,
\begin{gather*}
\swei{a}{ s}_{\mc{H}^{-}}=0\,,\quad \lp|\swei{a}{ s}_{\mc{I}^{-}}\rp|=1\,, \quad \lp|\swei{a}{ s}_{\mc{I}^{+}}\rp|^2=\frac{\mathfrak C_s}{4|\xi|^2(4|\xi|^2+s^2)\prod_{j=1}^{s-1}\lp[4|\xi|^2+(s-j)^2\rp]^2}\lp|\swei{\tilde{\mathfrak{R}}}{ s}\rp|^2\,, \\
\quad\lp|\swei{a}{s}_{\mc{H}^{+}}\rp|^2=\frac{(4|\xi|^2+s^2)\prod_{j=1}^{s-1}\lp[4|\xi|^2+(s-j)^2\rp]}{(2\omega)^{2s}}\lp(\frac{r_+-r_-}{2Mr_+}\rp)^2\lp|\swei{\tilde{\mathfrak{T}}}{s}\rp|^2\,,\\
\lp|\swei{a}{- s}_{\mc{H}^{-}}\rp|=1\,, \quad\swei{a}{- s}_{\mc{I}^{-}}=0\,, \quad\lp|\swei{a}{-s}_{\mc{I}^{+}}\rp|^2=\frac{4|\xi|^2(2\omega)^{2s}}{(4|\xi|^2+s^2)\prod_{j=1}^{s-1}\lp[4|\xi|^2+(s-j)^2\rp]}\lp|\swei{\tilde{\mathfrak{T}}}{-s}\rp|^2\,,\quad \\
\lp|\swei{a}{- s}_{\mc{H}^{+}}\rp|^2=\frac{\mathfrak C_s}{(2\omega)^{4s}}\lp|\swei{\tilde{\mathfrak{R}}}{- s}\rp|^2\,;
\end{gather*}
\item if $|a|<M$ and $s$ is a half-integer, setting products denoted by $\prod$ to be the identity whenever $s=1/2$,
\begin{gather*}
\swei{a}{ s}_{\mc{H}^{-}}=0\,,\quad \lp|\swei{a}{ s}_{\mc{I}^{-}}\rp|=1\,, \quad \lp|\swei{a}{ s}_{\mc{I}^{+}}\rp|^2=\frac{\mathfrak C_s}{(4|\xi|^2+s^2)\prod_{j=1}^{\lfloor s \rfloor}\lp[4|\xi|^2+(s-j)^2\rp]^2}\lp|\swei{\tilde{\mathfrak{R}}}{ s}\rp|^2\,, \\
\quad\lp|\swei{a}{s}_{\mc{H}^{+}}\rp|^2=\frac{(4|\xi|^2+s^2)\prod_{j=1}^{\lfloor s\rfloor}\lp[4|\xi|^2+(s-j)^2\rp]}{(2\omega)^{2s+1}}\frac{r_+-r_-}{2Mr_+}\lp|\swei{\tilde{\mathfrak{T}}}{s}\rp|^2\,,\\
\lp|\swei{a}{- s}_{\mc{H}^{-}}\rp|=1\,, \quad\swei{a}{- s}_{\mc{I}^{-}}=0\,, \quad\lp|\swei{a}{- s}_{\mc{H}^{+}}\rp|^2=\frac{\mathfrak C_s}{(2\omega)^{4s}}\lp|\swei{\tilde{\mathfrak{R}}}{- s}\rp|^2\,,\quad \\
\lp|\swei{a}{-s}_{\mc{I}^{+}}\rp|^2=\frac{(2\omega)^{2s+1}}{(4|\xi|^2+s^2)\prod_{j=1}^{\lfloor s\rfloor}\lp[4|\xi|^2+(s-j)^2\rp]}\frac{r_+-r_-}{2Mr_+}\lp|\swei{\tilde{\mathfrak{T}}}{-s}\rp|^2\,.
\end{gather*}
\end{itemize}
We deduce that, if $s$ is an integer,
\begin{gather}  \label{eq:identity-t-r-withsuperrad}
\lp|\swei{\mathfrak{R}}{s}\rp|^2+\frac{\omega-m\upomega_+}{\omega} \lp|\swei{\mathfrak{T}}{s}\rp|^2=1\,,\qquad\qquad
\lp|\swei{\tilde{\mathfrak{R}}}{s}\rp|^2+\frac{\omega}{\omega-m\upomega_+} \lp|\swei{\tilde{\mathfrak{T}}}{s}\rp|^2=1\,,
\end{gather}
and if $s$ is half-integer or if $a=0$, 
\begin{gather} \label{eq:identity-t-r-nosuperrad}
\lp|\swei{\mathfrak{R}}{s}\rp|^2+ \lp|\swei{\mathfrak{T}}{s}\rp|^2=1\,,\qquad\qquad
\lp|\swei{\tilde{\mathfrak{R}}}{s}\rp|^2+ \lp|\swei{\tilde{\mathfrak{T}}}{s}\rp|^2=1\,.
\end{gather}

From \eqref{eq:identity-t-r-nosuperrad}, one can immediately obtain the bound $G\equiv 1$ in the case $a=0$ or $s=1/2$. In the remaining cases, whenever the superradiant condition $\omega(\omega-m\upomega_+)< 0$ is met, the reflection coefficient can have absolute value greater than 1; however, noting that boundedness of the transmission coefficients follows directly from the computation \eqref{eq:def-t-r-2} and its analogue for $s<0$, together with Theorem \ref{thm:quantitative}, we can conclude using \eqref{eq:identity-t-r-withsuperrad}.
\end{proof}

\subsection{Estimates on the integral transformations}
\label{sec:estimates-u-tilde}

In this section, we present an estimate for the transformed equation, which is one of the main ingredients in the proof of Theorem \ref{thm:quantitative}.
\begin{proposition} \label{proposition:estimates-u-tilde}  Consider $|a|\leq M$ and let $\tilde{u}$ be as defined in \cref{prop:ode-u-tilde} for $|a|=M$ and as in \cref{prop:ode-u-tilde-sub} for $|a|<M$. Further define
\begin{equation*}
\zeta_1(x):=\begin{cases}
(x-r_+)(x-r_-)(x^2+a^2)^{-1} &\text{~if~} |a|<M\\
(x-M)(x-2M)(x^2+2M^2)^{-1}&\text{~if~} |a|=M
\end{cases}\,,
\end{equation*}
functions $f$ and $g$ by
\begin{align*}
f(x)&=\begin{dcases} 
1 &\text{if $|a|<M$}\\
x^{-1/4} &\text{if $a=M$ and $\omega(\omega-m\upomega_+)>0$}\\
x^{-1/4} \exp\lp(4\sqrt{-2M\omega(\omega-m\upomega_+)}x^{1/2}\rp)&\text{if $a=M$ and $\omega(\omega-m\upomega_+)\leq 0$}
\end{dcases}\,,\\
g(x)&=\begin{dcases} 
x^2 &\text{if $|a|<M$}\\
x^2 &\text{if $a=M$ and $\omega(\omega-m\upomega_+)>0$}\\
x^{5/2}\exp\lp(8\sqrt{-2M\omega(\omega-m\upomega_+)}x^{1/2}\rp)&\text{if $a=M$ and $\omega(\omega-m\upomega_+)\leq 0$}
\end{dcases}\,,
\end{align*}
and a coordinate $x^*$ such that $dx/dx^*=\zeta_1$ and $x^*(3M)=0$.

For any $(a,\omega,m,\lambda)\in [-M,M]\times\mathcal{A}$, we have
\begin{align}
\lp|\lp(f(x)\tilde{u}\rp)(+\infty)\rp|^2\leq G({C_\mc{A},M,|s|})\int_{-\infty}^\infty g(x) \lp|\tilde{H}\rp|^2 \zeta_1 dx^*\,, \label{eq:estimate-u-tilde}
\end{align} 
where the constant $G({C_\mc{A},M,|s|})$ is explicitly computable.
\end{proposition}

\begin{remark} \label{rmk:finite-u-tilde-estimates}
Proposition~\ref{proposition:estimates-u-tilde} gives an estimate for the boundary term of $\tilde{u}$ as $x\to \infty$ in terms of the inhomogeneity $\tilde{H}$ and an explicitly computable constant depending on $C_{\mc{A}}$, $M$ and $|s|$. 

Moreover, note that, as long as $\hat{F}^{[-s]}$ is smooth and compactly supported far from the horizon and away from $r=\infty$, $\tilde{H}$ (see definitions in \cref{prop:ode-u-tilde} and \cref{prop:ode-u-tilde-sub}) must decay faster than any polynomial as $x\to\infty$ (since it can be written as the a Fourier transform of $\hat{F}^{[-s]}$ extended to $\mathbb{R}$ by zero); moreover, we can apply \cref{lemma:technical-core-mode-stability-extremal} to obtain even more precise asymptotics for $\tilde{H}$ in the superradiant regime for the extremal case. These considerations are enough to show that the right hand side of the estimate in \cref{proposition:estimates-u-tilde} is finite for $\hat{F}^{[-s]}$ smooth and compactly supported in $r^*$. 
\end{remark}

\begin{proof}[Proof of Proposition~\ref{proposition:estimates-u-tilde}]
Recall the definition of the $h$ and $y$ currents: 
\begin{align*}
\tilde{Q}^y(x^*)&:=y|\tilde{u}'|^2+y\tilde{V}|\tilde{u}|^2\,,\quad (Q^y)'(x^*) = y'|\tilde{u}'|^2+(y\tilde{V})'|u|^2+2y\zeta_1\Re\lp[\tilde{u}'\overline{\tilde{H}}\rp]\,,\\
\tilde{Q}^h&:= h\Re\lp[\tilde{u}'\overline{\tilde{u}}\rp]-\frac12 h'|\tilde{u}|^2\,,\quad (\tilde{Q}^h)'= h |\tilde{u}'|^2 -\lp(h\tilde{V}+\frac12 h''\rp)|\tilde{u}|^2+ h\zeta_1 \Re(\tilde{u}\overline{\tilde{H}})\,,
\end{align*}
where the second identity in each line is obtained using  \eqref{eq:ode-u-tilde} if $|a|=M$ and \eqref{eq:ode-u-tilde-sub} if $|a|<M$. Throughout this proof, we will be heavily using the expressions for $\tilde{V}$ and the boundary conditions of $\tilde{u}$ as obtained in Propositions  \ref{prop:ode-u-tilde} and \ref{prop:ode-u-tilde-sub}, respectively. Moreover, we will use the notation $\lesssim_{C_{\mc{A}},M,|s|}$ whenever it is understood that the left hand side can be bounded by the product of the right hand side by an explicitly computable constant depending on $C_{\mc{A}}$, $M$ and $|s|$.

\medskip
\noindent\textit{Estimate near the horizon.} For the $x^*=-\infty$ end, we follow an identical approach to that in \cref{lemma:unique-continuation}. Decompose $\hat{V}(x^*):=\tilde{V}(x^*)-\omega_0^2$, where $\omega_0^2:=\tilde{V}(-\infty)>0$. Define
\begin{align*}
\hat{y}(x^*) := -\exp\lp(-C_1\int_{-\infty}^{x^*} \zeta_1(r^*) dr^*\rp)\,,
\end{align*}
where $\zeta_1(x^*)$, already defined, is a fixed positive function. Consequently, $\hat{y}'>0$ in $(-\infty,\infty)$, $\hat{y}(+\infty)=0$ and $\hat{y}(-\infty)=-1$. The fundamental theorem of calculus implies
\begin{align*}
\int_{-\infty}^\infty \lp(\hat{y}'|\tilde{u}'|^2+[\hat{y}'\omega_0^2-(\hat{y}\hat{V})']|\tilde{u}|^2\rp)dx^* = 2\tilde{Q}^T(-\infty)-\int_{-\infty}^\infty 2\hat{y}\zeta_1\Re\lp[\tilde{u}'\overline{\tilde{H}}\rp]dx^*\,,
\end{align*}
as $Q^{\hat{y}}(-\infty)=-2\tilde{Q}^T(-\infty)$.
The only term which can threaten this estimate on the left hand side is  $(\hat{y}\hat{V})'|u|^2$; however, this term cam be absorbed into the remaining ones on the left hand side by making $C_1$ large enough, as in the proof of \cref{lemma:unique-continuation}. Thus, after Cauchy--Schwarz on the integral with the inhomogeneity, we obtain
\begin{align*}
\int_{-\infty}^\infty \lp(\hat{y}'|\tilde{u}'|^2+\hat{y}'\omega_0^2|\tilde{u}|^2\rp)dx^* \lesssim \lp|\tilde{Q}^T(-\infty)\rp|+\int_{-\infty}^\infty \zeta_1 \lp|\tilde{H}\rp|^2dx^*\,.
\end{align*}

\medskip
\noindent\textit{Estimate as $x^*\to\infty$, for $|a|<M$.} For the $x^*=\infty$ end in the subextremal case, we follow the approach of \cite[Lemma 6.1]{Shlapentokh-Rothman2015}: define 
\begin{align*}
y(x^*) :=\exp\lp(-C_2\int_{x^*}^{+\infty} \zeta_2(r^*) dr^*\rp)\,,
\end{align*}
where $\zeta_2(x^*)$ is a fixed positive function which is identically 1 near $x^*=-\infty$ and is $(x^*)^{-2}$ near $x^*=+\infty$. In particular, $y'>0$ in $(-\infty,\infty)$, $y(+\infty)=1$ and $y(-\infty)=0$. The fundamental theorem of calculus implies
\begin{align*}
\int_{-\infty}^\infty \lp(y'|\tilde{u}'|^2+[y'\omega^2-(y\hat{V})']|\tilde{u}|^2\rp)dx^* = 2\tilde{Q}^T(\infty)-\int_{-\infty}^\infty 2y\zeta_1\Re\lp[\tilde{u}'\overline{\tilde{H}}\rp]dx^*\,,
\end{align*}
where $\hat{V}(x):=\tilde{V}(x)-\omega^2$ is $O(x^{-1})$ as $x\to \infty$. The only term which can threaten this estimate is  $(y\tilde{V})'|\tilde{u}|^2$. Let $C_3$ be a large constant to be chosen later and $\chi$ be a smooth function which is 1 on $(-\infty,C_3]$ and 0 on $[C_3+1,\infty)$. Define $\hat{V}_1:=\chi \hat{V}$ and  $\hat{V}_2:=(1-\chi) \hat{V}$. For the $\tilde{V}_1$ we can apply a similar trick as near the $x^*=-\infty$ end
\begin{align*}
\lp|\int_{-\infty}^\infty (y\hat{V}_1)'|\tilde{u}|^2 dx^* \rp|&\leq \int_{-\infty}^\infty  2|y||\tilde{u}'||\hat{V}_1||\tilde{u}| dx^* \\
&\leq \frac14 \int_{-\infty}^\infty \lp[y'|\tilde{u}'|^2+\lp(\frac{4C_3}{\omega C_2}\lp\lVert \frac{\hat{V} \chi}{C_3 \zeta_2}\rp\rVert_{\infty} \rp)^2y'\omega^2|\tilde{u}|^2 \rp]dx^* \,,
\end{align*}
and for $\hat{V}_2$ we note that $\hat{V}_2'$ has the same decay as $\zeta_2$, so
\begin{align*}
\lp|\int_{-\infty}^\infty (y\hat{V}_2)'|\tilde{u}|^2 dx^* \rp|&\leq \frac{1}{4}\int_{-\infty}^\infty  y'\omega^2\lp[\frac{\lVert r \hat{V_2}\rVert}{C_3 \omega^2} +\lp\lVert\frac{\hat{V}_2'}{\zeta_2 C_2 \omega^2}\rp\rVert_\infty \rp]|\tilde{u}|^2 dx^*\,.
\end{align*}
If we make $C_3$ large and then take $C_2\gg C_3$, after Cauchy--Schwarz on the integral with the inhomogeneity, we obtain
\begin{align*}
\int_{-\infty}^\infty \lp(y'|\tilde{u}'|^2+y'\omega^2|\tilde{u}|^2\rp)dx^* \lesssim \lp|\tilde{Q}^T(\infty)\rp|+\int_{-\infty}^\infty \zeta_2^{-1}\zeta_1^2 \lp|\tilde{H}\rp|^2dx^*\,.
\end{align*}
which, summing to the $\hat{y}$ current estimate from before, yields
\begin{align*}
\int_{-\infty}^\infty  x^{-2}\lp[|\tilde{u}'|^2+|\tilde{u}|^2\rp]\zeta_1 dx^* \lesssim \lp|\tilde{Q}^T(\infty)\rp|+\int_{-\infty}^\infty x^2 \lp|\tilde{H}\rp|^2\zeta_1dx^*\,.
\end{align*}

\medskip
\noindent\textit{Estimate as $x^*\to\infty$, for $|a|=M$ and $\omega(\omega-m\upomega_+)>0$.} For the extremal case, we begin with the non-superradiant frequencies. Define
\begin{gather*}
\tilde{f}=-\arctan(x)\,,\quad \frac{d\tilde{f}}{dx}=O(x^{-2})\,,\quad\frac{d^2\tilde{f}}{dx^2}=O(x^{-3})\,,\quad \frac{d^3\tilde{f}}{dx^3}=O(x^{-4})\,.
\end{gather*}
Set $f=\tilde{f}\chi$ for a smooth function $\chi$ which vanishes for $x\leq 3M$ and is identically 1 for $x\geq 4M$. Then, we have $\tilde{Q}^{h=f'}(\pm \infty)+\tilde{Q}^{h=f}(\pm \infty)=0$, so an application of the fundamental theorem of calculus yields
\begin{align*}
\int_{-\infty}^\infty \lp[2f'|\tilde{u}'|^2+\lp(f\tilde{V}'-\frac12 f'''\rp)|\tilde{u}|^2\rp]dx^* = -\int_{-\infty}^\infty \lp[2f \Re\lp(\tilde{u}'\overline{\tilde{H}}\rp)+f' \Re\lp(\tilde{u}\overline{\tilde{H}}\rp)\rp]\zeta_1 dx^*\,,
\end{align*}
We have
\begin{align*}
&f\tilde{V}'-\frac12 f''' =\lp(\frac{4\beta\gamma/M}{x^2}\arctan(x)+O(x^{-3})\rp) \text{~ as~} x\to \infty\,,
\end{align*}
but $f\tilde{V}'-\frac12 f'''$ may have a bad sign for $x\in[3M,R_1]$ with $R_1$ large. By rescaling $f$, this term of indeterminate sign can be made small enough to be absorbed by the left hand side of the ${\hat{y}}$ estimate from before. Then, after Cauchy--Schwarz on the term with an inhomogeneity, we obtain
\begin{align*}
\int_{-\infty}^\infty x^{-2}\lp[|\tilde{u}'|^2+|\tilde{u}|^2\rp]\zeta_1 dx^*\lesssim \int_{-\infty}^\infty x^2 |\tilde{H}|^2 \zeta_1 dx^*
\end{align*}

\medskip
\noindent\textit{Estimate as $x^*\to\infty$, for $|a|=M$ and $\omega(\omega-m\upomega_+)<0$.}  
Define $\tilde{h}(x)$ by
\begin{gather*}
\tilde{h}:=\exp\lp(8\sqrt{-\beta\gamma/M}x^{1/2}\rp)\,,\quad \frac{d\tilde{h}}{dx}=\lp(4\sqrt{-\beta\gamma/M}\rp)x^{-1/2}\tilde{h}(x)\,,\\
\frac{d^2\tilde{h}}{dx^2}=\lp(-16\beta\gamma/M-2\sqrt{-\beta\gamma/M}x^{-1/2}\rp)x^{-1}\tilde{h}(x)\,,
\end{gather*}
and $\tilde{y}(x)$ by
\begin{gather*}
\tilde{y}(x):=\frac{\tilde{h}(x)}{4\sqrt{-\beta\gamma/M}}\lp[x^{1/2}-\frac{A}{\sqrt{-\beta\gamma/M} }+Cx^{-1/2}\rp]\,,\\
\frac{d\tilde{y}}{dx}=\tilde{h}(x)\lp[1+\frac{1-2A}{8\sqrt{-\beta\gamma/M}}x^{-1/2}+Cx^{-1}+O(x^{-3/2})\rp]\,.
\end{gather*}
Set  $h=\chi\tilde{h}(1+Bx^{-1/2})$ and $y=-\chi\tilde{y}$ for a smooth function $\chi$ which vanishes for $x\leq 3M$ and is identically 1 for $x\geq 4M$.  Then
\begin{align*}
\tilde{Q}^{h}(+\infty)
&=\lp(-2\sqrt{-\beta\gamma/M}-\frac12\times 4\sqrt{-\beta\gamma/M}\rp)\lp|\exp\lp(4\sqrt{-\beta\gamma/M}x^{1/2}\rp)x^{-1/4}\tilde{u}\rp|^2(\infty)\\
&=-4\sqrt{\frac{-\beta\gamma}{M}}\lp|\exp\lp(4\sqrt{-\beta\gamma/M}x^{1/2}\rp)x^{-1/4}\tilde{u}\rp|^2(\infty)=-4\tilde{Q}^{y}(+\infty)\,,\\
\tilde{Q}^{h}(-\infty)&=\tilde{Q}^{y}(-\infty)=(-\infty)=0\,.
\end{align*}

Adding the $y$ and $h$ currents, an application of the fundamental theorem of calculus yields
\begin{align*}
&-\frac34\tilde{Q}^{h}(+\infty)+\int_{-\infty}^{\infty} \lp[(h+y')|\tilde{u}|^2+\lp(-h\mc{V}-\frac12 h''+(y\tilde V)'\rp)|\tilde{u}|^2\rp]dx^*\\
&\quad= -\int_{-\infty}^\infty \zeta_1\lp\{2 y\Re\lp[\tilde{u}'\overline{\tilde{H}}\rp]+h\Re\lp[\tilde{u}\overline{\tilde{H}}\rp]\rp\}dx^*\,,
\end{align*}
where, as $x\to\infty$,
\begin{align*}
h+y'=\tilde{h}\lp[\frac{-1+2A+8B}{8\sqrt{-\beta\gamma/M}}x^{-1/2}-Cx^{-1}+O(x^{-3/2})\rp]\,,
\end{align*}
and, recalling that 
\begin{align*}
\tilde{V}=\frac{4\beta\gamma/M}{x}-\frac{L+16\beta\gamma+\alpha^2}{x^2}+O(x^{-2})\,,
\end{align*}
we have
\begin{align*}
&-h\mc{V}-\frac12 h''+(y\tilde{V})' \\
&\quad= -\frac12(8B+2A-1)\sqrt{-\beta\gamma/M}\tilde{h}x^{-3/2}\\
&\quad\qquad+\tilde{h}\lp[\lp(-\frac{3}{16}+\frac{A}{2}+3B+L+\alpha^2+\frac{4\beta\gamma}{M}(4M-C)\rp)x^{-2}+O(x^{-5/2})\rp]  &\text{~as~} x\to \infty\,.
\end{align*}
Choosing
\begin{align*}
A=\frac14 +4(L+\alpha^2)+96\beta\gamma\,,\qquad B=\frac{3}{16}-(L+\alpha^2)-24\beta\gamma\,,\qquad C=-4M\,,
\end{align*}
we obtain finally
\begin{align*}
h+y'&=\tilde{h}\lp[\frac{4M}{x}+O(x^{-3/2})\rp] \text{~as~} x\to \infty\,,\\
-h\mc{V}-\frac12 h''+(y\tilde{V})' &= \tilde{h}\lp[\frac{1}{2x^2}+O(x^{-5/2})\rp]  \text{~as~} x\to \infty\,.
\end{align*}

As before, these weights may have a bad sign for $x\in[3M,R_1]$, where $R_1$ is some large constant. By rescaling $h$ and $y$, this term of indeterminate sign can be made small enough that it can be absorbed by the ${\hat{y}}$ estimate's left hand side. Then, after Cauchy--Schwarz on the term with an inhomogeneity, we obtain
\begin{align*}
&\int_{-\infty}^\infty  x^{-3/2} \lp[\lp|\exp\lp(4\sqrt{-\beta\gamma/M}x^{1/2}\rp)x^{1/4}\tilde{u}'\rp|^2+\lp|\exp\lp(4\sqrt{-\beta\gamma/M}x^{1/2}\rp)x^{-1/4}\tilde{u}\rp|^2\rp] \zeta_1dx^* \\
&\qquad+\lp|\lp[\exp\lp(4\sqrt{-\beta\gamma/M}x^{1/2}\rp)x^{-1/4}\tilde{u}\rp](+\infty)\rp|^2\\
&\qquad\lesssim \lp|\tilde{Q}^T(-\infty)\rp|+ \int_{-\infty}^\infty x^{5/2} \lp|\exp\lp(4\sqrt{-\beta\gamma/M}x^{1/2}\rp)x^{-1/4}\tilde{H^*}\rp|^2\zeta_1dx^*\,.
\end{align*}
Note that the estimate holds even if $\omega(\omega-m\upomega_+)=0$.

\medskip
\noindent\textit{The $T$ estimate for $|a|\leq M$.}  Consider the $\tilde{Q}^T$ current,
\begin{align*}
\tilde{Q}^T := \Im \lp(\tilde{u}'\overline{\omega \tilde{u}}\rp)\,, \quad \lp(\tilde{Q}^T\rp)' :=\zeta_1 \omega \Im \tilde{H}\overline{u}\,,
\end{align*}
where the second identity is obtained from \eqref{eq:ode-u-tilde-sub} and \eqref{eq:ode-u-tilde}. We recall that
\begin{align*}
\tilde{Q}^T(\infty)=\begin{cases}
\omega^2\lp|\tilde{u}(\infty)\rp|^2\,, & \text{~if~} |a|\leq M\\
4|\omega|\sqrt{2M\omega(\omega-m\upomega_+)}\lp|\lp(x^{-1/4}\tilde{u}\rp)(\infty)\rp|^2\,, & \text{~if~} |a|= M \text{~and~} \omega(\omega-m\upomega_+)>0\\
0\,, & \text{~if~} |a|= M \text{~and~} \omega(\omega-m\upomega_+)<0
\end{cases}\,,
\end{align*}
using the boundary conditions in Propositions~\ref{prop:ode-u-tilde} and \ref{prop:ode-u-tilde-sub}.

By the fundamental theorem of calculus,
\begin{align*}
\lp|\tilde{Q}^T(\infty)\rp|+\lp|\tilde{Q}^T(-\infty)\rp|\lesssim \int_{-\infty}^\infty\Im\lp(\zeta_1\overline{\tilde{H}}\tilde{u}\rp)dx^*\,,
\end{align*}
which we can Cauchy--Schwarz in the most appropriate way to combine with the previous estimates to finish the proof.
\end{proof}

\subsection{Proof of Theorem \ref{thm:quantitative}}
\label{sec:quantitative-proof}

In this section, we present the proof of Theorem~\ref{thm:quantitative}.  We first note that the definition of the Wronskian  via \eqref{def:Wronskian} cannot easily yield Theorem~\ref{thm:quantitative}, because the representations we have for $\swei{u}{s}_{\mc{I}^+}$ and $\swei{u}{s}_{\mc{I}^+}$ are merely asymptotic near $r=\infty$ and $r=r_+$, respectively. However, the Wronskian can be used to construct a Green's function for the inhomogeneous version of the radial ODE \eqref{eq:radial-ODE}:
\begin{lemma} \label{lemma:Greens-function} Let $\swei{\hat{F}}{s}$ be compactly supported in $(r_+,\infty)$ and set $\swei{H}{s}=\Delta^{1+s/2}(r^2+a^2)^{-3/2}\swei{\hat{F}}{s}$. Then, if $\lp(\swei{\mathfrak{W}}{s}\rp)^{-1}$ is finite,
\begin{align*}
\swei{u}{s}(r^*)=\frac{1}{\swei{\mathfrak{W}}{s}}\lp[\swei{u}{s}_{\mc{I}^+}(r^*)\int_{-\infty}^{r^*} \swei{u}{s}_{\mc{H}^+}(x^*)\swei{H}{s}(x^*)dx^*+\swei{u}{s}_{\mc{H}^+}(r^*)\int_{r^*}^{\infty} \swei{u}{s}_{\mc{I}^+}(x^*)\swei{H}{s}(x^*)dx^*\rp]
\end{align*}
solves the inhomogenous radial ODE \eqref{eq:u-Schrodinger} with boundary conditions compatible with those in Definition \ref{def:outgoing-radial-solution}.
\end{lemma}
\noindent Hence, if we choose a particular inhomogeneity and obtain suitable estimates on the solution to the the radial ODE with outgoing boundary condition and the inhomogeneity of our choice, we also obtain a lower bound on $\mathfrak{W}$. This will be our strategy:

\begin{proof}[Proof of \cref{thm:quantitative}]
We will begin by establishing the strong bound in $\mc{A}$. Let $s>0$ and $(\omega,m,\lambda)$ be any admissible frequency triple in the bounded range $\mc{A}$. Throughout the proof, we will use the notation $\lesssim_{C_{\mc{A}},M,|s|}$ whenever it is understood that the left hand side can be bounded by the product of the right hand side by an explicitly computable constant depending on $C_{\mc{A}}$, $M$ and $|s|$, and $\sim_{C_{\mc{A}},M,|s|}$ when the left hand side can also be bounded by the right hand side in this way. 

For $s\geq 0$, suppose we specify a smooth function $H^{[-s]}$, compactly supported away from $r_+$. Let $\swei{u}{-s}$ be the solution to the radial ODE \eqref{eq:u-Schrodinger} with inhomogeneity $\swei{H}{-s}$ and boundary conditions compatible with those in Definition \ref{def:outgoing-radial-solution}. Using the Teukolsky--Starobinsky identities \eqref{prop:TS-radial} to define
\begin{align*}
u^{[+s]}:=\Delta^{s/2}(r^2+a^2)^{1/2}\lp(\mc{D}_0^-\rp)^{2s}\lp(\Delta^{s/2}(r^2+a^2)^{-1/2}u^{[-s]}\rp)\,,
\end{align*}
we obtain a solution to the radial ODE of spin $+s$ with inhomogeneity (see \cref{lemma:TS-commutation})
\begin{align*}
H^{[+s]}:=(r^2+a^2)^{1/2}\Delta^{s/2}\lp(\mc{D}_0^+\rp)^{2s}\lp(\Delta^{s/2}(r^2+a^2)^{-1/2}H^{[-s]}\rp)\,.
\end{align*}
Moreover, assuming the support of $\swei{H}{-s}$ is sufficiently far from $r^*=\pm \infty$ that it does not affect the first $s$ coefficients of the asymptotic representation for $R^{[-s]}$ at either end, we conclude, by the proof of Lemma~\ref{lemma:TS-radial-boundary-conditions}, that $u^{[+s]}$ also has outgoing boundary conditions.  We can thus write, by Lemma~\ref{lemma:Greens-function},
\begin{align}
\begin{split}\label{eq:Wronskian-identity-Green}
&\lp|\int_{-\infty}^\infty u^{[- s]}_{\mc{I}^+}(r^*) H^{[- s]}(r^*)dr^*\rp|^2 \lp(\swei{\mathfrak{W}}{-s}\rp)^{-2} = \lp|\lp(\Delta^{- s/2} u^{[- s]}\rp)(-\infty)\rp|^2\\
&\lp|\int_{-\infty}^\infty u_{\mc{I}^+}^{[+ s]} (r^2+a^2)^{1/2}\Delta^{s/2}\lp(\mc{D}_0^+\rp)^{2s}\lp(\Delta^{s/2}(r^2+a^2)^{-1/2}H^{[-s]}\rp)dr^*\rp|^2\lp(\swei{\mathfrak{W}}{+s}\rp)^{-2} \\
&\quad= \lp|\lp((r^2+a^2)^{1/2}\lp(\mc{D}_0^-\rp)^{2s}\lp(\Delta^{s/2}(r^2+a^2)^{-1/2}u^{[-s]}\rp)\rp)(-\infty)\rp|^2
\,.
\end{split}
\end{align}

As the potential in the ODE for $u$ is, in general, complex, obtaining such bounds directly from the radial ODE \eqref{eq:u-Schrodinger} can be quite difficult (even for $a=0$, this is not the usual strategy in the literature, see \cite{Dafermos2016a}). The key estimate to circumvent these issues was obtained in statements 4(b) of Proposition~\ref{prop:ode-u-tilde} and \ref{prop:ode-u-tilde-sub}:
\begin{align}
\lp|\Delta^{-s/2}u^{[-s]}(-\infty)\rp|^2 \sim_{C_{\mc{A}},M,|s|} \lp|(f(x^*)\tilde{u})(+\infty)\rp|^2  \label{eq:crucial-identity}\,,
\end{align}
for some appropriate $f$ determined in \cref{proposition:estimates-u-tilde}. It follows from \cref{lemma:TS-radial-boundary-conditions} that
\begin{align*}
\lp|\lp((r^2+a^2)^{1/2}\lp(\mc{D}_0^-\rp)^{2s}\lp(\Delta^{s/2}(r^2+a^2)^{-1/2}u^{[-s]}\rp)\rp)(-\infty)\rp|^2 &\sim_{C_{\mc{A}},M,|s|}\lp|(f(x^*)\tilde{u})(+\infty)\rp|^2\,.
\end{align*}
It immediately follows that, in the frequency parameter range $\mc{A}$, \cref{proposition:estimates-u-tilde} an estimate, depending on $M$, $|s|$ and $C_{\mc{A}}$, for the boundary terms at $r^*=-\infty$ in the right hand side of \eqref{eq:Wronskian-identity-Green} in terms of the transformed inhomogeneity $\tilde{H}$ and thus in terms of $\swei{H}{-s}$. 

Hence, to obtain an explicit lower bound for $\mathfrak{W}$, we have to construct a suitable smooth $H^{[-s]}$ which is compactly supported very far from $r=r_+$ and $r=\infty$ (so that the first $s$ coefficients of the series expansion for $\swei{u}{-s}$ at either end are the same as for a homogeneous solution of outgoing boundary conditions) so that the left and right hand sides of \eqref{eq:Wronskian-identity-Green} are finite. For the right hand side, the assumption on the support is sufficient, by Remark~\ref{rmk:finite-u-tilde-estimates}. For the left hand side, we must work a bit more.

Recall that, for $r>r_+$,  we can determine the precise behavior of $u_{\mc{I}^+}$, against which the inhomogeneities are integrated:
\begin{align*}
\swei{u}{\pm s}_{\mc{I}^+}=e^{i\omega r}r^{2iM \omega \mp s}\lp[\sum_{k=0}^{2s} \swei{c}{\pm s}_kr^{-k}+O\lp(r^{-2s-1}\rp)\rp] \,,
\end{align*}
for $\swei{c}{\pm s}_0=1$, some complex $\swei{c}{\pm s}_k=\swei{c}{\pm s}_k(\omega,m,\lambda)$ and where the remainder can be explicity estimated (by adapting the proof of \cite[Lemma C1]{Shlapentokh-Rothman2014}, for instance). Fix $R_1$ and $R_2$ such that $r_+\ll R_1<R_2$, $1/R_2\gg 0$, fix a bump function $\chi$ supported in $[R_1,R_2]$ and choose $\swei{b}{-s}_k$  for $k=0,...,s$. Set
\begin{align*}
H^{[-s]}=\sum_{k=0}^s \swei{b}{-s}_k e^{-i\omega r}r^{-2iM \omega - s-k}\chi\,; 
\end{align*}
then
\begin{align*}
H^{[+s]}=e^{-i\omega r}r^{-2iM \omega}\lp[\sum_{k=0}^s \swei{b}{+s}_k r^{ s-k}\chi+\mathbbm{1}_{\supp \chi'}h\rp]\,,
\end{align*}
where, by a proof similar to \cref{lemma:TS-radial-boundary-conditions}, we find $\swei{b}{+s}_k$ and $h$ can be explicitly computed in terms of $\swei{b}{-s}_k$, $\omega$, $m$ and $s$. We thus have
\begin{align*}
\swei{u}{- s}_{\mc{I}^+}\swei{H}{- s}&=\sum_{k=0}^s\lp[\sum_{j=0}^{2s} \swei{b}{- s}_k\swei{c}{- s}_j r^{-k-j}+O\lp(r^{-2s-k-1}\rp)\rp]\chi\,,\\
\swei{u}{+ s}_{\mc{I}^+}\swei{H}{+ s}&=\sum_{k=0}^s \lp[\sum_{j=0}^{2s}\swei{b}{+ s}_k\swei{c}{+ s}_j r^{-k-j}+O\lp(r^{-2s-1-k}\rp)\rp]\chi \\
&\qquad+h\lp[\sum_{j=0}^{2s}\swei{c}{+ s}_j r^{-s-j}+O(r^{-3s-1})\rp]\mathbbm{1}_{\supp \chi'}\,,
\end{align*}
which is certainly integrable and where the terms denoted by $O$ can be explicitly computed. With this choice of $H^{[-s]}$, we finally obtain from \eqref{eq:Wronskian-identity-Green}
\begin{align*}
\lp(\swei{\mathfrak{W}}{\pm s}\rp)^{-2} &\lesssim_{C_\mc{A},M,|s|}\lp[\int_{2M}^\infty g(x)\lp|\tilde{H}\rp|^2dx\rp]\times\\
&\qquad\times\lp[\int_{R_1}^{R_2} \sum_{k=0}^s \lp|\sum_{j=0}^{2s}\swei{b}{\pm s}_k\swei{c}{\pm s}_j (\omega,m,\lambda)r^{-k-j}+O_{\omega,m,\lambda}\lp(r^{-2s-1-k}\rp)\rp|dr^*\rp.\\
&\quad\quad \lp.+\int_{R_1}^{R_2} |h|\sum_{j=0}^{2s} \lp|\swei{c}{+ s}_j (\omega,m,l)r^{-s-j}+O_{\omega,m,\lambda}\lp(r^{-3s-1}\rp)\rp|dr^*\rp]^{-2}\\
&\quad\lesssim_{C_\mc{A},M,|s|} 1\,, \numberthis \label{eq:Wronskian-explicit-bound}
\end{align*}
where $g(x)$ can be read off the statement of Proposition~\ref{proposition:estimates-u-tilde} and $\tilde{H}$ is given by \eqref{eq:H-tilde} or \eqref{eq:H-tilde-sub}  if $|a|<M$, under the identification $\hat{F}=\Delta^{-1+s/2}(r^2+a^2)^{3/2}\swei{H}{-s}$. We remark that the bound \eqref{eq:Wronskian-explicit-bound} can be explicitly computed because $R_1$, $R_2$, $\swei{b}{\pm s}_k$  and thus $\swei{H}{-s}$ have been chosen and, moreover, $\swei{c}{\pm s}_j(\omega,m,\lambda)$, $h$ and the dependence of $O_{\omega,m,\lambda}\lp(r^{-2s-1-k}\rp)$ on the frequency parameters are all determined uniquely by the condition that $\swei{c}{\pm s}_0=1$ and the theory of asymptotic analysis for the radial ODE.
\end{proof}

\section{Fixed-frequency solutions at zero frequency and at the superradiance threshold}
\label{sec:remark-threshold}

In this section, we discuss the fixed-frequency solutions characterized by the real frequencies our Theorems~\ref{thm:mode-stability-real-axis} and \ref{thm:quantitative} do not cover: $\omega=0$ and, if $|a|=M$, $\omega=m\upomega_+$. 

The most interesting question concerning such frequencies is determining what is the behavior of the Wronskian from Definition \ref{def:Wronskian} (and, hence, of the transmission and reflection coefficients from Definition~\ref{def:t-r-coeffs}) in the limit $\omega\to 0$ and, if $|a|=M$, $\omega\to m\upomega_+$. As there is no \textit{a priori} way of defining mode solutions at these frequencies that would correspond to ``continuous limits'' of Definition~\ref{def:mode-solution} when $\omega\to 0$ or $\omega\to m\upomega_+$; thus, this question is out of the scope of mode stability. However, in this section, we attempt to shed some light into these these limits: in Section~\ref{sec:finite-energy}, we discuss finite energy solutions at these frequencies, highlighting the similarities between the two limits; in Section~\ref{sec:quantitative-strengthening}, we obtain a quantitative lower bound on the decay rate for the Wronskian in the double limit $\omega\to m\upomega_+$ and $a\to M$.

\subsection{Finite energy solutions at zero frequency and at the superradiance threshold}
\label{sec:finite-energy}

Though there is no definition of mode solution that can be \textit{a priori} seen as continuous in the limits $\omega\to 0$ and, if $|a|=M$, $\omega\to m\upomega_+$, one can nonetheless investigate the existence of \textit{finite energy solutions} at such frequencies. The following two propositions show that, if $am\neq 0$, there there are no nontrivial solutions to the homogeneous radial ODE~\eqref{eq:radial-ODE} with finite energy along a hyperboloidal hypersuface (see Figure~\ref{fig:hypersurfaces}b).

\begin{proposition} \label{prop:mode-zero-freq} Fix $s\in\frac12\mathbb{Z}$, $M>0$ and $|a|\leq M$. Let $m$ be an admissible parameter with respect to $s$. Then, if $\omega=0$ and $a m\neq 0$ there is no nontrivial solution to the homogeneous radial ODE~\eqref{eq:radial-ODE} with finite energy along a hyperboloidal hypersuface (see Figure~\ref{fig:hypersurfaces}b).
\end{proposition}
\begin{proof}
For $\omega=0$, the homogeneous radial ODE~\eqref{eq:radial-ODE} has a regular singularity at $r=\infty$ (contrast with the irregular singularity for $\omega\neq 0$). An asymptotic analysis shows that solutions to the homogeneous radial ODE~\eqref{eq:radial-ODE} are given by  (see \cite[Chapter 5]{Olver1973}, e.g.) the following expansion as $r\to \infty$,
\begin{align*}
\swei{R}{s}(r)&= r^{\lp(s+\frac12\rp)-\sqrt{s^2+\frac14+\lambda}}\lp[c_1+O(r^{-1})\rp]+r^{\lp(s+\frac12\rp)+\sqrt{s^2+\frac14+\lambda}}\lp[c_2+O(r^{-1})\rp]\,,
\end{align*}
if $\sqrt{s^2+\frac14+\lambda}\neq 0$, and if, $\sqrt{s^2+\frac14+\lambda}=0$,
\begin{align*}
\swei{R}{s}(r)&= r^{s+\frac12}\lp[\frac{c_1}{r}-c_2\log r  +O(r^{-2})\rp]\,.
\end{align*}

By \cite[Lemma D.4]{Shlapentokh-Rothman2015}, we find that, for $s=0$, $R$ corresponds to a solution of finite energy if and only if
\begin{align*}
\int_{r_0}^\infty \lp|\frac{\Delta}{r^2+a^2}\frac{d}{dr}R\rp|^2r^2 dr <\infty
\end{align*}
for some $r_0>3M$ large, which holds only if $\lambda+1/4>1$ and $c_2=0$. The finite energy condition can be generalized for $s\in\frac12\mathbb{Z}$, if we take into account the properties of the algebraically special frame used to derive the Teukolsky equations: we require $\lambda+s^2+1/4\geq (1+|s|)^2$. (We note that, choosing $\lambda=\bm\uplambda$ from Proposition~\ref{def:angular-ode}, we have $\bm\uplambda+s^2= l(l+1)$ if $\omega=0$, hence $\bm\uplambda+s^2+1/4\geq (1+|s|)^2 \Leftrightarrow l\geq |s|+1$.) For $s\geq 0$, the Wronskian in Proposition~\ref{prop:wronskian-identity} becomes
\begin{align*}
W\lp[\swei{u}{+s},\overline{\swei{u}{-s}}\rp]&=\Delta^{-s+1}\lp(\frac{d}{dr}\lp(\Delta^s\swei{R}{+s}\rp)\overline{\swei{R}{-s}}-\lp(\Delta^s\swei{R}{+s}\rp)\frac{d}{dr}\overline{\swei{R}{-s}}\rp)\\
&=2\lp(1+2s-\sqrt{\lambda+s^2+1/4}\rp)r^{2s-2\sqrt{\lambda+s^2+1/4}}\text{~ as }r\to\infty\,,
\end{align*}
 so, under the finite energy assumption, $W(+\infty)=0$. By conservation of the Wronskian, (see proof of Proposition~\ref{prop:wronskian-identity})
  \begin{align*}
0= W(\infty)=W(-\infty)=2im\upomega_+ |\Delta^{\pm s/2} \swei{u}{\pm s}|(-\infty)\,.
\end{align*}
If $a\neq 0$ and $m\neq 0$, we can use a unique continuation argument as in Lemma~\ref{lemma:unique-continuation} to infer that $u\equiv 0$. 
\end{proof}

\begin{proposition} \label{prop:mode-superrad-threshold} Fix $s\in\frac12\mathbb{Z}$, $M>0$ and $|a|= M$. Let $m$ be an admissible parameter with respect to $s$. Then, if $\omega=m\omega_+$ and $m\neq 0$ there is no nontrivial solution to the homogeneous radial ODE~\eqref{eq:radial-ODE} with finite energy along a hyperboloidal hypersurface (see Figure~\ref{fig:hypersurfaces}b).
\end{proposition}
\begin{proof}
For $\omega=m\upomega_+$,  the homogeneous radial ODE~\eqref{eq:radial-ODE} has a regular singularity at $r=\infty$ (contrast with the irregular singularity for $\omega\neq 0$). An asymptotic analysis shows that solutions to the homogeneous radial ODE~\eqref{eq:radial-ODE} are given by  (see \cite[Chapter 5]{Olver1973}, e.g.) the following expansion as $r\to M$,
\begin{align*}
\swei{R}{s}(r)&= (r-M)^{-\lp(s+\frac12\rp)+\sqrt{s^2+\frac14+\lambda+M^2\omega^2-2m^2}}\lp[c_1+O(r^{-1})\rp]\\
&\qquad +(r-M)^{-\lp(s+\frac12\rp)-\sqrt{s^2+\frac14+\lambda+M^2\omega^2-2m^2}}\lp[c_2+O(r^{-1})\rp]\,,
\end{align*}
if $\sqrt{s^2+\frac14+\lambda}\neq 0$, and if, $\sqrt{s^2+\frac14+\lambda}=0$,
\begin{align*}
\swei{R}{s}(r)&= (r-M)^{-s-\frac12}\lp[c_1(r-M)+c_2\log (r-M)  +O(r^{-2})\rp]\,.
\end{align*}

By the proof of Lemma~\ref{lemma:smooth-extension-horizon}, in Kerr-star coordinates, which are regular along $\mc{H}^+$
\begin{align*}
\swei{\upalpha}{s}(r) = (r-M)^{2iM\omega}\swei{R}{s}(r)e^{-i\omega t^*}e^{im \phi^*}S^{[s],\,(a\omega)}_{m\lambda}(\theta)
\end{align*}
For $s=0$, the energy flux along the future event horizon controls the integral of $\lp|\p_t^*\upalpha\rp|^2$ along $\mc{H}^+$. Hence, a finite energy solution must have $\lp|\p_t^*\upalpha\rp|<\infty$ at $r=M$, so we require $c_2=0$ and $\lambda+M^2\omega^2-2m^2\geq 0$. If we take into account the properties of the algebraically special frame used to derive the Teukolsky equations, the latter condition is generalized for $s\in\frac12\mathbb{Z}$ as follows (see also \cite{Richartz2017}):
\begin{align*}
\sqrt{s^2+\frac14+\lambda+M^2\omega^2-2m^2}\geq |s|+\frac12\Leftrightarrow \lambda+M^2\omega^2-2m^2+s^2\geq 0\,.
\end{align*}

For $s\geq 0$, the Wronskian in Proposition~\ref{prop:wronskian-identity} becomes
\begin{align*}
W\lp[\swei{u}{+s},\overline{\swei{u}{-s}}\rp]&=\Delta^{-s+1}\lp(\frac{d}{dr}\lp(\Delta^s\swei{R}{+s}\rp)\overline{\swei{R}{-s}}-\lp(\Delta^s\swei{R}{+s}\rp)\frac{d}{dr}\overline{\swei{R}{-s}}\rp)\\
&=2\lp(-\frac12 +s+\sqrt{s^2+\frac14+\lambda+M^2\omega^2-2m^2}\rp)\times \\
&\qquad\times (r-M)^{2\sqrt{s^2+\frac14+\lambda+M^2\omega^2-2m^2}}\lp(1+O(r-M)\rp)\,,
\end{align*}
as $r\to M$, so, under the finite energy assumption, $W(-\infty)=0$. By conservation of the Wronskian, (see proof of Proposition~\ref{prop:wronskian-identity})
  \begin{align*}
0= W(-\infty)=W(+\infty)=2i\omega \lp|\Delta^{\pm s/2} \swei{u}{\pm s}\rp|(+\infty)\,.
\end{align*}
If $\omega=m\upomega_+\neq 0$, we can use a unique continuation argument as in Lemma~\ref{lemma:unique-continuation} to infer that $u\equiv 0$. 
\end{proof}

It is important to keep in mind that, while providing some reassurance, \textbf{Propositions~\ref{prop:mode-zero-freq} and \ref{prop:mode-superrad-threshold} cannot be used to infer quantitative bounds on the Wronskian} from Definition~\ref{def:Wronskian} in the limits $\omega\neq 0$ or $\omega\neq m\upomega_+$, if $|a|=M$, in Theorem~\ref{thm:quantitative}. 

\subsection{A quantitative lower bound on the decay rate of the Wronskian near the superradiant threshold}
\label{sec:quantitative-strengthening}

Indeed, to understand the rate of blow-up (or not) of the inverse of the Wronskian in these limits, another approach is needed: for instance, applying multiplier currents to the radial ODE~\ref{eq:radial-ODE} that are well-adapted to the limits $\omega\neq 0$ or $\omega\neq m\upomega_+$ or, alternatively, tracking the dependence of $|\omega|^{-1}$ and $|\omega-m\upomega_+|^{-1}$ in the proof of Theorem~\ref{thm:quantitative}.

The use of multiplier currents  has proved sucessful in the case $s=0$ and $|a|<M$: the complete scattering theory obtained in \cite{Dafermos2014}, which implies, in particular, that the Wronskian is at least of size $|\omega|$ in the limit $\omega\to 0$, was obtained by use of multiplier currents, adapted to the  radial ODE for small $\omega$, in \cite[Propositions 8.7.1 to 8.7.3]{Dafermos2016b}. The proof given there can be easily be extended to cover the full range $|a|\leq M$ and $\omega$ small for $s=0$. Similarly, multiplier currents can be employed to show that the Wronskian is at least of size $|\omega-m\upomega_+|$ in the limit $\omega\to m\upomega_+$, for $|a|=M$ and $s=0$. These techniques are  out of the scope of the present paper and will appear elsewhere. 

On the other hand, a closer look into the proof of Theorem~\ref{thm:quantitative} reveals that, while divisions by $|\omega|$ are abundant, divisions by $|\omega-m\upomega_+|$ are relatively scarce. Thus, one can easily deduce 
\begin{proposition}[Stronger quantitative bound on the Wronskian] Fix $M>0$ and $s\in\frac12\mathbb{Z}$. Let $\mc{A}'$ be a set of frequency parameters $(\omega,m,\lambda)$ admissible with respect to $s$ such that $\omega$ is real and
\begin{align*}
C_{\mc{A}'}&:=\sup_{(\omega,m,l)\in\mc{A}'}\lp(|\omega|+|\omega|^{-1}+|m|+\lp|\lambda\rp|\rp)<\infty\,.
\end{align*}
Then
\begin{align}
\begin{gathered}
\sup_{(a,\,\omega,\,m,\,\lambda)\in[-M,M]\times\mc{A}'}\lp(\lp|\omega-m\upomega_+\rp|^{1-2s}\delta_{|a|,M}+1\rp)\lp|\swei{\mathfrak{W}}{s}\rp|^{-2}\\
\leq G(C_{\mc{A}'},M,|s|)<\infty\,,
\end{gathered}\label{eq:Wronskian-bounded-strong}
\end{align}
where $G$ will be given in an explicitly computable manner by \eqref{eq:Wronskian-explicit-bound-strong}.
\label{prop:quantitative-strong}
\end{proposition}

\begin{remark}
It is important to highlight already the difference between Theorem~\ref{thm:quantitative} and Proposition~\ref{prop:quantitative-strong}. 

Theorem~\ref{thm:quantitative} should be seen as containing exactly the same information as mode stability or, concretely, Corollary~\ref{cor:nonexplicit-quantitative-mode-stability} (the result is a bound on the Wronskian directly which holds only away from frequencies excluded from the mode stability statement, Theorem~\ref{thm:mode-stability-real-axis}), albeit presented in an explicitly computable (in terms of the frequency range $\mc{A}$, the black hole mass $M$ and the Teukolsky spin $s$) manner.

By contrast, Proposition~\ref{prop:quantitative-strong} goes {\normalfont beyond} mode stability, by providing an upper bound for the rate of blow-up of $|\mathfrak{W}|^{-1}$ in the double limit $\omega\to m\upomega_+$ and $|a|\to M$.  
\end{remark}

An immediate consequence is a lower bound for the decay rate of the transmission coefficients in scattering
\begin{corollary}[of Proposition~\ref{prop:quantitative-strong}]\label{cor:bddness-scattering-stronger}
Fix $M>0$ and $s\in\{0,\frac12 ,1,\frac32 ,2\}$. Let $\mc{A}'$ be the set of frequency parameters $(\omega,m,\lambda)$ introduced in Proposition~\ref{prop:quantitative-strong}. Recall the transmission coefficients introduced in Definition~\ref{def:t-r-coeffs}. We have
\begin{align*}
\lp(\lp|\omega-m\upomega_+\rp|^{\mp(1+2s)}\delta_{|a|,M}+1)+1\rp)\lp|\swei{\mathfrak{{T}}}{\pm s}\rp|^2\leq G(C_{\mc{A}'},M,|s|)<\infty\,,\\
\lp(\lp|\omega-m\upomega_+\rp|^{\pm(1-2s)}\delta_{|a|,M}+1)+1\rp)\lp|\swei{\tilde{\mathfrak{T}}}{\pm s}\rp|^2\leq G(C_{{\mc{A}'}},M,|s|)<\infty\,,
\end{align*}
where $G$ can be given explicitly.
\end{corollary}

\begin{proof}
The proof follows easily from the expressions for the transmission coefficients given in \eqref{eq:def-t-r-2} for $s\geq 0$, and their analogue for $s<0$.
\end{proof}

To conclude, we outline the proof of Proposition~\ref{prop:quantitative-strong}:

\begin{proof}[Proof sketch for Proposition~\ref{prop:quantitative-strong}]  We refer the reader to the proof of Theorem~\ref{thm:quantitative}, in Section~\ref{sec:quantitative-proof}, to fill in the gaps in the sketch we present. Let $s\geq 0$ for simplicity and let $(\omega,m,\lambda)\in\mc{A}'$. Once again, we use the notation $\lesssim_{C_{\mc{A}'},M,|s|}$ to denote the existence of a constant which can be explicitly computed in terms of $C_{\mc{A}'}$, $M$  and $|s|$

Note that, if $|a|=M$, \eqref{eq:crucial-identity} in that proof can be replaced by
\begin{align}
\lp|\Delta^{-s/2}u^{[-s]}(-\infty)\rp|^2 \sim_{C_{\mc{A}'},M,|s|} \lp|\omega-m\upomega_+\rp|^{-1/2-2s}\lp|(f(x^*)\tilde{u})(+\infty)\rp|^2  \label{eq:crucial-identity-strong}\,,
\end{align}
according to statements 4(b) of Proposition~\ref{prop:ode-u-tilde}, for some appropriate $f$ given in \cref{proposition:estimates-u-tilde}. From Lemma~\ref{lemma:TS-radial-boundary-conditions} and Proposition~\ref{prop:TS-radial} in the case of outgoing boundary conditions, we obtain
\begin{align*}
&\lp|\lp((r^2+a^2)^{1/2}\lp(\mc{D}_0^-\rp)^{2s}\lp(\Delta^{s/2}(r^2+a^2)^{-1/2}u^{[-s]}\rp)\rp)(-\infty)\rp|^2 \\
&\quad\sim_{C_{\mc{A}'},M,|s|}\lp|\omega-m\upomega_+\rp|^{4s}\lp|\Delta^{-s/2}u^{[-s]}(-\infty)\rp|^2 \\
&\quad\sim_{C_{\mc{A}'},M,|s|} \lp|\omega-m\upomega_+\rp|^{-1/2+2s}\lp|(f(x^*)\tilde{u})(+\infty)\rp|^2\,.
\end{align*}

Moreover, in the proof of Proposition~\ref{proposition:estimates-u-tilde} given in Section~\ref{sec:estimates-u-tilde}, we show the slightly stronger statement (compared to \eqref{eq:estimate-u-tilde}) 
\begin{align}
|\omega-m\upomega_+|^{1/2}\lp|\lp(f(x)\tilde{u}\rp)(+\infty)\rp|^2\lesssim_{C_{\mc{A}'},M,|s|}\int_{-\infty}^\infty g(x) \lp|\tilde{H}\rp|^2 \zeta_1 dx^*\,,\label{eq:estimate-u-tilde-stronger}
\end{align}
for $|a|=M$. Thus, the same construction as for \eqref{eq:Wronskian-explicit-bound} in the proof of Theorem~\ref{thm:quantitative}, using \eqref{eq:estimate-u-tilde-stronger}, shows that for $|a|=M$,
\begin{align*}
&|\omega-m\upomega_+|^{1\pm 2s}\lp(\swei{\mathfrak{W}}{\pm s}\rp)^{-2} \\
&\quad\lesssim_{C_{\mc{A}'},M,|s|}|\omega-m\upomega_+|^{1/2}\lp|(f(x^*)\tilde{u})(+\infty)\rp|^2\\
&\quad\lesssim_{C_{\mc{A}'},M,|s|}\lp[\int_{2M}^\infty g(x)\lp|\tilde{H}\rp|^2dx\rp]\times\\
&\quad\qquad\times\lp[\int_{R_1}^{R_2} \sum_{k=0}^s \lp|\sum_{j=0}^{2s}\swei{b}{\pm s}_k\swei{c}{\pm s}_j (\omega,m,\lambda)r^{-k-j}+O_{\omega,m,\lambda}\lp(r^{-2s-1-k}\rp)\rp|dr^*\rp.\\
&\quad\qquad\qquad \lp.+\int_{R_1}^{R_2} |h|\sum_{j=0}^{2s} \lp|\swei{c}{+ s}_j (\omega,m,l)r^{-s-j}+O_{\omega,m,\lambda}\lp(r^{-3s-1}\rp)\rp|dr^*\rp]^{-2}\\
&\quad\lesssim_{C_{\mc{A}'},M,|s|} 1\,,\numberthis \label{eq:Wronskian-explicit-bound-strong}
\end{align*}
where $g(x)$, $\tilde{H}$ and $\swei{H}{-s}$, $R_1$, $R_2$, $\swei{b}{\pm s}_k$, $\swei{c}{\pm s}_j(\omega,m,\lambda)$ and $h$ are as in \eqref{eq:Wronskian-explicit-bound}. Combining with the uniform boundedness of the Wronskian in the subextremal case (Theorem~\ref{thm:quantitative}), one finally obtains Proposition~\ref{prop:quantitative-strong}.
\end{proof}

\printbibliography

\end{document}